\documentclass{amsart}

\usepackage{latexsym}
\usepackage{linearb}
\usepackage[LGR, OT1]{fontenc}
\usepackage{amsmath,amsthm}
\usepackage{amssymb}
\usepackage{upgreek}
\usepackage[greek,english]{babel}
\usepackage{teubner}
\usepackage{MnSymbol, wasysym}
\usepackage{graphicx}
\usepackage{marvosym}
\usepackage[margin=1.25in,dvips]{geometry}
\usepackage{color}
\usepackage{comment}
\usepackage{mathrsfs}
\usepackage{mathtools}

\usepackage{bbm}
\usepackage{bbold}

\DeclarePairedDelimiter{\ceil}{\lceil}{\rceil}
\DeclarePairedDelimiter\floor{\lfloor}{\rfloor}

\usepackage[toc,page]{appendix}

\usepackage{physics}

\usepackage{amsaddr}

\usepackage{xcolor}
 
\usepackage{hyperref}
\hypersetup{
  colorlinks = true,
  linkcolor  = black
}

\usepackage{slashed}

\newtheorem{definition}{Definition}[section]

\newtheorem{remark}{Remark}[section]
\newtheorem{lemma}{Lemma}[section]
\newtheorem{theorem}{Theorem}[section]
\newtheorem{proposition}{Proposition}[section]

\newtheorem*{rough version}{Rough Version}

\newtheorem*{theorem*}{Theorem}

\newtheorem*{corollary*}{Corollary}

\newenvironment{sketch proof}{\proof}{\endproof}

\newtheorem{thm}{Theorem}

\newtheorem{innercustomTheorem}{Theorem}
\newenvironment{customTheorem}[1]
  {\renewcommand\theinnercustomTheorem{#1}\innercustomTheorem}
  {\endinnercustomTheorem}
\makeindex

\numberwithin{equation}{section}

\title[Quasilinear wave equations]{Quasilinear wave equations on Schwarzschild--de~Sitter}

\author{Georgios Mavrogiannis}

\address[Georgios Mavrogiannis]{University of Cambridge, Department of Pure Mathematics and Mathematical Statistics, Wilberforce Road, Cambridge CB3 0WB, United Kingdom}
\email{gm615@cam.ac.uk}

\date\today

\begin{document}

\begin{abstract}
    We give an elementary new argument for global existence and exponential decay of solutions of quasilinear wave equations on Schwarzschild--de~Sitter black hole backgrounds, for appropriately small initial data. The core of the argument is entirely local, based on time translation invariant energy estimates in spacetime slabs of fixed time length. Global existence then follows simply by iterating this local result in consecutive spacetime slabs.  We infer that an appropriate future energy flux decays exponentially with respect to the energy flux of the initial data. 
\end{abstract}

\maketitle

\tableofcontents

\section{Introduction}

We revisit the problem of global existence of solutions of quasilinear wave equations, of the form
\begin{equation}\label{eq: quasilinear}
    \Box_{g(\nabla\psi)}\psi=\partial\psi\cdot\partial\psi,
\end{equation}
with $g(\nabla\psi)=g_{M,\Lambda}+h(\nabla\psi)$, where $g_{M,\Lambda}$ is the metric of the Schwarzschild--de~Sitter black hole spacetime, also $\partial\psi\cdot\partial\psi\:\dot{=}\:a^{ij}\partial_i\psi\partial_j\psi$, where $a,h$ are sufficiently regular tensors, with $h(0)=0$. Specifically, we are interested in a spacetime region that is slightly larger than that enclosed by the event and cosmological horizons, respectively $\mathcal{H}^+,\bar{\mathcal{H}}^+$, see the shaded region of Figure~\ref{fig: penrose}.

The problem of stability of quasilinear wave equations on such backgrounds has been extensively studied by Hintz and Vasy,~see~\cite{hintz3,hintz4,hintz5,hintz6}, where they arrived at global stability results. Their papers appeal to machinery from microlocal analysis and Nash Moser iteration arguments. These results were proceeded by a long list of results on the linear problem, see~\cite{vasy1,vasy2,bony,DR6,DR7,DR3,melr-barr-vasy,Dyatlov1,Dyatlov2}. Moreover, note the recent remarkable global non-linear stability proof for the slowly rotating Kerr--de~Sitter black hole as a solution of the Einstein vacuum equation with $\Lambda>0$, by Hintz--Vasy~\cite{hintz2}, based in part on the above works. (For some results on the cosmological region see~\cite{schlue2021decay,Volker}). Note that the non-linear stability of the pure de~Sitter spacetime has been obtained previously by Friedrich~\cite{friedrich}.

Schwarzschild--de~Sitter can be thought of as the $\Lambda>0$ analogue of the Schwarzschild and Kerr spacetimes, which are celebrated solutions of the vacuum Einstein equation with $\Lambda=0$. For the study of linear equations on the latter see for instance~\cite{DR4,tatarutohaneanuKerr} and the definitive~\cite{DR2}. These results have been used to prove non-linear stability results for equations of type~\eqref{eq: quasilinear}, see~\cite{luk2010null,lindbladquasilinear}. In general, these non-linear problems are more difficult than the $\Lambda>0$ case, because the expected decay is only polynomial and one has to assume and exploit suitable null structure (see~\cite{klainerman1984null}) for the non-linearities. For results on stability of black hole spacetimes with $\Lambda=0$ see \cite{DHR,rita1,hintzVasyHafner,anderssonBlueMa,klainermanszeftel,dafermos2021nonlinear,DR1}.

In principle, one approach to the study of~\eqref{eq: quasilinear} on Schwarzschild or Kerr--de~Sitter backgrounds would be to directly adapt the methods from the $\Lambda=0$ case. Such an approach, however, would not fully exploit the aspects that make the $\Lambda>0$ problem easier.

The purpose of this paper is to introduce a physical space approach to~\eqref{eq: quasilinear} which is well tailored to this setting. Our approach is based entirely on local in time translation invariant energy estimates (Theorem~\ref{thm: quasilinear, 1, local bootstrap}) and an iteration argument in consecutive spacetime regions (Theorem~\ref{thm: quasilinear, 2, exp decay}). This approach will use the results of our accompanying physical space linear paper~\cite{mavrogiannis}, where we utilized a physical space commutation with a vector field $\mathcal{G}$, see already~\eqref{eq: prototype new vector field}, and proved \textit{a relatively non-degenerate estimate}, using also a Morawetz estimate proved in~\cite{DR3}. Note that Holzegel--Kauffman originally introduced the analogue of the $\mathcal{G}$ vector field in the $\Lambda=0$ case, see~\cite{gustav}. Our physical space commutation with $\mathcal{G}$, in the high frequency limit, connects with the work of previous authors on `lossless estimates' and `non-trapping estimates', e.g. see~\cite{zworski,burq1,burq2,ikawa,nonnenmacher,hintz4,Dyatlov3}.

We will present the rough version of our Theorems, for which the reader may wish to refer to Figure~\ref{fig: penrose} for the Schwarzschild--de~Sitter spacetime.
\begin{figure}[htbp]
	\centering
	\includegraphics[scale=0.8]{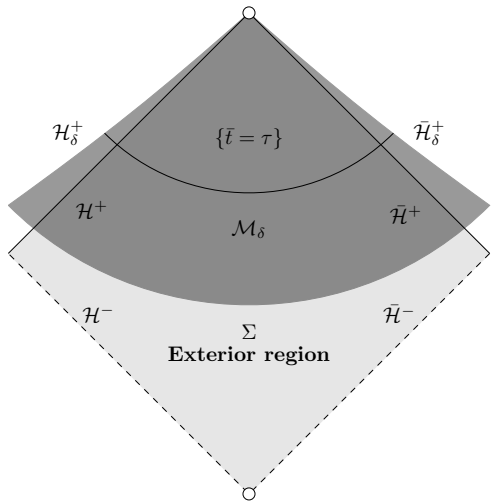}
	\caption{The Schwarzschild--de~Sitter spacetime}
	\label{fig: penrose}
\end{figure}

We denote as 
\begin{equation}
\mathcal{M}_\delta
\end{equation}
the `extended' exterior region, also see the dark shaded region of Figure~\ref{fig: penrose}, where $\delta$
is a smallness parameter that parametrises how far the boundaries of $\mathcal{M}_\delta$ (which we denote as $\mathcal{H}^+_\delta,\bar{\mathcal{H}}^+_\delta$) are from the event horizon $\mathcal{H}^+$ and the cosmological horizon $\bar{\mathcal{H}}^+$ respectively. We choose $\delta$ in the proof of our Theorem~\ref{thm: quasilinear, 2, exp decay} in Section~\ref{sec: proof of theorem 2, exp decay}. For the precise definition of $\mathcal{M}_\delta$ see already Definition \ref{def: spacetime domains}. The spacetime domain $\mathcal{M}_\delta$, see the dark shaded region of Figure~\ref{fig: penrose}, is foliated by the spacelike hypersurfaces
\begin{equation}
\{\bar{t}=\tau\},
\end{equation}
with $\tau\geq 0$, see Figure~\ref{fig: penrose}, where
\begin{equation}
	(\bar{t},r,\theta,\varphi)
\end{equation}
are appropriate (non-standard) hyperboloidal coordinates in which the metric takes the form~\eqref{eq: def: preliminairies, def 1, eq 1}. The coordinate vector field $\partial_{\bar{t}}$ is Killing. We denote by $\Omega_\alpha$, $\alpha=1,2,3$ the generators of the Lie algebra $so(3)$ associated with the $(\theta,\varphi)$ spheres.

We will consider two types of energies. 

The first energy is a non-degenerate energy which we define as 
\begin{equation}\label{eq: prototype EN energy}
	E_{j}[\psi](\tau)=  \sum_{1\leq i_1+i_2+i_3\leq j}\sum_{\alpha=1,2,3}\int_{\{\bar{t}=\tau\}}\left(\partial_{\bar{t}}^{i_1}\partial_r^{i_2}\Omega_\alpha^{i_3}\psi\right)^2
\end{equation}
for all $j\geq 1$, with respect to the induced volume form of $g_{M,\Lambda}$ on the $\{\bar{t}=\tau\}$ hypersurface.

Now, we discuss the second energy that we will consider. The first two terms of the second energy, see~\eqref{eq: prototype def: energies}, are associated with the $C^{0,1}$ vector field $\mathcal{G}$ of~\cite{mavrogiannis} (see the previous~\cite{gustav} for $\Lambda=0$), which is defined as 
\begin{equation}\label{eq: prototype new vector field}
\mathcal{G}\:\dot{=}\:
\begin{cases}
r\sqrt{1-\frac{2M}{r}-\frac{\Lambda}{3}r^2}\partial_r,&	\: 1-\frac{2M}{r}-\frac{\Lambda}{3}r^2>0\\
0 ,	&	\: 1-\frac{2M}{r}-\frac{\Lambda}{3}r^2\leq 0
\end{cases}
\end{equation}
in appropriate hyperboloidal coordinates $(\bar{t},r,\theta,\varphi)$, see already Definition \ref{def: new vector field}. We define the second energy as 
\begin{equation}\label{eq: prototype def: energies}
\begin{aligned}
E_{\mathcal{G},j}[\psi](\tau) &   =\int_{\{\bar{t}=\tau\}}\sum_{1\leq i_1+i_2+i_3\leq j-1,i_3\geq 1}\sum_{\alpha=1,2,3}\left(1-\frac{2M}{r}-\frac{\Lambda}{3}r^2\right)^{2i_3-1}\left(\partial_{\bar{t}}^{i_1}\Omega_\alpha^{i_2}\partial_r^{i_3}\mathcal{G}\psi\right)^2\\
&   \quad\quad\quad\quad\:\: +\sum_{1\leq i_1+i_2\leq j-1}\sum_{\alpha=1,2,3}(\partial_{\bar{t}}^{i_1}\Omega_\alpha^{i_2}\mathcal{G}\psi)^2\\
&	\quad\quad\quad\quad+ \sum_{1\leq i_1+i_2+i_3\leq j-1}\sum_{\alpha=1,2,3} \left(\partial_{\bar{t}}^{i_1}\partial_r^{i_2} \Omega_\alpha^{i_3} \psi\right)^2,  \\
\end{aligned}
\end{equation}
for all $j\geq 2$, with respect to the induced volume form of $g_{M,\Lambda}$ on the $\{\bar{t}=\tau\}$ hypersurface. Note that the highest order integrands of the energies in~\eqref{eq: prototype def: energies} are identically zero where $\mathcal{G}\equiv 0$. 

Note the inequality
\begin{equation}\label{eq: prototype def: energies, relation of the two energies}
E_{j-1}[\psi](\tau)\leq E_{\mathcal{G},j}[\psi](\tau)\lesssim E_{j}[\psi](\tau),
\end{equation}
also see already Remark \ref{rem: energies}. For the formal definition of the above energies see Definition \ref{def: energies}.

The rough version of our main Theorem~\ref{thm: quasilinear, 1, local bootstrap} is the following. (See the dark shaded region of Figure~\ref{fig: finite slab} for the local spacetime slab we will consider.)  
\begin{figure}[htbp]
	\centering
	\begin{minipage}[b]{0.4\textwidth}
		\includegraphics[width=\textwidth]{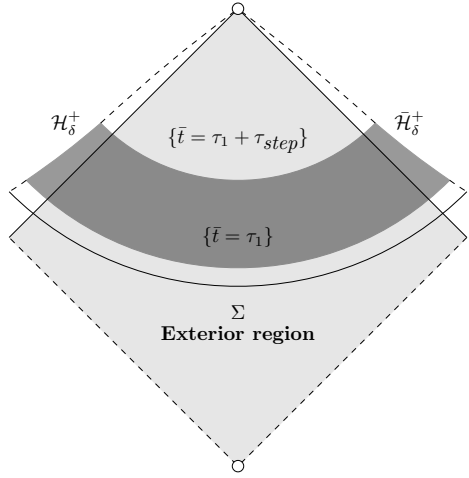}
		\caption{A spacetime slab of finite time length $\tau_{\textit{step}} $}
		\label{fig: finite slab}
	\end{minipage}
\end{figure}

\begin{customTheorem}{1}[rough version]\label{rough thm: finite slab}
Let $k\geq 7$. There exists a constant $C>0$, depending on $k,M,\Lambda$ and on the tensors $a,h$, see~\eqref{eq: quasilinear}, such that, for $L>0$ an arbitrary positive number, the following holds. 

There exists a $\tau_{\textit{step}}(L)>0$ sufficiently large and there exist
\begin{equation}
	\delta=\delta(\tau_{\textit{step}})>0,\qquad \epsilon=\epsilon(\tau_{\textit{step}},\delta)>0
\end{equation}
sufficiently small such that for all $\tau_1\geq 0$ and 
\begin{equation}
	\tau_2=\tau_1+\tau_{\textit{step}}
\end{equation}
if we take initial data for~\eqref{eq: quasilinear} on $\{\tau=\tau_1\}$ with
\begin{equation}
	E_{k+1}[\psi](\tau_1)\leq \epsilon
\end{equation}
then there exists a unique solution to the quasilinear wave equation~\eqref{eq: quasilinear} on $\mathcal{M}_\delta\cap \{\tau_1\leq \tau\leq \tau_2\}$ and the following estimates are satisfied 
\begin{equation}\label{eq: rough version thm 1, local bootstrap, eq 1}
	\begin{aligned}
		E_{\mathcal{G},k}[\psi](\tau_{2})	&	\leq \frac{1}{L}E_{\mathcal{G},k}[\psi](\tau_1),\\
	\end{aligned}
\end{equation}
and
\begin{equation}\label{eq: rough version thm 1, local bootstrap, eq 2}
	\begin{aligned}
		E_{\mathcal{G},k}[\psi](\tau^\prime)	&	\leq C E_{\mathcal{G},k}[\psi](\tau_1),\\
		E_{k+1}[\psi](\tau^\prime)	&	\leq C E_{k+1}[\psi](\tau_1),\\
		E_{k+2}[\psi](\tau^\prime)	&	\leq C E_{k+2}[\psi](\tau_1),\\
	\end{aligned}
\end{equation}  
for all $\tau^\prime\in [\tau_1,\tau_1+\tau_{\textit{step}}]$. 

Finally, 
\begin{equation}\label{eq: rough version thm 1, local bootstrap, eq 3}
\begin{aligned}
E_{k+1}[\psi](\tau_2)	\leq C e^{-\frac{1}{3}\log(L)+\frac{2}{3}\log C}\left(E_{\mathcal{G},k}[\psi](\tau_1)\right)^{1/3}\left(E_{k+2}[\psi](\tau_1)\right)^{2/3}.
\end{aligned}
\end{equation}
\end{customTheorem}

It is instructive to compare the results of Theorem~\ref{thm: quasilinear, 1, local bootstrap} with the main results of our linear theory~\cite{mavrogiannis} on a Schwarzschild--de~Sitter background. Specifically, the estimate of inequality~\eqref{eq: rough version thm 1, local bootstrap, eq 1} and the first inequality of~\eqref{eq: rough version thm 1, local bootstrap, eq 2} corresponds to the estimate in linear theory that one obtains from the commutation with the $\mathcal{G}$ vector field, see already Theorem~\ref{thm: wave equation with inhomogeneities}. The remaining inequalities of~\eqref{eq: rough version thm 1, local bootstrap, eq 2} correspond to the uniform boundedness results of~\cite{DR3}.

We use the result of Theorem~\ref{thm: quasilinear, 1, local bootstrap}, for a well chosen $L$ so that the constant in~\eqref{eq: rough version thm 1, local bootstrap, eq 3} is sufficiently small and, by a completely elementary iteration argument on consecutive spacetime regions, see Figure~\ref{fig: finite slab} for such a region, we prove that the solution of the quasilinear wave equation~\eqref{eq: quasilinear} exists globally and decays exponentially. The rough version of our main Theorem~\ref{thm: quasilinear, 2, exp decay} is the following.

\begin{customTheorem}{2}[rough version]\label{thm rough: exp decay}
Let $k\geq 7$. Then, there exist constants $c_{\text{d}},c_{\text{g}}>0$ and there exists a $\delta>0$ and an $\epsilon>0$ sufficiently small, such that if the initial energy satisfies
\begin{equation}
    E_{k+2}[\psi](0)\leq \epsilon,
\end{equation}
on $\Sigma=\{\bar{t}=0\}$ then the solution exists globally on $\mathcal{M}_\delta$ and the energy $E_{\mathcal{G},k}[\psi]$ decays exponentially
\begin{equation}\label{eq: thm rough: exp decay, eq 1}
    E_{\mathcal{G},k}[\psi](\tau)\lesssim e^{-c_{\text{d}}\tau} E_{\mathcal{G},k}[\psi](0).
\end{equation}

We note that the classical Sobolev energies $E_{k-1}[\psi](\tau),E_{k+1}[\psi](\tau)$ decay exponentially 
\begin{equation}\label{eq: thm rough: exp decay, eq 2}
    E_{k-1}[\psi](\tau)\lesssim e^{-c_{\text{d}}\tau} E_{\mathcal{G},k}[\psi](0),\qquad E_{k+1}[\psi](\tau)\lesssim  e^{-c_{\text{d}}\tau}E_{k+2}[\psi](0)
\end{equation}
while for the top order energy we only have
\begin{equation}\label{eq: thm rough: exp decay, eq 3}
	E_{k+2}[\psi](\tau)\lesssim e^{c_{\text{g}} \tau}E_{k+2}[\psi](0).
\end{equation}
\end{customTheorem}

\begin{remark}
Note that the growth constant~$c_g$ of~\eqref{eq: thm rough: exp decay, eq 3} of Theorem~\ref{thm: quasilinear, 2, exp decay} can in fact be made arbitrarily small, restricting to sufficiently small $\epsilon>0$. See already Remark~\ref{rem: sec: main theorems, rem 2}.
\end{remark}

\begin{remark}
In our Theorem~\ref{thm: quasilinear, 2, exp decay}, we improve slightly on the regularity assumption of initial data of~\cite{hintz6}. Specifically, we only require the initial data to be in the Sobolev space $H^9$, see our main Theorems \ref{thm: quasilinear, 1, local bootstrap}, \ref{thm: quasilinear, 2, exp decay} and Remark \ref{rem: sec: main theorems, rem 1}.
\end{remark}

\begin{remark}
	In Section \ref{sec: semilinear}, we present Theorems~\ref{thm: semilinear, local bootstrap} and \ref{thm: semilinear, exp decay} which give global non-linear stability in the semilinear case, with the tensor $h\equiv 0$, under weaker assumptions than those of Theorem~\ref{thm: quasilinear, 2, exp decay}. Specifically, we require the initial data only to be in the Sobolev space $H^7$.
\end{remark}

\begin{remark}
Note that by using the results of our forthcoming~\cite{mavrogiannis2} on Kerr--de~Sitter, the global stability results of the present paper generalize to the slowly rotating Kerr--de~Sitter case. 
\end{remark}

\subsection{Acknowledments} I would like to thank my supervisor Mihalis Dafermos, for his continuous support, for suggesting that the results of~\cite{mavrogiannis} may also treat the nonlinear stability problems of the present paper and for carefully reading previous versions of this paper. The author would also like to thank Christoph Kehle for valuable discussions and useful comments.

\section{Preliminaries and notation}\label{sec: preliminairies}

\subsection{The manifolds, metrics and spacetime domains}\label{subsec: manifold, metric, domains}

The definitions and notation of this Section have already been introduced in our~\cite{mavrogiannis}.

Fix $M,\Lambda>0$ such that  
\begin{equation}
    r_+<\bar{r}_+
\end{equation}
are the two positive real roots of 
\begin{equation}
    1-\mu\:\dot{=}\:1-\frac{2M}{r}-\frac{\Lambda}{3}r^2.
\end{equation}

We need the following definition

\begin{definition}\label{def: spacetime domains}
	For $\tau_2\geq\tau_1\geq 0$ we define 
	\begin{equation}
	D(\tau_1,\tau_2)=[\tau_1,\tau_2]_{\bar{t}}\times [r_+,\bar{r}_+]_r\times \mathbb{S}^2_{(\theta,\varphi)}.
	\end{equation}
	We also define, for $\delta>0$,
	\begin{equation}
	D_\delta(\tau_1,\tau_2)=[\tau_1,\tau_2]_{\bar{t}}\times[r_+-\delta,\bar{r}_+ +\delta]_r\times  \mathbb{S}^2_{(\theta,\varphi)}    
	\end{equation}
	and 
	\begin{equation}
		\mathcal{M}=D(0,\infty)\:\dot{=}\:[0,\infty)_{\bar{t}}\times [r_+,\bar{r}_+]_r\times \mathbb{S}^2_{(\theta,\varphi)},\quad\mathcal{M}_\delta=D_\delta(0,\infty)\:\dot{=}\:[0,\infty)_{\bar{t}}\times [r_+-\delta,\bar{r}_++\delta]_r\times \mathbb{S}^2_{(\theta,\varphi)}.
	\end{equation}
We refer to the coordinates $(\bar{t},r,\theta,\varphi)$ as regular hyperboloidal coordinates.
\end{definition}

We need the following definition 
\begin{definition}\label{def: preliminairies, def 1}
We denote as $(\mathcal{M}_\delta,g_{M,\Lambda})$ the Schwarzschild--de~Sitter spacetime, with metric
\begin{equation}\label{eq: def: preliminairies, def 1, eq 1}
    g_{M,\Lambda}=-\left(1-\frac{2M}{r}-\frac{\Lambda}{3}r^2\right)(d\bar{t})^2-2\frac{1-\frac{3M}{r}}{\sqrt{1-9M^2\Lambda}}\sqrt{1+\frac{6M}{r}} d\bar{t}dr +\frac{27M^2}{1-9M^2\Lambda}\frac{1}{r^2} (dr)^2+r^2 d\sigma_{\mathbb{S}^2},
\end{equation}
in regular hyperboloidal coordinates $(\bar{t},r,\theta,\varphi)$, where $d\sigma_{\mathbb{S}^2}=d\theta^2+\sin^2\theta d\varphi^2$ is the standard metric of the unit sphere $\mathbb{S}^2$. We will often denote the metric $g_{M,\Lambda}$ as 
\begin{equation}
	\mathring{g}.
\end{equation}
\end{definition}

Note the following remarks

\begin{remark}
	For a sufficiently small $\delta>0$, note that the hypersurfaces 
	\begin{equation}
	\{\bar{t}=\tau\}
	\end{equation}
	are spacelike in $\mathcal{M}_\delta$, with respect to the metric of the following Definition \ref{def: preliminairies, def 1}.
\end{remark}

\begin{remark}
	Note that $\delta$ will be fixed in the proof of Theorem~\ref{thm: quasilinear, 2, exp decay} in Section~\ref{sec: proof of theorem 2, exp decay}. In what follows $\delta$ will always be assumed sufficiently small.
\end{remark}

\begin{remark}\label{rem: sec: preliminairies, rem 0}
The inverse metric components of the metric of Definition \ref{def: preliminairies, def 1} are 
\begin{equation}
\begin{aligned}
\mathring{g}^{rr}=(1-\mu),\quad \mathring{g}^{\bar{t}\bar{t}}=-\frac{1-\xi^2(r)}{1-\mu},\quad \mathring{g}^{\bar{t}r}=-\xi(r).
\end{aligned}
\end{equation}
for 
\begin{equation}
	\xi(r)= \frac{1-\frac{3M}{r}}{\sqrt{1-9M^2\Lambda}}\sqrt{1+\frac{6M}{r}}.
\end{equation}
\end{remark}

\begin{remark}\label{rem: sec: preliminairies, rem 1}
Note the usual expression of the Schwarzschild--de~Sitter metric is in coordinates $(t,r,\theta,\varphi)\in\mathbb{R}_{t}\times (r_+,\bar{r}_+)_r\times \mathbb{S}^2_{(\theta,\varphi)}$, which reads
\begin{equation}
    \mathring{g}=-\left(1-\mu\right)(dt)^2+\left(1-\mu\right)^{-1}(dr)^2+r^2d\sigma_{\mathbb{S}^2},
\end{equation}
where note the relation 
\begin{equation}
	\bar{t}=t+H(r),\:\: H(r)=\int_{3M}^r\frac{\xi(r)}{1-\mu}dr.\:\: 
\end{equation}
\end{remark}

We denote as 
\begin{equation}
	\slashed{\nabla}
\end{equation}
the covariant derivative of the Riemannian metric $r^2d\sigma_{\mathbb{S}^2}$. Moreover, we denote the standard generators of the $so(3)$ Lie Algebra 
\begin{equation}
	\Omega_\alpha,\:\alpha=1,2,3,
\end{equation}
associated with the $(\theta,\varphi)$ spheres. These can be thought of as vector fields on $\mathcal{M}_\delta$.

Note that 
\begin{equation}
	|\slashed{\nabla}f|^2\sim\sum_\alpha \frac{1}{r^2}|\Omega_\alpha f|^2
\end{equation}
for any sufficiently regular $f$, where the constants in the above only depend on the black hole parameters. (We have included the inessential $r$ factor above for comparison with the asymptotically flat case.)

\subsection{The horizons and auxilliary spacelike hypersurfaces}\label{subsec: horizons, auxilliary hyper}

We need the following definition 
\begin{definition}\label{def: subsec horizons, horizons}
The event and cosmological horizons are defined as
\begin{equation}\label{eq: def: subsec horizons, horizons, eq 1}
    \begin{aligned}
        \mathcal{H}^+=\{r=r_+\},\:\bar{\mathcal{H}}^+=\{r=\bar{r}_+\}
    \end{aligned}
\end{equation}
also see~\cite{mavrogiannis,mavrogiannis2}. These hypersurfaces are null with respect to the metric $g_{M,\Lambda}$.

For sufficiently small $\delta>0$, we define the following spacelike hypersurfaces 
\begin{equation}\label{eq: def: subsec horizons, horizons, eq 2}
    \mathcal{H}^+_{\delta}=\{r=r_+-\delta,\:\bar{t}\geq 0\},\quad \bar{\mathcal{H}}^+_\delta=\{r=\bar{r}_++\delta,\:\bar{t}\geq 0\},
\end{equation}
in the coordinates $(\bar{t},r,\theta,\varphi)$ of Definition \ref{def: preliminairies, def 1}. 
\end{definition}

\subsection{Volume forms and normals for the metric \texorpdfstring{$g_{M,\Lambda}$}{g}}\label{subsec: volume forms}

We denote the spacetime volume form of $\mathring{g}$ as  
\begin{equation}\label{eq: hyperb volume form}
    d\mathring{g}=r^2\sin\theta d\bar{t} dr d\theta d\varphi
\end{equation}
with respect to the $\left(\bar{t},r,\theta,\varphi\right)$ coordinates.

By pulling back the spacetime volume form~\eqref{eq: hyperb volume form} into hypersurfaces of constant $\bar{t}$, we obtain that the $\{\bar{t}=\tau\}$ hypersurfaces admit the volume form 
\begin{equation}\label{eq: volume form of spacelike hypersurface}
d\mathring{g}_{\{\bar{t}=c\}}=r\sqrt{\frac{27M^2}{1-9M^2\Lambda}}\sin\theta dr d\theta d\varphi.    
\end{equation}

We note that the normal of $\{\bar{t}=\tau\}$, with respect to $\mathring{g}$, is
\begin{equation}
	\mathring{n}_{\{\bar{t}=\tau\}}=\sqrt{\frac{27M^2}{1-9M^2\Lambda}\frac{1}{r^2}}\frac{\partial}{\partial \bar{t}}+\frac{\xi(r)}{\sqrt{\frac{27M^2}{1-9M^2\Lambda}\frac{1}{r^2}}}\frac{\partial}{\partial r},
\end{equation}
which we also simply denote as 
\begin{equation}
	\mathring{n},
\end{equation}
for $\xi(r)$ see Remark \ref{rem: sec: preliminairies, rem 0}. We denote the normals of the event and cosmological horizons respectively as 
\begin{equation}
    \mathring{n}_{\mathcal{H}^+}=\partial_{\bar{t}},\:\mathring{n}_{\bar{\mathcal{H}}^+}=\partial_{\bar{t}}.
\end{equation}
With the above choice of normals, the corresponding volume forms of the respective null hypersurfaces take the form
\begin{equation} 
d\mathring{g}_{\mathcal{H}^+}=r^2\sin\theta d\bar{t}d\sigma_{\mathbb{S}^2},\:\: d\mathring{g}_{\bar{\mathcal{H}}^+}=r^2\sin\theta d\bar{t}d\sigma_{\mathbb{S}^2}.
\end{equation}

Furthermore, for $\delta>0$ sufficiently small, the vectors 
\begin{equation}
	\mathring{n}_{\mathcal{H}^+_\delta}=\frac{\nabla r}{\sqrt{|\mathring{g}(\nabla r,\nabla r)}|}\Bigg|_{r=r_+-\delta},\:\mathring{n}_{\bar{\mathcal{H}}^+_\delta}=-\frac{\nabla r}{\sqrt{|\mathring{g}(\nabla r,\nabla r)}|}\Bigg|_{r=\bar{r}_++\delta}
\end{equation}
are the unit outward normals of the spacelike hypersurfaces $\mathcal{H}^+_\delta,\bar{\mathcal{H}}^+_\delta$, respectively, with respect to the metric $\mathring{g}$. Note that there exist smooth functions 
\begin{equation}
	c_1:(r_+-\delta,r_+)\rightarrow \mathbb{R},\qquad \bar{c}_1:(\bar{r}_+,\bar{r}_++\delta)\rightarrow \mathbb{R},
\end{equation}
with
\begin{equation}\label{eq: subsec: volume forms, eq 1}
\sqrt{|1-\mu|}c_1(r)= 1+\mathcal{O}(1-\mu),\qquad \sqrt{|1-\mu|}\bar{c}_1(r)=1+\mathcal{O}(1-\mu),\: 
\end{equation}
as $r\rightarrow r_+,\:r\rightarrow \bar{r}_+$ respectively, and $c_2(r)=-\sqrt{|1-\mu|},\: \bar{c}_2(r)=\sqrt{|1-\mu|}$, such that the normals of $\mathcal{H}^+_\delta,\bar{\mathcal{H}}^+_\delta$ can be written respectively as 
\begin{equation}
\mathring{n}_{\mathcal{H}^+_\delta}=c_1(r)\partial_{\bar{t}}+c_2(r)\partial_r,\:\:\mathring{n}_{\bar{\mathcal{H}}^+_\delta}=\bar{c}_1(r)\partial_{\bar{t}}+\bar{c}_2(r)\partial_r.
\end{equation}

By pulling back the spacetime volume form~\eqref{eq: hyperb volume form} on hypersurfaces of constant $r=r_+-\delta,\text{ or }r=\bar{r}_++\delta$,  we obtain the respective volume forms 
\begin{equation} 
d\mathring{g}_{\mathcal{H}^+_\delta},\:\: d\mathring{g}_{\bar{\mathcal{H}}^+_\delta}=r^2\sqrt{|1-\mu|}\sin\theta d\bar{t}d\sigma_{\mathbb{S}^2},
\end{equation}
which we both denote simply as $d\mathring{g}_{\mathcal{H}_\delta}$.

\subsection{Coarea formula}\label{subsec: coarea}

There exists a constant $C(M,\Lambda)>0$ such that 
\begin{equation}\label{eq: subsec: coarea, eq 1}
	C^{-1} \int_{\tau_1}^{\tau_2}d\tau \int_{\{\bar{t}=\tau\}} fd\mathring{g}_{\{\bar{t}=\tau\}}\leq \int\int_{D_\delta(\tau_1,\tau_2)} f d\mathring{g}\leq C\int_{\tau_1}^{\tau_2}d\tau \int_{\{\bar{t}=\tau\}} fd\mathring{g}_{\{\bar{t}=\tau\}}
\end{equation}
for any continuous non-negative function $f$.

\subsection{The vector fields and the energies}\label{subsec: energies}

We need the following 
\begin{definition}\label{def: new vector field}
We define the following vector field
\begin{equation}\label{eq: new vector field}
    \mathcal{G}\:\dot{=}\:r\sqrt{1-\frac{2M}{r}-\frac{\Lambda}{3}r^2}\frac{\partial}{\partial r},
\end{equation}
for $r\in[r_+,\bar{r}_+]$ with respect to the hyperboloidal coordinates $(\bar{t},r,\theta,\varphi)$.

We extend the vector field~\eqref{eq: new vector field} to a $C^{0,1}$ vector field on $\mathcal{M}_\delta$ by 
\begin{equation}
	\mathcal{G}\equiv 0\text{ in }	\{r_+-\delta\leq r\leq r_+\}\cup \{\bar{r}_+\leq r\leq\bar{r}_++\delta\}.
\end{equation}  
\end{definition}

We need the following definition

\begin{definition}\label{def: energies}
On the spacelike hypersurface $\{\bar{t}=\tau\}\subset D_\delta(0,\infty)$, with respect to the metric $\mathring{g}$, we define the non-degenerate high order energy
\begin{equation}\label{eq: def: energies, high energy, non degenerate}
E_{j}[\psi](\tau)=  \sum_{1\leq i_1+i_2+i_3\leq j}\sum_{\alpha=1,2,3}\int_{\{\bar{t}=\tau\}}\left(\partial_{\bar{t}}^{i_1}\partial_r^{i_2}\Omega_\alpha^{i_3}\psi\right)^2d\mathring{g}_{\{\bar{t}=\tau\}},
\end{equation}
with $j\geq 1$, where for $\Omega^{\alpha}$ see Section \ref{subsec: manifold, metric, domains}.

We define the following high order energy
\begin{equation}\label{eq: def: energies, high energy, G}
    \begin{aligned}
        E_{\mathcal{G},j}[\psi](\tau) &   =\int_{\{\bar{t}=\tau\}}\Bigg(\sum_{1\leq i_1+i_2+i_3\leq j-1,i_3\geq 1}\sum_{\alpha=1,2,3}(1-\mu)^{2i_3-1}\left(\partial_{\bar{t}}^{i_1}\Omega_\alpha^{i_2}\partial_r^{i_3}\mathcal{G}\psi\right)^2\\
        &   \quad\quad\quad\quad\quad +\sum_{1\leq i_1+i_2\leq j-1}\sum_{\alpha=1,2,3}(\partial_{\bar{t}}^{i_1}\Omega_\alpha^{i_2}\mathcal{G}\psi)^2\\
        &	\quad\quad\quad\quad\quad +\sum_{1\leq i_1+i_2+i_3\leq j-1}\sum_{\alpha=1,2,3} \left(\partial_{\bar{t}}^{i_1}\partial_r^{i_2} \Omega_\alpha^{i_3} \psi\right)^2\Bigg)d\mathring{g}_{\{\bar{t}=\tau\}}, \\
    \end{aligned}
\end{equation}
for $j\geq 2$. 
\end{definition}

\begin{remark}\label{rem: energies}
For any $j\geq 2$, there exists a constant $C(j,M,\Lambda)>0$, such that for all $\tau\geq 0$ the following holds 
\begin{equation}\label{eq: def: energies, trivial inequality, 2}
    E_{j-1}[\psi](\tau)\leq  E_{\mathcal{G},j}[\psi](\tau)\leq C E_{j}[\psi](\tau).
\end{equation}  
\end{remark}

Moreover, note the interpolation statement
\begin{lemma}\label{lem: subsec: energies, lem 1}
Let $k\geq 0$. Then, there exist a constant 
\begin{equation}
	C_{\textit{int}}(k,M,\Lambda)>0
\end{equation}
such that for $\psi$ a sufficiently regular function we obtain 
\begin{equation}
	E_{k+1}[\psi](\tau)	\leq C_{\textit{int}} \left(E_{k-1}[\psi](\tau)\right)^{1/3}\left(E_{k+2}[\psi](\tau)\right)^{2/3},
\end{equation} 
for any $\tau\geq 0$.
\end{lemma}

\subsection{Notation for derivatives}\label{subsec: derivatives notation}

We need the following notations 
\begin{definition}\label{def: subsec: derivatives notation, def 1}
Let $X$ be a Lipschitz vector field on the manifold $\mathcal{M}_\delta$ and let $\psi$ be a smooth function on the manifold $\mathcal{M}_\delta$. Then, for all $j\geq 1$ we define
\begin{equation}\label{eq: def: subsec: derivatives notation, def 1, eq 1}
	\begin{aligned}
		|\partial\psi|	&	=|\partial_{\bar{t}}\psi|+|\partial_r\psi|+\sum_{\alpha=1,2,3}|\Omega_\alpha\psi|,\\
		(\partial^j\psi)^2	&	= \sum_{1\leq j_1+j_2+j_3\leq j}\sum_\alpha \left(\partial_{\bar{t}}^{j_1}\partial_r^{j_2}\Omega_\alpha^{j_3}\psi\right)^2,\\
		(X\partial^j\psi)^2	&	= \sum_{1\leq j_1+j_2+j_3\leq j}\sum_\alpha \left(X\partial_{\bar{t}}^{j_1}\partial_r^{j_2}\Omega_\alpha^{j_3}\psi\right)^2,\\
		(\partial^jX\psi)^2	&	= \sum_{1\leq j_1+j_2+j_3\leq j}\sum_\alpha \left(\partial_{\bar{t}}^{j_1}\partial_r^{j_2}\Omega_\alpha^{j_3}X\psi\right)^2.
	\end{aligned}
\end{equation}
Note that if $X$ is only Lipschitz (for instance $X=\mathcal{G}$) then the last expression of~\eqref{eq: def: subsec: derivatives notation, def 1, eq 1} may not necessarily be finite.
\end{definition}

It will be convenient to compare the above expressions with expressions of coordinate derivatives in ambient globally defined Cartesian coordinates. We define the map
\begin{equation}\label{eq: cartesian manifold}
	\begin{aligned}
		\mathcal{M}_\delta	&	\rightarrow \mathbb{R}^4\\
		(\bar{t},r,\theta,\varphi)	&\mapsto (x^0,x^1,x^2,x^3)
	\end{aligned}
\end{equation}
to be a change of coordinates from the coordinates ascribed to the manifold $\mathcal{M}_\delta$, see Definition \ref{def: spacetime domains}, to Cartesian coordinates, where
\begin{equation}
	x^0=\bar{t},\quad x^1 = r\sin\theta\cos\varphi,\quad x^2=r\sin\theta\sin\varphi,\quad x^3=r\cos\theta.
\end{equation}

Define the set 
\begin{equation}\label{eq: subsec: derivatives notation, eq 3}
\text{Cart}_m=\{\underbrace{\partial_{x^1}\partial_{x^2}\dots\partial_{x^1}}_m, \partial_{x^1}\partial_{x^2}\dots\partial_{x^0},\dots\}
\end{equation}
to have as elements all the operators that are comprised of any collection of $m$ derivatives of the coordinate vector fields of the Cartesian coordinates $(x^0,x^1,x^2,x^3)$, see~\eqref{eq: cartesian manifold}.

For all $j\geq 1$, we note the similarities 
\begin{equation}\label{eq: subsec: derivatives notation, eq 2}
\begin{aligned}
|\partial\psi| &	\sim |\partial_{x^0}\psi|+|\partial_{x^1}\psi|+|\partial_{x^2}\psi|+|\partial_{x^3}\psi|,\\
(\partial^j\psi)^2	&	\sim \sum_{1\leq i\leq j}\quad\sum_{\text{all } D^i\in\text{Cart}_i} \left(D^j\psi\right)^2.\\
\end{aligned}
\end{equation}

Finally, for an arbitrary smooth Lorentzian metric $g$, the wave operator is 
\begin{equation}
	\Box_g= g^{ab}\partial_a\partial_b+\Gamma^c\:_{ab}(g)\partial_c,
\end{equation} 
where $\Gamma^c\:_{ab}(g)$ are the Christoffel symbols, of the Levi-Civita connection, with respect to the metric~$g$. 

\subsection{The smooth tensor $a$}\label{subsec: semilinear derivatives}

We fix a smooth tensor 
\begin{equation}
	a:T\mathcal{M}_\delta\times T\mathcal{M}_\delta\rightarrow\mathbb{R}.
\end{equation}
Note that for any $m\in\mathbb{N}\cup\{0\}$ there exist constants $A_m<\infty$, such that if $a^{ij}$ are the components of its inverse in Cartesian coordinates~\eqref{eq: cartesian manifold} and $D^m$ any element of the set $\text{Cart}_m$, see~\eqref{eq: subsec: derivatives notation, eq 3}, then
\begin{equation}
	|D^m a^{ij}|\leq A_m
\end{equation}
for all $m,i,j$.

In what follows, if a constant $C$ depends on $A_0,A_1,\dots, A_m$ we denote it as 
\begin{equation}
	C=C(A_{[m]}).
\end{equation}

\subsection{Sobolev inequality}

Note the following Sobolev inequality. 

\begin{lemma}\label{lem: sobolev estimate}
There exists a constant $C=C(M,\Lambda)>0$ such that for $f\in C^\infty(\{\bar{t}=c\})$, the following holds 
\begin{equation}
    \|f\|^2_{L^\infty (\{\bar{t}=c\})}\leq C(M,\Lambda) \sum_{1\leq i_1+i_2\leq 2}\sum_{\alpha=1,2,3}\int_{\{\bar{t}=c\}}\left(\partial_r^{i_1}\Omega_\alpha^{i_2}f\right)^2 d\mathring{g}_{\{\bar{t}=c\}}.
\end{equation}
\end{lemma}

\subsection{The $\mathcal{G}$ vector field commutation estimate}\label{sec: G linear theory}

In our~\cite{mavrogiannis}, we proved the following on a Schwarzschild--de~Sitter background. 

\begin{theorem}{(Theorem 3 of~\cite{mavrogiannis})}\label{thm: wave equation with inhomogeneities}
	Let $\psi$ satisfy the inhomogeneous wave equation 
	\begin{equation}\label{eq: inhomogeneous equation}
	\Box_{g_{M,\Lambda}}\psi =F 
	\end{equation}
	on $D_\delta(\tau_1,\tau_2)$. Then, for any $k\geq 3$ there exists a constant $C(k,M,\Lambda)>0$, such that the following estimate holds
	\begin{equation}\label{eq: energy estimate of inhomogeneuos equation}
	\begin{aligned}
	&	E_{\mathcal{G},k}[\psi](\tau_2)+\int_{\tau_1}^{\tau_2}d\tau E_{\mathcal{G},k}[\psi](\tau)	\\
	&	\quad\leq C E_{\mathcal{G},k}[\psi](\tau_1)\\
	&	\quad\quad+C\int\int_{D_\delta(\tau_1,\tau_2)}d\mathring{g} \sum_{0\leq i_1+i_2+i_3\leq k-2}\sum_\alpha\left(\partial_{\bar{t}}^{i_1}\partial_r^{i_2}\left(\Omega_\alpha\right)^{i_3} F\right)^2+(1-\mu)\sum_{0\leq i_1+i_3\leq k-2}\sum_\alpha\left(\partial_{\bar{t}}^{i_1} (\Omega_\alpha)^{i_3}\mathcal{G}F\right)^2\\
	&	\quad\quad +C \int_{\{\bar{t}=\tau_2\}}d\mathring{g}_{\{\bar{t}=\tau\}} \sum_{0\leq i_1+i_2+i_3\leq k-3}\sum_\alpha\left(\partial_{\bar{t}}^{i_1}\Omega_\alpha^{i_2}\partial_r^{i_3}F\right)^2+(1-\mu)^{2i_3+1}\left(\partial_{\bar{t}}^{i_1}\Omega_\alpha^{i_2}\partial_r^{i_3}\mathcal{G}F\right)^2
	\end{aligned}
	\end{equation}
where the volume forms are with respect to the Schwarzschild--de~Sitter metric $g_{M,\Lambda}$, see Section~\ref{subsec: volume forms}. 
\end{theorem}
\begin{proof}
	This is proved in~\cite{mavrogiannis}. 
\end{proof}

\begin{remark}
Note that the highest order term of the error hypersurface terms on the right hand side of~\eqref{eq: energy estimate of inhomogeneuos equation} is of order $k-2$, while the highest order term of the error bulk terms on the right hand side of~\eqref{eq: energy estimate of inhomogeneuos equation} is of order $k-1$. The weights in $(1-\mu)$ on the error terms on the right hand side of~\eqref{eq: energy estimate of inhomogeneuos equation} ensure that terms on the right hand side of~\eqref{eq: energy estimate of inhomogeneuos equation} are regular for a sufficiently regular function $F$.
\end{remark}

\section{Metric close to Schwarzschild--de~Sitter and the Local well-posedness result}\label{sec: local well posedness}

\subsection{Metric close to Schwarzschild--de~Sitter}\label{subsec: assumption}

We define the class of metrics close to Schwarzschild--de~Sitter.

We fix a sufficiently regular tensor 
\begin{equation}
	h \in \Gamma\left(T\mathcal{M}_\delta^\star\times T\mathcal{M}_\delta^\star\times T\mathcal{M}_\delta^\star\right)
\end{equation}
with 
\begin{equation}
h(\cdot,\cdot,v=0)=0,
\end{equation}
and define 
\begin{equation}
	g(v)=\mathring{g}+h(v).
\end{equation}

We also define $h^{ij}(v)\:\dot{=}\:g^{ij}(v)-\mathring{g}^{ij}$, where $g^{ij}(v),\mathring{g}^{ij}$ are the components of the inverses of the respective tensors in Cartesian coordinates.

Note that for every $m\in\mathbb{N}\cup\{0\}$ there exist constants $B_m<\infty$, such that, for $m\geq 0$, the following hold in Cartesian coordinates
\begin{equation}
\begin{aligned}
|D^m h_{ijk}|,\quad |D^m h^{ij}\:_k|\leq B_m,
\end{aligned}
\end{equation} 
for all $m,i,j,k$, where $h^{ij}\:_k \:v^k=h^{ij}(v)$ and $D^m$ is any element of the set $\text{Cart}_m$, see~\eqref{eq: subsec: derivatives notation, eq 3}.

In what follows, if a constant $C$ depends on $B_0,B_1,\dots,B_m$, we denote it as 
\begin{equation}
	C=C(B_{[m]}).
\end{equation}

\subsection{The Cauchy stability type result}

We present a Cauchy stability type result for the quasilinear wave equation, in the form that we will use, appropriately tailored to display the energies that are used in our main global stability results, see already Theorem \ref{thm: quasilinear, 1, local bootstrap} and Theorem \ref{thm: quasilinear, 2, exp decay}.

\begin{proposition}\label{prop: well posedness quasilinear}
Let $k\geq 7$ and let the tensors $a,h$ be as in Sections \ref{subsec: semilinear derivatives}, \ref{subsec: assumption} respectively. There exists a constant
\begin{equation}
C_{\textit{wp}}(k,M,\Lambda,A_{[k+1]},B_{[k+1]})>1,
\end{equation}
where for $A_{[k+1]},B_{[k+1]}$ see Sections \ref{subsec: semilinear derivatives}, \ref{subsec: assumption} respectively, such that the following holds.

Let $\tau_{\textit{max}}>0$ be an arbitrary positive number. There exist 
\begin{equation}
	\delta(\tau_{\textit{max}},k,M,\Lambda,A_{[k+1]},B_{[k+1]}),\qquad\epsilon(\tau_\textit{max},\delta,k,M,\Lambda,A_{[k+1]},B_{[k+1]})>0
\end{equation}
both sufficiently small, such that if we take initial data for~\eqref{eq: quasilinear} on $\{\bar{t}=\tau_1\}$ with  
\begin{equation}\label{eq: prop: well posedness quasilinear, eq 0}
    E_{k+1}[\psi](\tau_1)\leq \epsilon,
\end{equation}
for some $\tau_1\geq 0$, then there exists a unique $H^{k+1}$ solution $\psi$ to the quasilinear wave equation~\eqref{eq: quasilinear}, on $D_\delta(\tau_1,\tau_1+\tau_{\textit{max}})$, such that 
\begin{equation}\label{eq: prop: well posedness quasilinear, eq 1}
	\begin{aligned}
		E_{k+1}[\psi](\tau^\prime)	&	\leq C_{\textit{wp}} E_{k+1}[\psi](\tau_1), \\
	\end{aligned}
\end{equation}
\begin{equation}\label{eq: prop: well posedness quasilinear, eq 2}
\begin{aligned}
E_{k+2}[\psi](\tau^\prime)	&	\leq C_{\textit{wp}} E_{k+2}[\psi](\tau_1),\\
\end{aligned}
\end{equation}
for all $\tau^\prime\in [\tau_1,\tau_1+\tau_{\textit{max}}]$, where~\eqref{eq: prop: well posedness quasilinear, eq 2} holds if $E_{k+2}[\psi](\tau_1)<\infty$, in which case the solution is $H^{k+2}$.
\end{proposition}

\begin{proof}
See the appendix \ref{appendix: local well posedness}.
\end{proof}

\begin{remark}\label{rem: local well posedness}
The result of Proposition \ref{prop: well posedness quasilinear} is a refinement of the usual Cauchy stability and thus we will have to prove it explicitly. The independence of $C_{\textit{wp}}$ on $\tau_{\text{max}}$ is connected to the uniform boundedness result \cite{DR3} for the linear wave equation.
\end{remark}

\section{The main Theorems}

\subsection{Theorem on an arbitrary slab of length $\tau_{\textit{step}}$.}

We give the detailed statement of the Theorem, on a fixed large time domain

\begin{customTheorem}{1}\label{thm: quasilinear, 1, local bootstrap}
Let $k\geq 7$ and let the tensors $a,h$ be as in Sections \ref{subsec: semilinear derivatives}, \ref{subsec: assumption} respectively. There exists a constant
\begin{equation}
	C(k,M,\Lambda,A_{[k+1]},B_{[k+1]})>0
\end{equation}
where for $A_{[k+1]},B_{[k+1]}$ see Sections \ref{subsec: semilinear derivatives}, \ref{subsec: assumption} respectively, such that the following holds.

For all $L>0$ there exists a $\tau_{\textit{step}}(L,k,M,\Lambda,A_{[k+1]},B_{[k+1]},C_{\textit{int}},C_{\textit{wp}})$ sufficiently large, where for $C_{\textit{int}},C_{\textit{wp}}$ see respectively Lemma~\ref{lem: subsec: energies, lem 1} and Proposition~\ref{prop: well posedness quasilinear}, such that there exist 
\begin{equation}
	\delta(\tau_{\textit{step}},k,M,\Lambda,A_{[k+1]},B_{[k+1]}),\qquad \epsilon(\tau_{\textit{step}},\delta,k,M,\Lambda,A_{[k+1]},B_{[k+1]})>0
\end{equation}
sufficiently small such that if we take initial data for~\eqref{eq: quasilinear} in $\{\tau=\tau_1\}$ with
\begin{equation}\label{eq: corollary quasilinear, eq 1, -1}
    E_{k+1}[\psi](\tau_1)\leq \epsilon,\qquad E_{k+2}[\psi](\tau_1)<\infty
\end{equation}
and 
\begin{equation}
    \tau_2=\tau_1+\tau_{\textit{step}},
\end{equation}
then there exists a unique $H^{k+2}$ solution in $D_\delta(\tau_1,\tau_2)$ to the quasilinear wave equation~\eqref{eq: quasilinear}, and for all $\tau^\prime\in[\tau_1,\tau_2]$ the following inequalities hold
\begin{equation}\label{eq: corollary quasilinear, eq 1}
    \begin{aligned}
        E_{\mathcal{G},k}[\psi](\tau^\prime)    &   \leq C  E_{\mathcal{G},k}[\psi](\tau_1),\\
    \end{aligned}
\end{equation}
\begin{equation}\label{eq: corollary quasilinear, eq 1.1}
	E_{k+1}[\psi](\tau^\prime)	\leq C_{\textit{wp}} E_{k+1}[\psi](\tau_1),\\
\end{equation}
\begin{equation}\label{eq: corollary quasilinear, eq 1.2}
E_{k+2}[\psi](\tau^\prime)		\leq C_{\textit{wp}} E_{k+2}[\psi](\tau_1).\\
\end{equation}

Moreover, the following holds
\begin{equation}\label{eq: corollary quasilinear, eq 2}
	\begin{aligned}
	 	E_{\mathcal{G},k}[\psi](\tau_{2}) & \leq  \frac{1}{L}E_{\mathcal{G},k}[\psi](\tau_1).\\
	\end{aligned}
\end{equation}

Finally, the following holds 
\begin{equation}\label{eq: corollary quasilinear, eq 3}
\begin{aligned}
E_{k+1}[\psi](\tau_2)	\leq C_{\textit{int}} e^{-\frac{1}{3}\log(L)+\frac{2}{3}\log C_{\textit{wp}}}\left(E_{\mathcal{G},k}[\psi](\tau_1)\right)^{1/3}\left(E_{k+2}[\psi](\tau_1)\right)^{2/3}.
\end{aligned}
\end{equation}
\end{customTheorem}

\begin{remark}\label{rem: sec: main theorems, rem 1}
For the requirement $k\geq 7$, see already the computations of inequality~\eqref{eq: energy estimate for quasilinear, 2, bound} and Proposition \ref{prop: well posedness quasilinear}. 
\end{remark}

\subsection{The global nonlinear stability of the quasilinear wave equation \texorpdfstring{$\eqref{eq: quasilinear}$}{q}}

Now, we present the main Theorem of our paper, which is in fact a Corollary of Theorem \ref{thm: quasilinear, 1, local bootstrap}.
\begin{customTheorem}{2}\label{thm: quasilinear, 2, exp decay}
Let $k\geq 7$, let the tensors $a,h$ be as in Sections \ref{subsec: semilinear derivatives}, \ref{subsec: assumption} respectively. Then, there exist constants
\begin{equation}
	c_{\text{d}}(k,M,\Lambda,A_{[k+1]},B_{[k+1]}),\qquad c_{\text{g}}(k,M,\Lambda,A_{[k+1]},B_{[k+1]}), \qquad C(k,M,\Lambda,A_{[k+1]},B_{[k+1]})>0
\end{equation}
and there exist
\begin{equation}
	\delta=\delta(k,M,\Lambda,A_{[k+1]},B_{[k+1]})>0,\qquad\epsilon=\epsilon(\delta,k,M,\Lambda,A_{[k+1]},B_{[k+1]})>0
\end{equation}
sufficiently small, such that for 
\begin{equation}
    E_{k+2}[\psi](0)\leq \epsilon,
\end{equation}
on $\Sigma=\{\bar{t}=0\}$, the arising $H^{k+2}$ solution of the quasilinear wave equation~\eqref{eq: quasilinear} exists globally on $D_\delta(0,\infty)$ and satisfies the following exponential decay 
\begin{equation}\label{eq: thm 2, exp decay, eq 1}
    E_{\mathcal{G},k}[\psi](\tau)\leq C e^{-c_{\text{d}}\tau} E_{\mathcal{G},k}[\psi](0),
\end{equation}
for any $\tau\geq 0$.

Moreover, we have the following exponential decay of the lower order energies 
\begin{equation}\label{eq: thm 2, exp decay, eq 2}
	E_{k-1}[\psi](\tau)\leq C e^{-c_{\text{d}}\tau}E_{\mathcal{G},k}[\psi](0),\quad E_{k+1}[\psi](\tau)\leq C e^{-c_{\text{d}}\tau}E_{k+2}[\psi](0),
\end{equation}
while for the top order energy we have the growth estimate
\begin{equation}\label{eq: thm 2, exp decay, eq 3}
	\begin{aligned}
		E_{k+2}[\psi](\tau)	&	\leq C e^{c_{\text{g}}\tau}E_{k+2}[\psi](0).\\
\end{aligned}
\end{equation}
\end{customTheorem}

\begin{remark}\label{rem: sec: main theorems, rem 2}
	Note that the growth constant $c_g$ of inequality~\eqref{eq: thm 2, exp decay, eq 3} of Theorem~\ref{thm: quasilinear, 2, exp decay} can be made arbitrarily small, if $\delta>0$ and $\epsilon>0$ are sufficiently small. See already the proof of Theorem~\ref{thm: quasilinear, 2, exp decay} in Section~\ref{sec: proof of theorem 2, exp decay}.  Although~\eqref{eq: thm 2, exp decay, eq 3} allows the top order energy $E_{k+2}[\psi]$ to a priori grow exponentially, one can in fact prove stronger results. Specifically, with the additional use of a top order uniform boundedness estimate, one can prove Theorem \ref{thm: quasilinear, 2, exp decay} assuming only small $E_{k+1}[\psi](0)$ initial data, and also uniformly bound the energy $E_{k+1}[\psi](\tau)$ for all times. (Therefore, we expect to only need the initial energy to lie in the Sobolev space $H^8$.) We will not however pursue this here as the weaker estimate provided by Proposition~\ref{prop: well posedness quasilinear} leading to~\eqref{eq: thm 2, exp decay, eq 3} is sufficient. This improvement is however important in the asymptotically flat case.    
\end{remark}

\subsection{Extension of our results on Kerr--de~Sitter}
We can extend our results for the quasilinear wave equation to the slowly rotating Kerr--de~Sitter case by utilizing the fixed frequency vector field commutation that we introduce in our forthcoming~\cite{mavrogiannis2}. This will be presented in a future paper.

\section{Energy estimates on a spacetime slab of fixed time length}\label{sec: main energy estimates}

Before turning to the proofs of Theorems \ref{thm: quasilinear, 1, local bootstrap} and \ref{thm: quasilinear, 2, exp decay}, we need a preliminary Proposition.

\begin{proposition}\label{prop: quasilinear}
	Let $k\geq 7$ and let the tensors $a,h$ be as in Sections \ref{subsec: semilinear derivatives}, \ref{subsec: assumption} respectively. Let $\psi$ be a solution of the quasilinear wave equation~\eqref{eq: quasilinear} on the Schwarzschild--de~Sitter domain $D_\delta(\tau_1,\tau_2)$ for $\tau_1\leq \tau_2$ and any $\delta>0$ sufficiently small. Then, there exists a positive constant
	\begin{equation}
		C(k,M,\Lambda,A_{[k]},B_{[k]})>0,
	\end{equation}
	 where for $A_{[k]},B_{[k]}$ see Sections \ref{subsec: semilinear derivatives}, \ref{subsec: assumption} respectively, and an $\epsilon>0$ sufficiently small, such that if 
	\begin{equation}\label{eq: main theorem quasilinear, eq 0}
		\sup_{D_\delta(\tau_1,\tau_2)}\sum_{1\leq j\leq k-1}\sum_{\partial\in\{\partial_{\bar{t}},\partial_r,\Omega_1,\Omega_2,\Omega_3\}}|\partial^j\psi|\leq \sqrt{\epsilon}
	\end{equation}
	holds, then the following estimate holds
	\begin{equation}\label{eq: main theorem quasilinear, eq 1}
	\begin{aligned}
	E_{\mathcal{G},k}[\psi](\tau_2)+\int_{\tau_1}^{\tau_2}d\tau  E_{\mathcal{G},k} (\tau)  \leq & CE_{\mathcal{G},k}[\psi](\tau_1)	+C\int_{\tau_1}^{\tau_2}d\tau E_{\mathcal{G},k}[\psi](\tau)E_{k+1}[\psi](\tau),
	\end{aligned}
	\end{equation}
	for all $\tau_1\leq\tau_2$. 
\end{proposition}

\begin{proof}[\textbf{Proof of Proposition \ref{prop: quasilinear}}]
	We will use the derivatives notation of Section \ref{subsec: derivatives notation}. Moreover, for the semilinear term $\partial\psi\cdot\partial\psi$ and for the definition of smooth tensors $a,h$ see Sections \ref{subsec: semilinear derivatives}, \ref{subsec: assumption} respectively. In this proof, when we write 
	\begin{equation}
		h_{ab}
	\end{equation}
	it is to be understood as $h_{ab}(\nabla\psi)$ in Cartesian coordinates, see~\eqref{eq: cartesian manifold}. When we write 
	\begin{equation}
		h^{ab}
	\end{equation}
	it is to be understood as $h^{ab}(\nabla\psi)=g^{ab}(\nabla\psi)-\mathring{g}^{ab}$, in Cartesian coordinates. Finally, in this proof when we write $\lesssim$ it is to be understood that we omit a constant $C(k,M,\Lambda,A_{[m]},B_{[m]})$, where for $A_{[m]},B_{[m]}$ see Sections \ref{subsec: semilinear derivatives}, \ref{subsec: assumption} respectively, and $m\leq k+1$.

	Note that we may rewrite $\Box_{g(\nabla\psi)}\psi=\partial\psi\cdot\partial\psi$ as 
	\begin{equation}\label{eq: energy estimate for quasilinear, 0}
	\Box_{\mathring{g}}\psi=\underbrace{\left(\Box_{\mathring{g}}-\Box_{g(\nabla\psi)}\right)\psi}_{F_1[\psi]}\underbrace{+\partial\psi\cdot\partial\psi}_{F_2[\psi]}.    
	\end{equation}
	We name, for convenience,
	\begin{equation}
	F[\psi]=F_1[\psi]+F_2[\psi]=\left(\Box_{\mathring{g}}-\Box_{g(\nabla\psi)}\right)\psi+\partial\psi\cdot\partial\psi.
	\end{equation}
	 Then, by the integrated energy estimate~\eqref{eq: energy estimate of inhomogeneuos equation} of Theorem \ref{thm: wave equation with inhomogeneities} for the inhomogeneous wave equation~\eqref{eq: energy estimate for quasilinear, 0}, we obtain the following
	\begin{equation}\label{eq: energy estimate for quasilinear, 1}
	\begin{aligned}
	&   E_{\mathcal{G},k}[\psi](\tau_2)+\int_{\tau_1}^{\tau_2}d\tau E_{\mathcal{G},k}[\psi](\tau) \\
	&   \lesssim E_{\mathcal{G},k}[\psi](\tau_1)+\int\int_{D(\tau_1,\tau_2)} \sum_{0\leq i_1+i_2+i_3\leq k-2}\sum_\alpha\left((\partial_{\bar{t}})^{i_1}\partial_r^{i_2}(\Omega_\alpha)^{i_3} F\right)^2+(1-\mu)\sum_{0\leq i+j\leq k-2}\sum_\alpha\left((\partial_{\bar{t}})^i (\Omega_\alpha)^j\mathcal{G}F\right)^2\\
	&	\quad\quad\quad\quad\quad\quad+ \int_{\{\bar{t}=\tau_2\}} \sum_{0\leq i_1+i_2+i_3\leq k-3}\sum_\alpha(1-\mu)^{2i_3+1}\left(\partial_{\bar{t}}^{i_1}\Omega_\alpha^{i_2}\partial_r^{i_3}\mathcal{G}F\right)^2+\sum_{0\leq i_1+i_2+i_3\leq k-3}\sum_\alpha\left(\partial_{\bar{t}}^{i_1}\Omega_\alpha^{i_2}\partial_r^{i_3}F\right)^2\\
	&   \lesssim E_{\mathcal{G},k}[\psi](\tau_1)+\int\int_{D(\tau_1,\tau_2)} \sum_{0\leq i_1+i_2+i_3\leq k-2}\sum_\alpha\left((\partial_{\bar{t}})^{i_1}\partial_r^{i_2}(\Omega_\alpha)^{i_3} F_1\right)^2+\sum_{0\leq i+j\leq k-2}\sum_\alpha\left((\partial_{\bar{t}})^i (\Omega_\alpha)^j\mathcal{G}F_1\right)^2\\
	&   \quad\quad\quad\quad\quad\quad +\int\int_{D(\tau_1,\tau_2)} \sum_{0\leq i_1+i_2+i_3\leq k-2}\sum_\alpha\left((\partial_{\bar{t}})^i(\Omega_\alpha)^j F_2\right)^2+(1-\mu)\sum_{0\leq i+j\leq k-2}\sum_\alpha\left((\partial_{\bar{t}})^{i_1}\partial_r^{i_2} (\Omega_\alpha)^{i_3}\mathcal{G}F_2\right)^2\\
	&\quad\quad\quad\quad\quad\quad+ \int_{\{\bar{t}=\tau_2\}} \sum_{0\leq i_1+i_2+i_3\leq k-3}\sum_\alpha(1-\mu)^{2i_3+1}\left(\partial_{\bar{t}}^{i_1}\Omega_\alpha^{i_2}\partial_r^{i_3}\mathcal{G}F\right)^2+\sum_{0\leq i_1+i_2+i_3\leq k-3}\sum_\alpha\left(\partial_{\bar{t}}^{i_1}\Omega_\alpha^{i_2}\partial_r^{i_3}F\right)^2.
	\end{aligned}
	\end{equation}

	We readily bound the term $\int\int_{D_\delta(\tau_1,\tau_2)}(1-\mu)\sum_{i=0}^{k-2} \left(\partial_{\bar{t}}^i \mathcal{G} F_2\right)^2$ as follows

	\begin{equation}\label{eq: energy estimate for quasilinear, 1.09}
	\begin{aligned}
	&   \int\int_{D(\tau_1,\tau_2)}(1-\mu)\sum_{i=0}^{k-2} \left(\partial_{\bar{t}}^i\mathcal{G} F_2\right)^2 =\int\int_{D(\tau_1,\tau_2)}(1-\mu)\sum_{i=0}^{k-2} \left( \partial_{\bar{t}}^i\mathcal{G} \left(a^{\alpha\beta}\partial_\alpha\psi\partial_\beta\psi\right)\right)^2\\
	&	\quad = \int\int_{D(\tau_1,\tau_2)}(1-\mu)\sum_{i=0}^{k-2} \left( \partial_{\bar{t}}^i\left([\mathcal{G},a^{\alpha\beta}]\partial_\alpha\psi\partial_\beta\psi+a^{\alpha\beta}\mathcal{G}\left(\partial_\alpha\psi\partial_\beta\psi\right)\right)\right)^2 \\
	&	\quad = \int\int_{D(\tau_1,\tau_2)}(1-\mu)\sum_{i=0}^{k-2} \left( \partial_{\bar{t}}^i\left([\mathcal{G},a^{\alpha\beta}]\partial_\alpha\psi\partial_\beta\psi+a^{\alpha\beta}\partial_\beta\psi\mathcal{G}\partial_\alpha\psi+\partial_\alpha\psi\mathcal{G}\partial_\beta\psi\right)\right)^2 \\
	&	\quad = \int\int_{D(\tau_1,\tau_2)}(1-\mu)\sum_{i=0}^{k-2} \Bigg( \partial_{\bar{t}}^i\Bigg([\mathcal{G},a^{\alpha\beta}]\partial_\alpha\psi\partial_\beta\psi+a^{\alpha\beta}\partial_\beta\psi[\mathcal{G},\partial_\alpha]\psi+a^{\alpha\beta}\partial_\beta\psi\partial_\alpha\mathcal{G}\psi\\
	&	\quad\quad\quad\quad\quad\quad\quad\quad\quad\quad\quad\quad\quad\quad\quad+a^{\alpha\beta}\partial_\alpha\psi[\mathcal{G},\partial_\beta]\psi+a^{\alpha\beta}\partial_\alpha\psi\partial_\beta\mathcal{G}\psi\Big)\Bigg)^2	\\
	&	\quad \lesssim \int\int_{D(\tau_1,\tau_2)}(1-\mu)\sum_{i=0}^{k-2} \left(\partial_{\bar{t}}^i\left([\mathcal{G},a^{\alpha\beta}]\partial_\alpha\psi\partial_\beta\psi\right)\right)^2+(1-\mu)\sum_{i=0}^{k-2}\left(\partial_{\bar{t}}^i\left(a^{\alpha\beta}\partial_\beta\psi[\mathcal{G},\partial_\alpha]\psi\right)\right)^2\\
	&	\quad\quad\quad\quad\quad\quad+(1-\mu)\sum_{i=0}^{k-2}\left(\partial_{\bar{t}}^i\left(a^{\alpha\beta}\partial_\alpha\psi[\mathcal{G},\partial_\beta]\psi\right)\right)^2\\
	&	\quad\quad\quad\quad\quad\quad+(1-\mu)\sum_{i=0}^{k-2}\left(\partial_{\bar{t}}^i\left(a^{\alpha\beta}\partial_\alpha\psi\partial_\beta\mathcal{G}\psi\right)\right)^2+(1-\mu)\sum_{i=0}^{k-2}\left(\partial_{\bar{t}}^i\left(a^{\alpha\beta}\partial_\beta\psi\partial_\alpha\mathcal{G}\psi\right)\right)^2\\
	&	\quad \lesssim  \int\int_{D(\tau_1,\tau_2)}\sum_{1\leq i+j\leq k-1}(\partial^i\psi)^2(\partial^j\psi)^2\\
	&	\quad\quad\quad\quad\quad\quad+(1-\mu)\sum_{i=0}^{k-2}\left(\partial_{\bar{t}}^i\left(a^{\alpha\beta}\partial_\alpha\psi\partial_\beta\mathcal{G}\psi\right)\right)^2+(1-\mu)\sum_{i=0}^{k-2}\left(\partial_{\bar{t}}^i\left(a^{\alpha\beta}\partial_\beta\psi\partial_\alpha\mathcal{G}\psi\right)\right)^2 \\
	&	\quad \lesssim \int\int_{D(\tau_1,\tau_2)}\sum_{1\leq i+j\leq k-1}(\partial^i\psi)^2(\partial^j\psi)^2+\int\int_{D(\tau_1,\tau_2)}\sum_{1\leq i+j\leq k-2,j\leq\floor{\frac{k-2}{2}}}(1-\mu)(\partial_{\bar{t}}^i\partial\mathcal{G}\psi)^2(\partial^j\partial\psi)^2.\\
	\end{aligned}
	\end{equation}
	Therefore, by using the coarea formula~\eqref{eq: subsec: coarea, eq 1}, see Section~\ref{subsec: coarea}, we bound~\eqref{eq: energy estimate for quasilinear, 1.09} from
	\begin{equation}\label{eq: energy estimate for quasilinear, 1.1}
	\begin{aligned}
	&	\int_{\tau_1}^{\tau_2}d\tau \sum_{1\leq i+j\leq k-1,i\leq \floor{\frac{k-1}{2}}}\sup_{\{\bar{t}=\tau\}}(\partial^i\psi)^2 \int_{\{\bar{t}=\tau\}}(\partial^j\psi)^2+ \int_{\tau_1}^{\tau_2}d\tau \sup_{\{\bar{t}=\tau\}}\sum_{1\leq i+j\leq k-2,j\leq \floor{\frac{k-2}{2}}}(\partial^j\partial\psi)^2\int_{\{\bar{t}=\tau\}}(1-\mu)(\partial_{\bar{t}}^i\partial\mathcal{G}\psi)^2  \\
	&	\quad\quad\quad +\int_{\tau_1}^{\tau_2}d\tau \sum_{j\geq \floor{\frac{k-2}{2}}+1,i\leq k-2-\ceil{\frac{k-2}{2}}-1} \sup_{\{\bar{t}=\tau\}} (\partial^{i+2}\psi)^2\int_{\{\bar{t}=\tau\}}(\partial^{j+1}\psi)^2\\
	&   \quad \lesssim \int_{\tau_1}^{\tau_2}d\tau E_{k-1}[\psi](\tau)E_{k-1}[\psi](\tau)+\int_{\tau_1}^{\tau_2}d\tau E_{k-1}[\psi](\tau)E_{\mathcal{G},k}[\psi](\tau)\\
	&   \quad \leq C(k,M,\Lambda,A_{[k-1]},B_{[k-1]}) \int_{\tau_1}^{\tau_2}d\tau E_{k}[\psi](\tau)E_{\mathcal{G},k}[\psi](\tau)
	\end{aligned}
	\end{equation}
	where in the first inequality we use the Sobolev inequality, see Lemma~\ref{lem: sobolev estimate}, and we also used the requirement $k\geq 6$. The remaining terms
	\begin{equation}
	\int\int_{D(\tau_1,\tau_2)}(1-\mu)\sum_{0\leq i+j\leq k-2}\sum_\alpha\left((\partial_{\bar{t}})^{i_1} (\Omega_\alpha)^{i_3}\mathcal{G}F_2\right)^2 
	\end{equation}
	are similarly bounded by the right hand side of~\eqref{eq: energy estimate for quasilinear, 1.1}.

	Moreover, we readily bound the term $\int\int_{D(\tau_1,\tau_2)}\sum_{i=0}^{k-2} \left(\partial_{\bar{t}}^i F_2\right)^2$ on the right hand side of equation~\eqref{eq: energy estimate for quasilinear, 1} by using the coarea formula~\eqref{eq: subsec: coarea, eq 1}, see Section~\ref{subsec: coarea}, by the expression below
	\begin{equation}\label{eq: energy estimate for quasilinear, 1.1, 2}
	\begin{aligned}
	&   \int\int_{D(\tau_1,\tau_2)}\sum_{i=0}^{k-2} \left(\partial_{\bar{t}}^i F_2\right)^2=\int\int_{D(\tau_1,\tau_2)}\sum_{i=0}^{k-2} \left(\partial_{\bar{t}}^i\left(a^{\alpha\beta}\partial_\alpha\psi\partial_\beta\psi\right)\right)^2 \\
	&	\quad \lesssim  \int\int_{D(\tau_1,\tau_2)} \sum_{1\leq i+j\leq k-1}(\partial^i\psi)^2(\partial^j\psi)^2 \\
	&	\quad \lesssim \int_{\tau_1}^{\tau_2}\sum_{1\leq i+j\leq k-1,i\leq \floor{\frac{k-1}{2}}} \sup_{\{\bar{t}=\tau\}}(\partial^i\psi)^2\int_{\{\bar{t}=\tau\}} (\partial^j\psi)^2\\
	&	\quad \lesssim \int_{\tau_1}^{\tau_2}d\tau E_{k-1}[\psi](\tau)E_{k-1}[\psi](\tau) \\	
	&   \quad \leq C(k,M,\Lambda,A_{[k-1]},B_{[k-1]}) \int_{\tau_1}^{\tau_2}d\tau E_{k-1}[\psi](\tau) E_{\mathcal{G},k}[\psi](\tau),
	\end{aligned}
	\end{equation}	
	where in the second to last inequality we use the Sobolev inequality, see Lemma~\ref{lem: sobolev estimate}, and we also used the requirement $k\geq 4$. The remaining terms 
	\begin{equation}
	\int\int_{D(\tau_1,\tau_2)} \sum_{0\leq i_1+i_2+i_3\leq k-2}\sum_\alpha\left((\partial_{\bar{t}})^i(\Omega_\alpha)^j F_2\right)^2
	\end{equation}
	are similarly bounded by the right hand side of~\eqref{eq: energy estimate for quasilinear, 1.1, 2}. Note that \textit{we have not yet appealed to the smallness}~\eqref{eq: main theorem quasilinear, eq 0}, since we have thus far only been discussing semi-linear terms.

	We estimate the term $\int\int_{D_\delta(\tau_1,\tau_2)} \sum_{i=0}^{k-2} \left(\partial_{\bar{t}}^i\mathcal{G} F_1\right)^2$, on the right hand side of equation~\eqref{eq: energy estimate for quasilinear, 1}, by using the coarea formula~\eqref{eq: subsec: coarea, eq 1}, by the following expression
	\begin{equation}\label{eq: energy estimate for quasilinear, 2}
	\begin{aligned}
	&  \int_{\tau_1}^{\tau_2}d\tau\int_{\{\bar{t}=\tau\}}\sum_{i=0}^{k-2}\Bigg(\partial_{\bar{t}}^i\mathcal{G}\left((\mathring{g}^{ab}-g^{ab}(\nabla\psi))\partial_a\partial_b\psi\right)  \\
	&   \quad\quad\quad\quad\quad\quad\quad\quad\quad\quad+\partial_{\bar{t}}^i\mathcal{G}\left(\mathring{g}^{ab}\Gamma^c\:_{ab}(\mathring{g})\partial_c\psi-g^{ab}(\nabla\psi)\Gamma^c\:_{ab}(g(\nabla\psi))\partial_c\psi\right)\Bigg)^2.
	\end{aligned}
	\end{equation}
	
	Now, we note that 
	\begin{equation}\label{eq: energy estimate for quasilinear, 2, 1, Christophel symbols}
	\Gamma^c\:_{ab}(g(\nabla\psi))=\Gamma^c\:_{ab}(\mathring{g})+S^c\:_{ab}(h),
	\end{equation}
	where 
	\begin{equation}\label{eq: energy estimate for quasilinear, 2, 2, the nonlinear part of the christophel symbols}
	\begin{aligned}
	S^c\:_{ab}(h)	\:\dot{=}\:& \frac{1}{2}\mathring{g}^{ce}\left(\partial_e h_{ab}+\partial_a h_{eb}+\partial_b h_{ea}\right) \\
	&   +\frac{1}{2} h^{ce}\left(-\partial_{e} h_{ab}+\partial_a h_{eb}+\partial_b h_{ea}\right)\\
	&   +\frac{1}{2} h^{ce}\left(-\partial_e 
	\mathring{g}_{ab}+\partial_a \mathring{g}_{eb}+\partial_b \mathring{g}_{ea}\right).
	\end{aligned}
	\end{equation}
	Therefore, we write~\eqref{eq: energy estimate for quasilinear, 2} in the following form 	
	\begin{equation}
	\begin{aligned}
	&  \int_{\tau_1}^{\tau_2}d\tau\int_{\{\bar{t}=\tau\}} \sum_{i=0}^{k-2}\left(\partial_{\bar{t}}^i\mathcal{G}\left(h^{ab}\partial_a\partial_b\psi\right)  +\partial_{\bar{t}}^i\mathcal{G}\left(\left(\mathring{g}^{ab}S^c\:_{ab}(h)-h^{ab}\Gamma^c\:_{ab}(\mathring{g})-h^{ab}S^c\:_{ab}(h)\right)\partial_c\psi\right)\right)^2 \\
	&	\quad=\int_{\tau_1}^{\tau_2}d\tau\int_{\{\bar{t}=\tau\}}\sum_{i=0}^{k-2}\left(\partial_{\bar{t}}^i\mathcal{G}(h^{ab}\partial_a\partial_b\psi)+\partial_{\bar{t}}^i\mathcal{G}(\mathring{g}^{ab}S^c\:_{ab}(h)\partial_c\psi)-\partial_{\bar{t}}^i\mathcal{G}(h^{ab}\Gamma^c\:_{ab}(\mathring{g})\partial_c\psi)-\partial_{\bar{t}}^i\mathcal{G}(h^{ab}S^c\:_{ab}(h)\partial_c\psi)\right)^2\\
	&	\quad \lesssim\int_{\tau_1}^{\tau_2}d\tau\int_{\{\bar{t}=\tau\}}\sum_{1\leq i\leq k-1}\Bigg(\left(\partial^{i}\left(h^{ab}\partial_a\partial_b\psi\right)\right)^2+\left(\partial^i\left(\mathring{g}^{ab}S^c\:_{ab}(h)\partial_c\psi)\right)\right)^2+\left(\partial^i\left(h^{ab}\Gamma^c\:_{ab}(\mathring{g})\partial_c\psi\right)\right)^2\\
	&	\quad\quad\quad\quad\quad\quad\quad\quad\quad\quad\quad\quad\quad\quad\quad+\left(\partial^i\left(h^{ab}S^c\:_{ab}(h)\partial_c\psi\right)\right)^2\Bigg)
	\end{aligned}
	\end{equation}
	which, by using the smallness~\eqref{eq: main theorem quasilinear, eq 0} for a sufficiently small $\epsilon>0$, the definition of $h$, see Section \ref{subsec: assumption}, and by distributing the derivatives (with the Cartesian coordinates of Section \ref{subsec: derivatives notation}), we bound the above from
	\begin{equation}\label{eq: energy estimate for quasilinear, 2, bound,-1}
	\begin{aligned}
		&\lesssim  \int_{\tau_1}^{\tau_2}d\tau\int_{\{\bar{t}=\tau\}}\sum_{1\leq i\leq k-1}\Bigg( \sum_{i_1+i_2=i}\left(|\partial^{i_1} h^{ab}||\partial^{i_2}\partial_a\partial_b\psi|\right)^2+ \sum_{i_1+i_2+i_3 =i}\left( |\partial^{i_1}\mathring{g}^{ab}||\partial^{i_2}S^c\:_{ab}(h)||\partial^{i_3}\partial_c\psi|\right)^2 \\
		&	\quad\quad\quad\quad\quad\quad\quad\quad\quad\quad\quad +\sum_{i_1+i_2+i_3=i}\left(|\partial^{i_1}h^{ab}||\partial^{i_2}\Gamma^c\:_{ab}(\mathring{g})||\partial^{i_3}\partial_c\psi|\right)^2\\
		& \quad\quad\quad\quad\quad\quad\quad\quad\quad\quad\quad +\sum_{i_1+i_2+i_3=i}\left(|\partial^{i_1}h^{ab}||\partial^{i_2}S^c\:_{ab}(h)||\partial^{i_3}\partial_c\psi|\right)^2 \Bigg)\\
		&	\lesssim \int_{\tau_1}^{\tau_2}d\tau \int_{\{\bar{t}=\tau\}} \sum_{1\leq i\leq k-1,i_1\leq\floor{\frac{i}{2}}}\left(\partial^{i_1+1}\psi\right)^2\left(\partial^{i_2+2}\psi\right)^2+\sum_{4\leq i\leq k-1}\sum_{i_1= \floor{\frac{i}{2}}+1,i_2=\ceil{\frac{i}{2}}-1}\left(\partial^{i_1+1}\psi\right)^2\left(\partial^{i_2+2}\psi\right)^2\\
		&	\quad\quad\quad\quad\quad\quad\quad+\sum_{4\leq i\leq k-1}\sum_{i_1= \floor{\frac{i}{2}}+2,i_2=\ceil{\frac{i}{2}}-2}\left(\partial^{i_1+1}\psi\right)^2\left(\partial^{i_2+2}\psi\right)^2+\sum_{5\leq i\leq k-1}\sum_{i_2\leq\ceil{\frac{i}{2}}-3,i_1>\floor{\frac{i}{2}}+3}\left(\partial^{i_1+1}\psi\right)^2\left(\partial^{i_2+2}\psi\right)^2	\\
		&	\quad\quad\quad\quad\quad\quad\quad	+\sum_{1\leq i\leq 4}\sum_{i_1+i_2=i}(\partial^{i_1+1}\psi)^2(\partial^{i_2+2}\psi)^2\\
		&\quad\quad\quad\quad\quad\quad\quad+\sum_{1\leq i_1+i_2\leq k-1,i_1\leq \ceil{\frac{k-1}{3}}}\left(\partial^{i_1+1}\psi\right)^2\left(\partial^{i_2+1}\psi\right)^2 \\
		&	\quad\quad\quad\quad\quad\quad\quad +\sum_{1\leq i_1+i_2+i_3\leq k-1,i_1\leq\ceil{\frac{k-1}{3}}}\left(\partial^{i_1}\psi\right)^2\left(\partial^{i_2+1}\psi\right)^2\left(\partial^{i_3+1}\psi\right)^2\\
		\end{aligned}
		\end{equation}
		therefore, by using Sobolev inequalities, see Lemma~\ref{lem: sobolev estimate}, we obtain from~\eqref{eq: energy estimate for quasilinear, 2, bound,-1} the following 		
		\begin{equation}\label{eq: energy estimate for quasilinear, 2, bound}
		\begin{aligned}
		& \lesssim \int_{\tau_1}^{\tau_2}d\tau  \sum_{1\leq i\leq k-1,i_1\leq\floor{\frac{i}{2}}}\sup_{\{\bar{t}=\tau\}}\left(\partial^{i_1+1}\psi\right)^2\int_{\{\bar{t}=\tau\}}\left(\partial^{i_2+2}\psi\right)^2\\
		&	\quad\quad\quad\quad+\sum_{4\leq i\leq k-1}\sum_{i_1= \floor{\frac{i}{2}}+1,i_2=\ceil{\frac{i}{2}}-1}\sup_{\{\bar{t}=\tau\}}\left(\partial^{i_1+1}\psi\right)^2\int_{\{\bar{t}=\tau\}}\left(\partial^{i_2+2}\psi\right)^2\\
		&	\quad\quad\quad\quad+\sum_{4\leq i\leq k-1}\sum_{i_1= \floor{\frac{i}{2}}+2,i_2=\ceil{\frac{i}{2}}-2}\sup_{\{\bar{t}=\tau\}}\left(\partial^{i_1+1}\psi\right)^2\int_{\{\bar{t}=\tau\}}\left(\partial^{i_2+2}\psi\right)^2\\
		&	\quad\quad\quad\quad+	\sum_{5\leq i\leq k-1}\sum_{i_2\leq\ceil{\frac{i}{2}}-3,i_1}\sup_{\{\bar{t}=\tau\}}\left(\partial^{i_2+2}\psi\right)^2\int_{\{\bar{t}=\tau\}}\left(\partial^{i_1+1}\psi\right)^2	\\
		&	\quad\quad\quad\quad	+\sum_{1\leq i\leq 4}\sum_{i_1+i_2=i}\sup_{\{\bar{t}=\tau\}}(\partial^{i_1+1}\psi)^2\int_{\{\bar{t}=\tau\}}(\partial^{i_2+2}\psi)^2\\
		&	\quad\quad\quad\quad+\sum_{1\leq i_1+i_2\leq k-1,i_1\leq \ceil{\frac{k-1}{3}}}\sup_{\{\bar{t}=\tau\}}\left(\partial^{i_1+1}\psi\right)^2\int_{\{\bar{t}=\tau\}}\left(\partial^{i_2+1}\psi\right)^2\\
		&	\quad\quad\quad\quad +\sum_{1\leq i_1+i_2+i_3\leq k-1,i_1\leq\ceil{\frac{k-1}{3}}}\sup_{\{\bar{t}=\tau\}}\left(\partial^{i_1}\psi\right)^2\int_{\{\bar{t}=\tau\}}\left(\partial^{i_2+1}\psi\right)^2\left(\partial^{i_3+1}\psi\right)^2\\
		&	\lesssim \int_{\tau_1}^{\tau_2}d\tau E_{k-1}[\psi](\tau)E_{k+1}[\psi](\tau)	\\
		&	\leq C(k,M,\Lambda,A_{[k]},B_{[k]}) \int_{\tau_1}^{\tau_2}d\tau E_{\mathcal{G},k}[\psi](\tau)E_{k+1}[\psi](\tau)
	\end{aligned}
	\end{equation}
	where in the second to last inequality we also used that $k=7$ is the smallest integer such that all of the following hold 
	\begin{equation}
		\left(\floor{\frac{k-1}{2}}+1\right)+2\leq k-1,\quad \left(\floor{\frac{k-1}{2}}+1+1\right)+2\leq k+1,\quad \left(\floor{\frac{k-1}{2}}+2+1\right)+2\leq k+1,\quad \left(\ceil{\frac{k-1}{2}}-3+2\right)+2\leq k-1
	\end{equation}
	where the addition of $+2$ on the above inequalities comes from the Sobolev inequality, see Lemma~\ref{lem: sobolev estimate}. 

	Now, by similar calculations, we bound the following term of the right hand side of~\eqref{eq: energy estimate for quasilinear, 1} 
	\begin{equation}\label{eq: energy estimate for quasilinear 2, bound 2}
		\int\int_{D(\tau_1,\tau_2)} \sum_{0\leq i_1+i_2+i_3\leq k-2}\sum_\alpha\left((\partial_{\bar{t}})^{i_1}\partial_r^{i_2}(\Omega_\alpha)^{i_3} F_1\right)^2+\sum_{0\leq i+j\leq k-2}\sum_\alpha\left((\partial_{\bar{t}})^i (\Omega_\alpha)^j\mathcal{G}F_1\right)^2\leq C \int _{\tau_1}^{\tau_2}E_{\mathcal{G},k}[\psi](\tau)E_{k}[\psi](\tau).
	\end{equation}
	
	Finally, by using the bounds~\eqref{eq: energy estimate for quasilinear, 1.1}~\eqref{eq: energy estimate for quasilinear, 1.1, 2},~\eqref{eq: energy estimate for quasilinear, 2, bound},~\eqref{eq: energy estimate for quasilinear 2, bound 2}, the integrated energy estimate of equation~\eqref{eq: energy estimate for quasilinear, 1} implies the following
	\begin{equation}\label{eq: energy estimate for quasilinear, 3}
	\begin{aligned}
	E_{\mathcal{G},k}[\psi](\tau_2)& +\int_{\tau_1}^{\tau_2}d\tau  E_{\mathcal{G},k} (\tau) \\
	&	\leq C(M,\Lambda) E_{\mathcal{G},k}[\psi](\tau_1)  +C(k,M,\Lambda,A_{[k]},B_{[k]}) \int_{\tau_1}^{\tau_2}d\tau E_{\mathcal{G},k}[\psi](\tau)E_{k+1}[\psi](\tau)\\
	&	\quad+ C(M,\Lambda)\int_{\{\bar{t}=\tau_2\}} \sum_{0\leq i_1+i_2+i_3\leq k-3}\sum_\alpha(1-\mu)^{2i_3+1}\left(\partial_{\bar{t}}^{i_1}\Omega_\alpha^{i_2}\partial_r^{i_3}\mathcal{G}F\right)^2+\sum_{0\leq i_1+i_2+i_3\leq k-3}\sum_\alpha\left(\partial_{\bar{t}}^{i_1}\Omega_\alpha^{i_2}\partial_r^{i_3}F\right)^2,
	\end{aligned}
	\end{equation}
	where recall that $F=\left(\Box_{\mathring{g}}-\Box_{g(\nabla\psi)}\right)\psi+\partial\psi\cdot\partial\psi$.

	Finally, we want to absorb the boundary terms of the right hand side at $\{\bar{t}=\tau_2\}$ by the relevant boundary term $E_{\mathcal{G},k}[\psi](\tau_2)$ of the left hand side. By using the smallness assumption~\eqref{eq: main theorem quasilinear, eq 0}, it is evident that we can appropriately distribute derivatives, in view of the requirement $k\geq 7$, to conclude that 
	\begin{equation}
	\begin{aligned}
		E_{\mathcal{G},k}[\psi](\tau_2)& +\int_{\tau_1}^{\tau_2}d\tau  E_{\mathcal{G},k} (\tau) \\
		&	\leq C(k,M,\Lambda,A_{[k]},B_{[k]}) E_{\mathcal{G},k}[\psi](\tau_1)  +C(k,M,\Lambda,A_{[k]},B_{[k]}) \int_{\tau_1}^{\tau_2}d\tau E_{\mathcal{G},k}[\psi](\tau)E_{k+1}[\psi](\tau),
	\end{aligned}
	\end{equation}
	for a constant $C(k,M,\Lambda,A_{[k]},B_{[k]})$.

	The proof of Proposition \ref{prop: quasilinear} is thus complete.  
\end{proof}

\section{Proof of Theorem \ref{thm: quasilinear, 1, local bootstrap}}\label{sec: proof of thm: quasilinear, 1, local bootstrap}

We are ready to prove Theorem \ref{thm: quasilinear, 1, local bootstrap}.

\begin{proof}[\textbf{Proof of Theorem \ref{thm: quasilinear, 1, local bootstrap}}]
	
	Given $L>0$, and for a
	\begin{equation}
	\tau_{\textit{step}}(L,k,M,\Lambda,A_{[k+1]},B_{[k+1]})>0
	\end{equation}
	to be determined later (see already~\eqref{eq: proof of corollary quasilinear, eq 8.3}), we apply Proposition~\ref{prop: well posedness quasilinear} with $\tau_{\textit{max}}=\tau_{\textit{step}}$ to find
	\begin{equation}
	\delta(\tau_{\textit{step}},k,M,\Lambda,A_{[k+1]},B_{[k+1]})>0,\qquad \epsilon(\tau_{\textit{step}},\delta,k,M,\Lambda,A_{[k+1]},B_{[k+1]})>0
	\end{equation}
	sufficiently small, such that if we take 
	\begin{equation}
	E_{k+1}[\psi](\tau_1)\leq \epsilon
	\end{equation}
	then the solution $\psi$ of the quasilinear wave equation~\eqref{eq: quasilinear} exists in $D_\delta(\tau_1,\tau_2)$, where 
	\begin{equation}
	\tau_2=\tau_1+\tau_{\textit{step}}. 
	\end{equation}
	
	The two inequalities~\eqref{eq: corollary quasilinear, eq 1.1}, \eqref{eq: corollary quasilinear, eq 1.2} are an immediate consequence of the Cauchy stability results~\eqref{eq: prop: well posedness quasilinear, eq 1},~\eqref{eq: prop: well posedness quasilinear, eq 2} respectively, of Proposition \ref{prop: well posedness quasilinear}.

	To obtain the remaining statements, we assume the bootstrap  
	\begin{equation}\label{eq: proof of corollary quasilinear, eq 0}
		\sup_{D_\delta(\tau_1,\tau_2)}\sum_{1\leq j\leq k-1}\sum_{\partial\in\{\partial_{\bar{t}},\partial_r,\Omega_1,\Omega_2,\Omega_3\}}|\partial^j\psi|\leq C_{\text{b}}\sqrt{\epsilon},
	\end{equation}
	for a constant
	\begin{equation}
	C_{\text{b}}(k,M,\Lambda,A_{[k+1]},B_{[k+1]})>0
	\end{equation}
	 which will be determined later.

	First, we prove inequality~\eqref{eq: corollary quasilinear, eq 1}, namely 
	\begin{equation}\label{eq: proof of corollary quasilinear, eq 1}
	\begin{aligned}
	\sup_{\tau\in[\tau_1,\tau_2]}E_{\mathcal{G},k}[\psi](\tau)  &   \leq C(k,M,\Lambda,A_{[k+1]},B_{[k+1]}) E_{\mathcal{G},k}[\psi](\tau_1).\\
	\end{aligned}
	\end{equation}
	To prove~\eqref{eq: proof of corollary quasilinear, eq 1}, we will introduce an additional bootstrap assumption
	\begin{equation}\label{eq: proof of corollary quaslilinear, eq 2, bootstrap}
	\sup_{\tau\in[\tau_1,\tau_2]}E_{\mathcal{G},k}[\psi](\tau)\leq C_{\text{boot}} E_{\mathcal{G},k}[\psi](\tau_1).
	\end{equation}
	We choose $\epsilon(C_{\text{boot}})>0$ sufficiently small so that, in view of the bootstrap~\eqref{eq: proof of corollary quasilinear, eq 0}, the energy estimate~\eqref{eq: main theorem quasilinear, eq 1} of Proposition~\ref{prop: quasilinear} holds, namely 
	\begin{equation}\label{eq: proof of corollary quasilinear, eq 6}
		\begin{aligned}
			&	E_{\mathcal{G},k}[\psi](\tau_2)+\int_{\tau_1}^{\tau_2}E_{\mathcal{G},k}[\psi](\tau^\prime)d\tau^\prime \\
			&	\qquad \leq C(k,M,\Lambda,A_{[k]},B_{[k]}) E_{\mathcal{G},k}[\psi](\tau_1)+C(k,M,\Lambda,A_{[k]},B_{[k]})\int_{\tau_1}^{\tau_2}E_{\mathcal{G},k}[\psi](\tau^\prime)E_{k+1}[\psi](\tau^\prime)d\tau^\prime.
		\end{aligned}
	\end{equation}
	Now we use the bootstrap assumption~\eqref{eq: proof of corollary quaslilinear, eq 2, bootstrap} and inequality~\eqref{eq: corollary quasilinear, eq 1.1}, which is already proven, to obtain
	\begin{equation}\label{eq: proof of corollary quasilinear, eq 7}
		\begin{aligned}
			E_{\mathcal{G},k}[\psi](\tau_2)+\int_{\tau_1}^{\tau_2}E_{\mathcal{G},k}[\psi](\tau^\prime)d\tau^\prime \leq & C(k,M,\Lambda,A_{[k+1]},B_{[k+1]}) E_{\mathcal{G},k}[\psi](\tau_1)\\
			&	+C(k,M,\Lambda,A_{[k+1]},B_{[k+1]})C_{\text{boot}}\tau_{\textit{step}}E_{\mathcal{G},k}[\psi](\tau_1)E_{k+1}[\psi](\tau_1)\\
			\leq & C(k,M,\Lambda,A_{[k+1]},B_{[k+1]}) E_{\mathcal{G},k}[\psi](\tau_1)\\
			&	+C(k,M,\Lambda,A_{[k+1]},B_{[k+1]})C_{\text{boot}}\tau_{\textit{step}}\cdot\epsilon\cdot E_{\mathcal{G},k}[\psi](\tau_1)\\
			\leq & C(k,M,\Lambda,A_{[k+1]},B_{[k+1]})\left(1	+C_{\text{boot}}\tau_{\textit{step}}\epsilon\right)E_{\mathcal{G},k}[\psi](\tau_1).\\
		\end{aligned}
	\end{equation}
	Therefore, by choosing
	\begin{equation}
	\frac{C_{\text{boot}}(k,M,\Lambda,A_{[k+1]},B_{[k+1]})}{2}\gg C(k,M,\Lambda,A_{[k+1]},B_{[k+1]}) 
	\end{equation}
	we take $\epsilon(\tau_\textit{step},k,M,\Lambda,A_{[k+1]},B_{[k+1]})>0$ sufficiently small in inequality~\eqref{eq: proof of corollary quasilinear, eq 7}, and obtain 
	\begin{equation}
	\sup_{\tau\in[\tau_1,\tau_2]}E_{\mathcal{G},k}[\psi](\tau)\leq \frac{C_{\text{boot}}}{2}E_{\mathcal{G},k}[\psi](\tau_1).
	\end{equation}
	We improved the bootstrap~\eqref{eq: proof of corollary quaslilinear, eq 2, bootstrap} and thus proved~\eqref{eq: proof of corollary quasilinear, eq 1}. Therefore, we concluded inequality~\eqref{eq: corollary quasilinear, eq 1}.

	Now, we proceed to prove inequality~\eqref{eq: corollary quasilinear, eq 2}. We use the energy estimate~\eqref{eq: proof of corollary quasilinear, eq 6}, in view of the now established bound~\eqref{eq: proof of corollary quasilinear, eq 1}, and inequality~\eqref{eq: corollary quasilinear, eq 1.1}, namely $E_{k+1}[\psi](\tau^\prime)\leq C_{\textit{wp}} E_{k+1}[\psi](\tau_1)$ for all $\tau^\prime\in[\tau_1,\tau_2]$, to obtain
	\begin{equation}\label{eq: proof of corollary quasilinear, eq 8}
	\begin{aligned}
	&    \int_{\tau_1}^{\tau_2}E_{\mathcal{G},k}[\psi](\tau^\prime)d\tau^\prime \leq C E_{\mathcal{G},k}[\psi](\tau_1)+C\tau_{\textit{step}}E_{\mathcal{G},k}[\psi](\tau_1)E_{k+1}[\psi](\tau_1)\\
	\implies & \int_{\frac{\tau_1+\tau_2}{2}}^{\tau_2}E_{\mathcal{G},k}[\psi](\tau^\prime)d\tau^\prime \leq C E_{\mathcal{G},k}[\psi](\tau_1)+C\tau_{\textit{step}}E_{\mathcal{G},k}[\psi](\tau_1)E_{k+1}[\psi](\tau_1)\\
	\implies & E_{\mathcal{G},k}[\psi](\tau_{1,2})\leq \frac{2C(k,M,\Lambda,A_{[k+1]},B_{[k+1]})}{\tau_{\textit{step}}}\left(E_{\mathcal{G},k}[\psi](\tau_1)+\tau_{\textit{step}}E_{\mathcal{G},k}[\psi](\tau_1)E_{k+1}[\psi](\tau_1)\right)
	\end{aligned}
	\end{equation}
	where we used that there exists a $\tau_{1,2}\in[\frac{\tau_1+\tau_2}{2},\tau_2]$ such that 
	\begin{equation}
	\frac{\tau_{\textit{step}}}{2}E_{\mathcal{G},k}[\psi](\tau_{1,2})\leq \int_{\frac{\tau_1+\tau_2}{2}}^{\tau_2}E_{\mathcal{G},k}[\psi](\tau^\prime)d\tau^\prime,
	\end{equation}
	by the mean value theorem. Therefore, from inequality~\eqref{eq: proof of corollary quasilinear, eq 8} we conclude that, if we take 
	\begin{equation}
	E_{k+1}[\psi](\tau_1)\leq \epsilon,
	\end{equation}
	then there exists a $\tau_{1,2}\in[\frac{\tau_1+\tau_2}{2},\tau_2]$ such that 
	\begin{equation}\label{eq: proof of corollary quasilinear, eq 8.1}
	E_{\mathcal{G},k}[\psi](\tau_{1,2})\leq \left(\frac{2C}{\tau_{\textit{step}}}+2C\epsilon\right)E_{\mathcal{G},k}[\psi](\tau_1).
	\end{equation}
	Furthermore, by using the finite in time bound~\eqref{eq: proof of corollary quasilinear, eq 1}, for $E_{\mathcal{G},k}[\psi]$, in the time domain 
	\begin{equation}
	[\tau_{1,2},\tau_2],
	\end{equation}
	in conjunction with~\eqref{eq: proof of corollary quasilinear, eq 8.1}, we find a, potentially different, constant $C(k,M,\Lambda,A_{[k+1]},B_{[k+1]})>0$ such that
	\begin{equation}\label{eq: proof of corollary quasilinear, eq 8.2}
	E_{\mathcal{G},k}[\psi](\tau_2)\leq \left(\frac{2C}{\tau_{\textit{step}}}+2C\epsilon\right)E_{\mathcal{G},k}[\psi](\tau_1).
	\end{equation}
	Finally, we choose $\tau_{\textit{step}}(L,k,M,\Lambda,A_{[k+1]},B_{[k+1]})>0$ sufficiently large and conclude 
	\begin{equation}\label{eq: proof of corollary quasilinear, eq 8.3}
	E_{\mathcal{G},k}[\psi](\tau_2)\leq \frac{1}{L}E_{\mathcal{G},k}[\psi](\tau_1),
	\end{equation}
	after taking $\epsilon(\tau_{\textit{step}})$ sufficiently small. We have concluded the inequality~\eqref{eq: corollary quasilinear, eq 2}.

To conclude~\eqref{eq: corollary quasilinear, eq 3}, we use the now established~\eqref{eq: corollary quasilinear, eq 1},~\eqref{eq: corollary quasilinear, eq 1.1},~\eqref{eq: corollary quasilinear, eq 1.2},~\eqref{eq: corollary quasilinear, eq 2} and the  classical interpolation Lemma~\ref{lem: subsec: energies, lem 1}. Namely, by recalling the constant 
\begin{equation}
C_{\textit{int}}(k,M,\Lambda)>0
\end{equation}
from the intepolation Lemma~\ref{lem: subsec: energies, lem 1}, we obtain
\begin{equation}\label{eq: proof of corollary quasilinear, eq 9}
	\begin{aligned}
		E_{k+1}[\psi](\tau_2)	&	\leq C_{\textit{int}} \left(E_{k-1}[\psi](\tau_2)\right)^{1/3}\left(E_{k+2}[\psi](\tau_2)\right)^{2/3}\\
		&	\leq C_{\text{int}} \left(E_{\mathcal{G},k}[\psi](\tau_2)\right)^{1/3}\left(E_{k+2}[\psi](\tau_2)\right)^{2/3}\\
		&	\leq C_{\textit{int}} \left(\frac{1}{L}\right)^{1/3}\left(C_{\textit{wp}}\right)^{2/3} \left(E_{\mathcal{G},k}[\psi](\tau_1)\right)^{1/3}\left(E_{k+2}[\psi](\tau_1)\right)^{2/3}\\
		&	\leq C_{\textit{int}} e^{-\frac{1}{3}\log(L)+\frac{2}{3}\log C_{\textit{wp}}}\left(E_{\mathcal{G},k}[\psi](\tau_1)\right)^{1/3}\left(E_{k+2}[\psi](\tau_1)\right)^{2/3},
	\end{aligned}
\end{equation}
where for the constant
\begin{equation}
C_{\textit{wp}}(k,M,\Lambda,A_{[k+1]},B_{[k+1]})>0
\end{equation}
see Proposition~\ref{prop: well posedness quasilinear}. We have concluded~\eqref{eq: corollary quasilinear, eq 3}.

Of course our inequalities are still conditional on improving the bootstrap assumption~\eqref{eq: proof of corollary quasilinear, eq 0}. We see immediately however that by taking $C_{\text{b}}\gg C_{\textit{wp}}$ we improve the bootstrap~\eqref{eq: proof of corollary quasilinear, eq 0} by applying a Sobolev inequality on the already proved inequality~\eqref{eq: corollary quasilinear, eq 1.1}.

The proof of Theorem \ref{thm: quasilinear, 1, local bootstrap} is thus complete.
\end{proof}

\section{Proof of Theorem \ref{thm: quasilinear, 2, exp decay}}\label{sec: proof of theorem 2, exp decay}

We are now ready to prove Theorem \ref{thm: quasilinear, 2, exp decay}.

\begin{proof}[\textbf{Proof of Theorem \ref{thm: quasilinear, 2, exp decay}}]
	We will use the results of Theorem \ref{thm: quasilinear, 1, local bootstrap} in conjunction with an iteration argument in consecutive spacetime regions.

	To apply Theorem \ref{thm: quasilinear, 1, local bootstrap}, we choose $L>0$ sufficiently large so that 
	\begin{equation}\label{eq: subsec: delta, eq 1}
	L>\left(2\left(C_{\textit{int}}+1\right)e^{\frac{2}{3}\log (C_{\textit{wp}})}\right)^3
	\end{equation}
	where for the constants $C_{\textit{int}},C_{\textit{wp}}>0$ see equation~\eqref{eq: corollary quasilinear, eq 3}. Note that with this choice, the constant on the right hand side of~\eqref{eq: corollary quasilinear, eq 3} satisfies 
	\begin{equation}
	C_{\textit{int}}e^{-\frac{1}{3}\log L+ \frac{2}{3}\log C_{\textit{wp}}}<\frac{1}{2}. 
	\end{equation}
	This gives us  $\tau_{\textit{step}},\delta$ and $\epsilon$ from the statement of Theorem~\ref{thm: quasilinear, 1, local bootstrap} and our choise of $L$, namely~\eqref{eq: subsec: delta, eq 1}. We obtain 
	\begin{equation}\label{eq: proof: thm 2, exp decay, eq 3, -2}
		E_{k+1}[\psi](\tau_2)	\leq \left(E_{\mathcal{G},k}[\psi](\tau_1)\right)^{1/3}\left(E_{k+2}[\psi](\tau_1)\right)^{2/3},
	\end{equation}
	for $\tau_2=\tau_1+\tau_{\textit{step}}$.

	First, we prove that there exists a strictly increasing sequence of real numbers $\{t_i\}_{i\in\mathbb{N}}$, $t_i\rightarrow \infty$, such that 
	\begin{equation}\label{eq: proof: thm 2, exp decay, eq 3, -1}
	t_0=0,\qquad t_{i+1}-t_i= \tau_{\textit{step}},
	\end{equation}
	the solution exists in $D_\delta(0,t_i)$ and the following hold
	\begin{equation}\label{eq: proof: thm 2, exp decay, eq 3}
	\begin{aligned}
	E_{\mathcal{G},k}[\psi](t_i)    &   \leq \left(\frac{1}{L}\right)^iE_{\mathcal{G},k}[\psi](0),\\
	E_{k+1}[\psi](t_i) &	\leq \epsilon,\\
	E_{k+2}[\psi](t_i)	&	\leq \left(C_{\textit{wp}}\right)^i E_{k+2}[\psi](0)
	\end{aligned}
	\end{equation}

	Note that~\eqref{eq: proof: thm 2, exp decay, eq 3} holds for $t_0$. For the purpose of using induction, we assume that the solution exists in $D_\delta(0,t_i)$ and also assume the iterative step that~\eqref{eq: proof: thm 2, exp decay, eq 3} holds for $t_i$, so we want to prove~\eqref{eq: proof: thm 2, exp decay, eq 3} for $t_{i+1}$. Now, in view of Theorem~\ref{thm: quasilinear, 1, local bootstrap} we obtain that the solution exists in $D_\delta(0,t_{i+1})$ and moreover in view of inequality~\eqref{eq: proof: thm 2, exp decay, eq 3, -2}, the condition for $L$ namely~\eqref{eq: subsec: delta, eq 1} and  by the iterative step assumptions we obtain
	\begin{equation}\label{eq: proof: thm 2, exp decay, eq 3..2}
		\begin{aligned}
			E_{k+1}[\psi](t_{i+1})	&	\leq  \left(E_{\mathcal{G},k}[\psi](t_i)\right)^{1/3}\left(E_{k+2}[\psi](t_i)\right)^{2/3}\\
			&	\leq  \left(\frac{1}{L}\right)^{i/3}e^{\frac{2}{3} i\log(C_{\textit{wp}})}\left(E_{\mathcal{G},k}[\psi](0)\right)^{1/3}\left(E_{k+2}[\psi](0)\right)^{2/3}\\
			&	\leq e^{i\left(-\frac{1}{3}\log(L)+\frac{2}{3}\log(C_{\textit{wp}})\right)}\left(E_{\mathcal{G},k}[\psi](0)\right)^{1/3}\left(E_{k+2}[\psi](0)\right)^{2/3}\\
			&	\leq \left(E_{\mathcal{G},k}[\psi](0)\right)^{1/3}\left(E_{k+2}[\psi](0)\right)^{2/3}\\
			&	\leq \epsilon.
		\end{aligned}
	\end{equation}
	Moreover, since $E_{k+1}[\psi](0)\leq \epsilon$ is sufficiently small $\epsilon>0$, we apply~\eqref{eq: corollary quasilinear, eq 2} of Theorem~\ref{thm: quasilinear, 1, local bootstrap} in conjunction with the inductive step and obtain 
	\begin{equation}\label{eq: proof: thm 2, exp decay, eq 3..25}
		E_{\mathcal{G},k}[\psi](t_{i+1})		\leq \frac{1}{L}E_{\mathcal{G},k}[\psi](t_i)\leq \left(\frac{1}{L}\right)^{i+1}E_{\mathcal{G},k}[\psi](0).\\
	\end{equation}
	Furthermore, we apply inequality~\eqref{eq: corollary quasilinear, eq 1.2} of Theorem~\ref{thm: quasilinear, 1, local bootstrap} in conjunction with the inductive step and obtain 
	\begin{equation}\label{eq: proof: thm 2, exp decay, eq 3..3}
		\begin{aligned}
			E_{k+2}[\psi](t_{i+1})	\leq C_{\textit{wp}} E_{k+2}[\psi](t_i)	\leq \left(C_{\textit{wp}}\right)^{i+1} E_{k+2}[\psi](0)
		\end{aligned}
	\end{equation}
	Therefore, by~\eqref{eq: proof: thm 2, exp decay, eq 3..2},~\eqref{eq: proof: thm 2, exp decay, eq 3..25},~\eqref{eq: proof: thm 2, exp decay, eq 3..3}, it follows that~\eqref{eq: proof: thm 2, exp decay, eq 3} holds for $t_{i+1}$. Therefore, by induction~\eqref{eq: proof: thm 2, exp decay, eq 3} hold for all $t_j$, $j\in\mathbb{N}$.

	Now we proceed to prove the the exponential decay~\eqref{eq: thm 2, exp decay, eq 1}, of $E_{\mathcal{G},k}[\psi]$ for all times $\tau\geq 0$. Note that for any $\tau\in\mathbb{R}$ there exists a $t_l\in \{t_i\}_{i\in\mathbb{N}}$, such that 
	\begin{equation}
	|\tau-t_l|\leq \tau_{\textit{step}},\qquad t_l<\tau.
	\end{equation}
	We use the finite in time energy estimate~\eqref{eq: corollary quasilinear, eq 1}, of Theorem~\ref{thm: quasilinear, 1, local bootstrap}, and the decay of the energy of the sequence $\{t_i\}$, see~\eqref{eq: proof: thm 2, exp decay, eq 3}, and obtain the following 
	\begin{equation}
	E_{\mathcal{G},k}[\psi](\tau)\leq C(k,M,\Lambda,A_{[k+1]},B_{[k+1]}) E_{\mathcal{G},k}[\psi](t_l)\leq C(k,M,\Lambda,A_{[k+1]},B_{[k+1]}) \left(\frac{1}{L}\right)^{l}E_{\mathcal{G},k}[\psi](0).
	\end{equation}
	Now, by using one final time the finite in time energy estimate~\eqref{eq: corollary quasilinear, eq 1}, of Theorem~\ref{thm: quasilinear, 1, local bootstrap}, we conclude the desired inequality, by noting 
	\begin{equation}\label{eq: proof: thm 2, exp decay, eq 5}
		\begin{aligned}
			E_{\mathcal{G},k}[\psi](\tau)	&\leq C(k,M,\Lambda,A_{[k+1]},B_{[k+1]}) \left(\frac{1}{L}\right)^{l}E_{\mathcal{G},k}(0)\\
			&	\leq C(k,M,\Lambda,A_{[k+1]},B_{[k+1]}) e^{-l\: log L}E_{\mathcal{G},k}[\psi](0)\\
			&	\leq C(k,M,\Lambda,A_{[k+1]},B_{[k+1]}) e^{-c_{\text{d}} \tau}E_{\mathcal{G},k}[\psi](0),
		\end{aligned}
	\end{equation}
	for $c_{\text{d}}(K,M,\Lambda,A_{[k+1]},B_{[k+1]})=\frac{\log L}{\tau_{\textit{step}}}$, where in the last inequality we utilized that  
	\begin{equation}
	t_l=(t_l-t_{l-1})+(t_{l-1}-t_{l-2})+(t_{l-2}-t_{l-3})+\dots+(t_2-t_1)+t_1= \tau_{\textit{step}} l
	\end{equation}
	so 
	\begin{equation}
	e^{-l\log L\: }\leq e^{-\log L\frac{t_l}{\tau_{\textit{step}}}}\leq e^{\log L-\frac{\log L}{\tau_{\textit{step}}}\tau},
	\end{equation}
	which concludes inequality~\eqref{eq: thm 2, exp decay, eq 1}.

	The exponential decay of the lower order energy $E_{k+1}[\psi](\tau)$ follows by similar considerations, from inequality~\eqref{eq: proof: thm 2, exp decay, eq 3..2}, which concludes~\eqref{eq: thm 2, exp decay, eq 2}, for a certain constant $c_{\text{d}}(k,M,\Lambda,A_{[k+1]},B_{[k+1]})>0$.

	Note that for the top order energy $E_{k+2}[\psi](\tau)$, we use~\eqref{eq: corollary quasilinear, eq 1.2} of Theorem~\ref{thm: quasilinear, 1, local bootstrap} and~\eqref{eq: proof: thm 2, exp decay, eq 3}, and we only obtain 
	\begin{equation}
		\begin{aligned}
			E_{k+2}[\psi](\tau)	&	\leq C(k,M,\Lambda,A_{[k+1]},B_{[k+1]}) e^{\frac{\log(C_{\textit{wp}})}{\tau_{\textit{step}}}\tau}E_{k+2}[\psi](0),\\
		\end{aligned}
	\end{equation}
	which concludes~\eqref{eq: thm 2, exp decay, eq 3}, for a constant $c_{\text{g}}(k,M,\Lambda,A_{[k+1]},B_{[k+1]})=\frac{\log(C_{\textit{wp}})}{\tau_{\textit{step}}}>0$, where note that $C_{\textit{wp}}>1$, see Proposition~\ref{prop: well posedness quasilinear}.

	The proof of Theorem~\ref{thm: quasilinear, 2, exp decay} is thus complete.
\end{proof}

\section{The semilinear equation on Schwarzschild--de~Sitter}\label{sec: semilinear}

We here present a simplified proof for the global non-linear stability for the semilinear case 
\begin{equation}\label{eq: semilinear}
\Box_{g_{M,\Lambda}}\psi=\partial\psi\cdot\partial\psi,
\end{equation}
with $\partial\psi\cdot\partial\psi=a^{ij}\partial_i\psi\partial_j\psi$, where $a^{ij}$ are sufficiently regular components of a smooth tensor
\begin{equation}
	a:T\mathcal{M}\times T\mathcal{M}\rightarrow\mathbb{R}
\end{equation}
in Cartesian coordinates, see Section \ref{subsec: semilinear derivatives}. The results here give a more direct proof and a better regularity for the initial data than the relevant Theorems \ref{thm: quasilinear, 1, local bootstrap}, \ref{thm: quasilinear, 2, exp decay} that treat the more general quasilinear wave equation.

\begin{remark}
Note that in this Section we study $D(\tau_1,\tau_2)$ instead of $D_\delta(\tau_1,\tau_2)$, as $\delta=0$ here. One may easily a posteriori extend the results to $D_\delta(\tau_1,\tau_2)$. See Remark~\ref{rem: appendix: local well posedness, rem 1}.
\end{remark}

We give the statement of our Theorem, the analogue of Theorem \ref{thm: quasilinear, 1, local bootstrap}, on a fixed large time domain.

\begin{customTheorem}{1$^\prime$}\label{thm: semilinear, local bootstrap}
	Let $k\geq 7$ and let the tensor $a$ be as in Section \ref{subsec: semilinear derivatives}. 
	
	Then, there exists a $\tau_{\textit{step}}(k,M,\Lambda,A_{[k-1]})>0$ sufficiently large and there exists
	\begin{equation}
	\epsilon(\tau_{\textit{step}},k,M,\Lambda,A_{[k-1]})>0
	\end{equation}
	sufficiently small, where for $A_{[k]}$ see Section \ref{subsec: semilinear derivatives}, such that if we take initial data for~\eqref{eq: semilinear} on $\{\tau=\tau_1\}$ with
	\begin{equation}
	E_{\mathcal{G},k}[\psi](\tau_1)\leq \epsilon,\qquad E_{k}[\psi](\tau_1)<\infty
	\end{equation}
	and 
	\begin{equation}
	\tau_2=\tau_1+\tau_{\textit{step}},
	\end{equation}
	then there exists a unique $H^k$ solution in $D_\delta(\tau_1,\tau_2)$ to the semilinear wave equation~\eqref{eq: semilinear}, and the following inequality holds
	\begin{equation}\label{eq: thm: semilinear, local bootstrap, eq 1}
	\begin{aligned}
	E_{\mathcal{G},k}[\psi](\tau^\prime)    &   \leq \frac{1}{2}  E_{\mathcal{G},k}[\psi](\tau_1),\\
	\end{aligned}
	\end{equation}
	for all $\tau^\prime\in[\tau_1,\tau_2]$. 
\end{customTheorem}

A Corollary of Theorem \ref{thm: semilinear, local bootstrap}, which is also the analogue of Theorem \ref{thm: quasilinear, 2, exp decay}, is the following global existence and exponential decay result.

\begin{customTheorem}{2$^\prime$}\label{thm: semilinear, exp decay}
	Let $k\geq 7$ and let the tensor $a$ be as in Section \ref{subsec: semilinear derivatives}. Then, there exist positive constants
	\begin{equation}
		c_{\text{d}}(k,M,\Lambda,A_{[k-1]}),\qquad C(k,M,\Lambda,A_{[k-1]})>0
	\end{equation}
	 such that the following holds. 
	
	There exists an $\epsilon>0$ sufficiently small, such that for 
	\begin{equation}
	E_{\mathcal{G},k}[\psi](0)\leq \epsilon,\qquad E_k[\psi](0)<\infty
	\end{equation}
	the solution of the semilinear wave equation~\eqref{eq: semilinear} exists globally on $D(0,\infty)$ and the following holds
	\begin{equation}
	E_{\mathcal{G},k}[\psi](\tau)\leq C e^{-c_{\text{d}}\tau} E_{\mathcal{G},k}[\psi](0),
	\end{equation}
	for all $\tau\geq 0$.
	
\end{customTheorem}

Before proving Theorem~\ref{thm: semilinear, local bootstrap} we need the following energy estimate on a fixed spacetime slab. This is the analogue of Proposition \ref{prop: well posedness quasilinear}. 

\begin{proposition}\label{prop: semilinear}
	Let $k\geq 5$ and let the tensor $a$ be as in Section \ref{subsec: semilinear derivatives}. There exists a positive constant
	\begin{equation}
		C(k,M,\Lambda,A_{[k-1]})>0,
	\end{equation}
	where for $A_{[k-1]}$ see Section \ref{subsec: semilinear derivatives}, such that the following holds. 
	
	 Let $\psi$ be a solution of the semilinear wave equation~\eqref{eq: semilinear}, on a Schwarzschild--de~Sitter background $D(\tau_1,\tau_2)$ for $\tau_1<\tau_2$. Then, if the following holds  
	 \begin{equation}\label{eq: prop: semilinear, eq 0}
	 \sup_{D(\tau_1,\tau_2)}\sum_{1\leq i\leq \floor{\frac{k-1}{2}}+1}\sum_{\partial\in\{\partial_{\bar{t}},\partial_r,\Omega_1,\Omega_2,\Omega_3\}}|\partial^i\psi|\leq \sqrt{\epsilon}
	 \end{equation}
	 for a sufficiently small $\epsilon>0$, then we obtain 
	\begin{equation}\label{eq: thm semilinear, eq 1}
	\begin{aligned}
	E_{\mathcal{G},k}[\psi](\tau_2)+\int_{\tau_1}^{\tau_2}d\tau E_{\mathcal{G},k}[\psi](\tau)\leq & \:C E_{\mathcal{G},k}[\psi](\tau_1)	+ C\left(\sup_{\tau_1\leq\tau^\prime\leq\tau_2}E_{k-1}[\psi](\tau^\prime)\right) \left(\int_{\tau_1}^{\tau_2}d\tau E_{\mathcal{G},k}[\psi](\tau)\right),
	\end{aligned}
	\end{equation}
	for all $0\leq\tau_1\leq\tau_2$. 
\end{proposition}

\begin{proof}[\textbf{Proof of Proposition \ref{prop: semilinear}}]
	We will use the derivatives notation of Section \ref{subsec: derivatives notation}. Moreover, for the semilinear term $\partial\psi\cdot\partial\psi$ and for the definition of the tensor $a$ see Section \ref{subsec: semilinear derivatives}.

	Let 
	\begin{equation}
		F=\partial\psi\cdot\partial\psi
	\end{equation}
	then by~\eqref{eq: energy estimate of inhomogeneuos equation} we obtain
	\begin{equation}\label{eq: energy estimate for semilinear, 1}
		\begin{aligned}
			&	E_{\mathcal{G},k}[\psi](\tau_2)+\int_{\tau_1}^{\tau_2}d\tau E_{\mathcal{G},k}[\psi](\tau)\\
			&	\quad\quad\leq C(M,\Lambda) E_{\mathcal{G},k}[\psi](\tau_1)\\
			&	\quad\quad\quad\quad+C(M,\Lambda)\int\int_{D(\tau_1,\tau_2)} \sum_{i+j=0}^{k-2}\sum_\alpha (1-\mu)\left(\partial_{\bar{t}}^i(\Omega_\alpha)^j \mathcal{G} F\right)^2 + \sum_{0\leq i_1+i_2+i_3\leq k-2}\sum_\alpha\left(\partial_{\bar{t}}^{i_1}\partial_r^{i_2}\left(\Omega_\alpha\right)^{i_3} F\right)^2\\
			&	\quad\quad\quad\quad+C(M,\Lambda) \int_{\{\bar{t}=\tau_2\}} \sum_{0\leq i_1+i_2+i_3\leq k-3}\sum_\alpha(1-\mu)^{2i_3+1}\left(\partial_{\bar{t}}^{i_1}\Omega_\alpha^{i_2}\partial_r^{i_3}\mathcal{G}F\right)^2+\sum_{0\leq i_1+i_2+i_3\leq k-3}\sum_\alpha\left(\partial_{\bar{t}}^{i_1}\Omega_\alpha^{i_2}\partial_r^{i_3}F\right)^2.			 
		\end{aligned}
	\end{equation}
	
	For the terms 
	\begin{equation}
	\int\int_{D(\tau_1,\tau_2)} \sum_{i+j=0}^{k-2}\sum_\alpha (1-\mu)\left(\partial_{\bar{t}}^i(\Omega_\alpha)^j \mathcal{G} F\right)^2
	\end{equation}
	we have already computed in the proof of Proposition \ref{prop: quasilinear}, see inequality~\eqref{eq: energy estimate for quasilinear, 1.1}, that
	\begin{equation}\label{eq: energy estimate for semilinear, 1, 1}
		\begin{aligned}
			\int\int_{D(\tau_1,\tau_2)} \sum_{i+j=0}^{k-2}\sum_\alpha (1-\mu)\left(\partial_{\bar{t}}^i(\Omega_\alpha)^j \mathcal{G} F\right)^2	&	\leq C(k,M,\Lambda,A_{[k-1]}) \int_{\tau_1}^{\tau_2}d\tau E_{k-1}[\psi](\tau)E_{\mathcal{G},k}[\psi](\tau)\\
			&	\leq C(k,M,\Lambda,A_{[k-1]}) \sup_{\tau^\prime\in[\tau_1,\tau_2]}E_{k-1}[\psi](\tau^\prime)\int_{\tau_1}^{\tau_2}E_{\mathcal{G},k}[\psi](\tau)d\tau.
		\end{aligned}
	\end{equation}

	From the terms 
	\begin{equation}
	\int\int_{D(\tau_1,\tau_2)} \sum_{0\leq i_1+i_2+i_3\leq k-2}\sum_\alpha\left(\partial_{\bar{t}}^{i_1}\partial_r^{i_2}\left(\Omega_\alpha\right)^{i_3} F\right)^2
	\end{equation}
	on the right hand side of~\eqref{eq: energy estimate for semilinear, 1}, we only discuss the term 
	\begin{equation}\label{eq: energy estimate for semilinear, 1.2}
	\int\int_{D(\tau_1,\tau_2)}\sum_{i=0}^{k-2} \left(\partial_{\bar{t}}^iF\right)^2
	\end{equation}
	as the rest will admit the same bound. So, we bound the term~\eqref{eq: energy estimate for semilinear, 1.2}, by using the coarea formula~\eqref{eq: subsec: coarea, eq 1}, as follows
	\begin{equation}\label{eq: energy estimate for semilinear, 1, 2}
	\begin{aligned}
	\int\int_{D(\tau_1,\tau_2)}\sum_{i=0}^{k-2} \left(\partial_{\bar{t}}^i F\right)^2 &	=\int\int_{D(\tau_1,\tau_2)}\sum_{i=0}^{k-2}\left(\partial_{\bar{t}}^i(a^{\alpha\beta}\partial_\alpha\psi\partial_\beta\psi)\right)^2\\
	&	 \leq C(k,M,\Lambda,A_{[k-1]})  \int\int_{D(\tau_1,\tau_2)}\sum_{1\leq i+j\leq k-1}(\partial^i\psi)^2(\partial^j\psi)^2\\
	&  \leq C(k,M,\Lambda,A_{[k-1]})  \int_{\tau_1}^{\tau_2}d\tau\sum_{1\leq i+j\leq k-1,i\leq \floor{\frac{k-1}{2}}} \sup_{\{\bar{t}=\tau\}}(\partial^i\psi)^2\int_{\{\bar{t}=\tau\}}(\partial^j\psi)^2 \\	
	&  \leq C(k,M,\Lambda,A_{[k-1]})  \sup_{\tau_1\leq\tau^\prime\leq\tau_2} \left(E_{k-1}[\psi](\tau^\prime)\right)\int_{\tau_1}^{\tau_2}d\tau E_{\mathcal{G},k}[\psi](\tau),
	\end{aligned}
	\end{equation}
	where, in the last inequality, we used a Sobolev inequality, see Lemma~\ref{lem: sobolev estimate}, and the assumption that $k\geq 5$. 
	
	Therefore, by using~\eqref{eq: energy estimate for semilinear, 1, 1},~\eqref{eq: energy estimate for semilinear, 1, 2}, inequality~\eqref{eq: energy estimate for semilinear, 1} implies 
	\begin{equation}\label{eq: energy estimate for semilinear, 3}
	\begin{aligned}
	&	E_{\mathcal{G},k}[\psi](\tau_2)+\int_{\tau_1}^{\tau_2}d\tau E_{\mathcal{G},k}[\psi](\tau)\\
	&	\quad \leq  C(k,M,\Lambda) E_{\mathcal{G},k}[\psi](\tau_1)\\
	&	\quad\quad + C(k,M,\Lambda,A_{[k-1]}) \left(\sup_{\tau_1\leq\tau^\prime\leq\tau_2}E_{k-1}(\tau^\prime)\right) \left(\int_{\tau_1}^{\tau_2}d\tau E_{\mathcal{G},k}[\psi](\tau)\right)\\
	&	\quad\quad+C(k,M,\Lambda) \int_{\{\bar{t}=\tau_2\}} \sum_{0\leq i_1+i_2+i_3\leq k-3}\sum_\alpha(1-\mu)^{2i_3+1}\left(\partial_{\bar{t}}^{i_1}\Omega_\alpha^{i_2}\partial_r^{i_3}\mathcal{G}F\right)^2+\sum_{0\leq i_1+i_2+i_3\leq k-3}\sum_\alpha\left(\partial_{\bar{t}}^{i_1}\Omega_\alpha^{i_2}\partial_r^{i_3}F\right)^2,
	\end{aligned}
	\end{equation}
	where recall that $F=\partial\psi\cdot\partial\psi$. 
	
	Finally, we want to absorb the boundary terms of the right hand side at $\{\bar{t}=\tau_2\}$ by the relevant boundary term $E_{\mathcal{G},k}[\psi](\tau_2)$ of the left hand side. By using the smallness assumption~\eqref{eq: prop: semilinear, eq 0}, it is evident that we can appropriately distribute derivatives, in view of the restriction $k\geq 5$, to conclude that 
	\begin{equation}
		\begin{aligned}
			E_{\mathcal{G},k}[\psi](\tau_2)+\int_{\tau_1}^{\tau_2}d\tau E_{\mathcal{G},k}[\psi](\tau)\leq & \:C E_{\mathcal{G},k}[\psi](\tau_1)	+ C\left(\sup_{\tau_1\leq\tau^\prime\leq\tau_2}E_{k-1}[\psi](\tau^\prime)\right) \left(\int_{\tau_1}^{\tau_2}d\tau E_{\mathcal{G},k}[\psi](\tau)\right),
		\end{aligned}
	\end{equation}
	for a constant $C=C(k,M,\Lambda,A_{[k]})$. 
	
	The proof of Proposition \ref{prop: semilinear} is thus complete. 
\end{proof}

Now we are ready to prove Theorem~\ref{thm: semilinear, local bootstrap}.

\begin{proof}[\textbf{Proof of Theorem~\ref{thm: semilinear, local bootstrap}}]

Note that the semilinear wave equation~\eqref{eq: semilinear} is well posed in $H^k$, $k\geq 7$, by well-known arguments. The existence of the solution in $D(\tau_1,\tau_2)$, for $\tau_2$ defined below, will follow easily from the estimates that we will prove, so for convenience we here assume existence.

Moreover, we assume the bootstrap assumption 
\begin{equation}\label{eq: thm: semilinear, exp decay, eq -1}
\sup_{D(\tau_1,\tau_2)}\sum_{1\leq i\leq \floor{\frac{k-1}{2}}+1}\sum_{\partial\in\{\partial_{\bar{t}},\partial_r,\Omega_1,\Omega_2,\Omega_3\}}|\partial^i\psi|\leq C_{\text{b}} \sqrt{\epsilon}
\end{equation}
for any $\tau\geq 0$, for a constant $C_{\text{b}}(k,M,\Lambda,A_{[k-1]})>0$ to be determined later.

First, we want to prove that for $\tau_{\textit{step}}>0$ to be chosen later, there exists an $\epsilon(\tau_{\textit{step}},C_{\text{b}})>0$ sufficiently small such that if 
\begin{equation}
E_{\mathcal{G},k}[\psi](\tau_1)\leq \epsilon
\end{equation}
for some $\tau_1\geq 0$, then for 
\begin{equation}
	\tau_2=\tau_1+\tau_{\textit{step}}
\end{equation}
we obtain 
\begin{equation}\label{eq: thm: semilinear, exp decay, eq -0.1}
\begin{aligned}
E_{k-1}[\psi](\tau_2)	&	\leq \epsilon, \\
E_{\mathcal{G},k}[\psi](\tau_2)	&	\leq \frac{1}{2} E_{\mathcal{G},k}[\psi](\tau_1).\\
\end{aligned}
\end{equation}

Let $\tau_1\geq 0$, then in view of the bootstrap assumption~\eqref{eq: thm: semilinear, exp decay, eq -1}, and for $\epsilon(\tau_{\textit{step}},C_{\text{b}})$ we obtain the result of Proposition~\ref{prop: semilinear}, namely 
\begin{equation}\label{eq: thm: semilinear, exp decay, eq -0}
\begin{aligned}
E_{\mathcal{G},k}[\psi](\tau_2)+\int_{\tau_1}^{\tau_2}d\tau E_{\mathcal{G},k}[\psi](\tau)\leq & \:C E_{\mathcal{G},k}[\psi](\tau_1)	+ C\left(\sup_{\tau_1\leq\tau^\prime\leq\tau_2}E_{k-1}[\psi](\tau^\prime)\right) \left(\int_{\tau_1}^{\tau_2}d\tau E_{\mathcal{G},k}[\psi](\tau)\right)
\end{aligned}
\end{equation}
for some constant $C=C(k,M,\Lambda,A_{[k-1]})$. Now, for a sufficiently small $\epsilon(\tau_{\textit{step}})$ if we assume that 
\begin{equation}
E_{\mathcal{G},k}[\psi](\tau_1)\leq\epsilon,
\end{equation}
we also have that $E_{k-1}[\psi](\tau_1)\leq \epsilon$, and from~\eqref{eq: thm: semilinear, exp decay, eq -0} we conclude that there exists a constant 
\begin{equation}
C(k,M,\Lambda,A_{[k-1]})>0
\end{equation}
where for $A_{[k-1]}$ see Section \ref{subsec: semilinear derivatives}, such that 
\begin{equation}\label{eq: thm: semilinear, exp decay, eq 0}
E_{\mathcal{G},k}[\psi](\tau^\prime)\leq C(k,M,\Lambda,A_{[k-1]}) E_{\mathcal{G},k}[\psi](\tau_1)
\end{equation}
for all $\tau^\prime \in[\tau_1,\tau_2]$.

Moreover, by using again the result of Proposition \ref{prop: semilinear} (also see the relevant computation in the proof of Theorem~\ref{thm: quasilinear, 1, local bootstrap}), there exists a constant $C(k,M,\Lambda,A_{[k-1]})>0$ such that for any $\tau_{\textit{step}}>0$ sufficiently large, there exists an $\epsilon(\tau_{\textit{step}})>0$ sufficiently small, and a value
\begin{equation}
\tau_{1,2}\in\left[\frac{\tau_1+\tau_2}{2},\tau_2\right]
\end{equation}
such that, after taking
\begin{equation}
E_{\mathcal{G},k}[\psi](\tau_1)\leq  \epsilon 
\end{equation}
and recalling $E_{k-1}[\psi](\tau_1)\leq E_{\mathcal{G},k}[\psi](\tau_1)$,  we obtain	
\begin{equation}\label{eq: thm: semilinear, exp decay, eq 0.1}
E_{\mathcal{G},k}[\psi](\tau_{1,2})\leq \frac{1}{\tau_{\textit{step}}}C(k,M,\Lambda,A_{[k-1]}) E_{\mathcal{G},k}[\psi](\tau_1),
\end{equation}
where we used the mean value theorem. By using~\eqref{eq: thm: semilinear, exp decay, eq 0} and~\eqref{eq: thm: semilinear, exp decay, eq 0.1} we obtain
\begin{equation}\label{eq: thm: semilinear, exp decay, eq 1}
E_{\mathcal{G},k}[\psi](\tau_2)\leq \frac{1}{\tau_{\textit{step}}}C(k,M,\Lambda,A_{[k-1]})E_{\mathcal{G},k}[\psi](\tau_1)
\end{equation}
for $\tau_{\textit{step}}>0$ and $\epsilon(\tau_{\textit{step}},C_{\text{b}})>0$ sufficiently small, and for a different constant $C(k,M,\Lambda,A_{[k-1]})$. Now, note that 
\begin{equation}
E_{k-1}[\psi](\tau_2)\leq E_{\mathcal{G},k}[\psi](\tau_2)\leq\frac{1}{\tau_{\textit{step}}}C(k,M,\Lambda,A_{[k-1]}) E_{\mathcal{G},k}[\psi](\tau_1).
\end{equation}
Therefore, we choose a $\tau_{\textit{step}}$ sufficiently large, such that
\begin{equation}
\frac{1}{\tau_{\textit{step}}}C(k,M,\Lambda,A_{[k-1]})<\frac{1}{2}
\end{equation}
and we obtain 
\begin{equation}
\begin{aligned}
E_{k-1}[\psi](\tau_2)\leq E_{\mathcal{G},k}[\psi](\tau_1)&	\leq \epsilon,\\
E_{\mathcal{G},k}[\psi](\tau_2)	&	\leq \frac{1}{2} E_{\mathcal{G},k}[\psi](\tau_1).\\
\end{aligned}
\end{equation}

Of course our inequalities are still conditional on improving the bootstrap assumption~\eqref{eq: thm: semilinear, exp decay, eq -1}. For the purpose of improving constant $C_{\text{b}}>0$ of the bootstrap assumption \eqref{eq: thm: semilinear, exp decay, eq -1} we note from~\eqref{eq: thm: semilinear, exp decay, eq 0} that
\begin{equation}\label{eq: thm: semilinear, exp decay, eq 4}
E_{k-1}[\psi](\tau^\prime)\leq C(k,M,\Lambda,A_{[k-1]}) E_{\mathcal{G},k}[\psi](\tau_1)\leq  C(k,M,\Lambda,A_{[k-1]})\epsilon,
\end{equation}
for all $\tau^\prime \in[\tau_1,\tau_2]$. Now, since $k=7$ is the smallest integer such that 
\begin{equation}
\left(\floor{\frac{k-1}{2}}+1\right)+2\leq k-1
\end{equation}
we use a Sobolev inequality in~\eqref{eq: thm: semilinear, exp decay, eq 4} and prove 
\begin{equation}
\sup_{D(0,\tau)}\sum_{1\leq i\leq \floor{\frac{k-1}{2}}+1}\sum_{\partial\in\{\partial_{\bar{t}},\partial_r,\Omega_1,\Omega_2,\Omega_3\}}|\partial^i\psi|\leq C(k,M,\Lambda,A_{[k-1]}) \sqrt{\epsilon}
\end{equation}
which improves the bootstrap for a sufficiently large $C_{\text{b}}\gg C(k,M,\Lambda,A_{[k-1]})$.

This concludes the proof the Theorem. 
\end{proof}

Now, we are ready to prove Theorem~\ref{thm: semilinear, exp decay}, which is in fact a Corollary of Theorem \ref{thm: semilinear, local bootstrap}.

\begin{proof}[\textbf{Proof of Theorem~\ref{thm: semilinear, exp decay}}]

	We want to prove that for sufficiently small initial data
	\begin{equation}
		E_{\mathcal{G},k}[\psi](0)\leq \epsilon,
	\end{equation}
	and for the following sequence of ascending real numbers 
	\begin{equation}
		t_0=\tau_{\textit{step}},\qquad t_{i+1}-t_i=\tau_{\textit{step}},
	\end{equation}
	the solution exists in $D(0,t_i)$ and we obtain that 
	\begin{equation}\label{eq: thm: semilinear, exp decay, eq 2}
		\begin{aligned}
			E_{k-1}[\psi](t_i)	&	\leq C(k,M,\Lambda)\epsilon, \\
			E_{\mathcal{G},k}[\psi](t_i)	&	\leq \left(\frac{1}{2}\right)^i E_{\mathcal{G},k}[\psi](0),\\
		\end{aligned}
	\end{equation}
	for all $i\in\mathbb{N}$, where for the constant $C(k,M,\Lambda)$ see Remark \ref{rem: energies}.
	
	Note that the inequalities~\eqref{eq: thm: semilinear, exp decay, eq 2} hold for $t_0$. Therefore, for the purpose of using induction we assume the inductive step that~\eqref{eq: thm: semilinear, exp decay, eq 2} hold for $t_i$. Then, from Theorem~\ref{thm: semilinear, local bootstrap}, specifically from inequality~\eqref{eq: thm: semilinear, local bootstrap, eq 1} and the inductive step we obtain 
	\begin{equation}
		\begin{aligned}
			E_{k-1}[\psi](t_{i+1})	&	\leq C(k,M,\Lambda)\epsilon, \\
			E_{\mathcal{G},k}[\psi](t_{i+1})	&	\leq \left(\frac{1}{2}\right)^{i+1} E_{\mathcal{G},k}[\psi](0),\\
		\end{aligned}
	\end{equation}
	and moreover conclude that the solution exists in $D(0,t_{i+1})$.

	Therefore, we have concluded global existence and the inequalities~\eqref{eq: thm: semilinear, exp decay, eq 2} for all $i\in\mathbb{N}$.

	Now, we want to conclude the result of the Theorem, namely we want to prove that for 
	\begin{equation}
		E_{\mathcal{G},k}[\psi](0)\leq \epsilon
	\end{equation}
	we obtain 
	\begin{equation}\label{eq: thm: semilinear, exp decay, eq 3}
		\begin{aligned}
			E_{\mathcal{G},k}[\psi](\tau)	&	\leq C e^{-c_\text{d}\tau}E_{\mathcal{G},k}[\psi](0), \\
		\end{aligned}
	\end{equation}
	for constants $c_{\text{d}}(k,M,\Lambda,A_{[k-1]}),\:  C(k,M,\Lambda,A_{[k-1]})>0$. With initial data 
	\begin{equation}
		E_{\mathcal{G},k}[\psi](0)\leq \epsilon
	\end{equation}
	for a sufficiently small $\epsilon>0$, we use equation~\eqref{eq: thm: semilinear, exp decay, eq 2} and the local in time inequality~\eqref{eq: thm: semilinear, local bootstrap, eq 1} of Theorem \ref{thm: semilinear, local bootstrap}, to immediately conclude~\eqref{eq: thm: semilinear, exp decay, eq 3}.

	We conclude the Theorem. 	
\end{proof}

\appendix

\section{Proof of the Cauchy stability result of Proposition \ref{prop: well posedness quasilinear}}\label{appendix: local well posedness}

In this Section we prove the Cauchy stability statement of Proposition \ref{prop: well posedness quasilinear}. This proof will require using geometric estimates with respect to the metric $g(\nabla\psi)$, so as to not lose derivatives. We will use some standard definitions and results of the linear theory. Therefore, we first need to discuss some standard notions.

\subsection{Energy momentum tensor and the divergence theorem of a metric \texorpdfstring{$g$}{g}}\label{subsec: energy momentum tensor}

Let $g$ be a sufficiently regular Lorentzian metric, and $\psi$ a sufficiently regular function. Then, we define the energy momentum tensor with respect to that metric as 
\begin{equation}
\mathbb{T}(g)[\psi]=d\psi \otimes d\psi-\frac{1}{2}g |\nabla\psi|_g^2,
\end{equation}
where $|\nabla\psi|_g^2\:\dot{=}\:g^{ab}\partial_a\psi\partial_b\psi$.

Note the following Proposition
\begin{proposition}\label{prop: energy identity, prop 1, SdS}
	Let $g$ be a sufficiently regular Lorentzian metric on $\mathcal{M}_\delta$ and $\delta (M,\Lambda)\geq 0$ sufficiently small. Also, let $\psi$ satisfy the equation 
	\begin{equation}
	\Box_g \psi=F.
	\end{equation}
	Then, we obtain the following identity
	\begin{equation}\label{eq: lem energy identity, eq 1}
	\begin{aligned}
	&   \int_{\{\bar{t}=\tau_2\}\cap D_\delta(\tau_1,\tau_2)}\mathbb{T}(g)(X,n)[\psi]+\int_{\mathcal{H}^+_{\delta}\cap D_\delta(\tau_1,\tau_2)}\mathbb{T}(g)(X,n)[\psi]+\int_{\bar{\mathcal{H}}^+_\delta\cap D_\delta(\tau_1,\tau_2)}\mathbb{T}(g)(X,n)[\psi]   \\
	&   +\int\int_{D_\delta(\tau_1,\tau_2)}K^X(g) = \int_{\{\bar{t}=\tau_1\}\cap D_\delta(\tau_1,\tau_2)} \mathbb{T}(g)(X,n)[\psi]+\int\int_{D_\delta(\tau_1,\tau_2)}\textit{Err}^X[\psi]
	\end{aligned}
	\end{equation}
	for any $0\leq \tau_1\leq \tau_2$ and for any vector field $X$, where 
	\begin{equation}
	K^X(g)= \frac{1}{2}\:^{(X)}\pi_{\mu\nu}(g)\mathbb{T}^{\mu\nu}(g)[\psi], \qquad \textit{Err}^X[\psi] = -(X\psi) F,
	\end{equation}
	where note that \textit{the raised indices are with the metric} $g$.

	The volume forms on the spacetime domains are to be understood with respect to the metric $g$, and the volume forms on the hypersurfaces are to be understood as the ones induced by the volume form of $g$. Also 
	\begin{equation}
	n
	\end{equation}
	is to be understood as the unit normal of each hypersurface with respect to $g$. 
	
	Note that in the case $g=g_{M,\Lambda},\:\delta\geq 0$, we have explicitly computed the volume forms and the normals, with which the divergence theorem is to be understood, in Section \ref{subsec: volume forms}. 
\end{proposition}

\subsection{The redshift vector field \texorpdfstring{$N$}{N} on the Schwarzschild background $\mathring{g}$}

We present the following redshift Lemmata, see the Lecture notes~\cite{DR5}.

\begin{lemma}\label{lem: redshift, lem 1}
	Let $q(M,\Lambda)>0$ be sufficiently small. Then, for all $\delta(M,\Lambda)>0$ sufficiently small there exist a timelike vector field
	\begin{equation}
	N,
	\end{equation}
	on $D_\delta(0,\infty)$, such that the following hold
	\begin{equation}\label{eq: redsfhift, lem 1, eq 1}
	\begin{aligned}
	K^N(\mathring{g})	&	\geq C(M,\Lambda) \mathbb{T}(\mathring{g})(N,N)[\psi],\:\textit{ in } \{r_+-\delta\leq r\leq r_++q\}\cup\{\bar{r}_+-q\leq r\leq\bar{r}_++\delta\} \\
	N	&	=\partial_{\bar{t}},\: \textit{ in }D_\delta(0,\infty)\setminus \{r_+-\delta\leq r\leq r_++q\}\cup\{\bar{r}_+-q\leq r\leq\bar{r}_++\delta\}.
	\end{aligned}
	\end{equation}
	where for $K^N(\mathring{g})$ see the divergence theorem in Proposition \ref{prop: energy identity, prop 1, SdS}.
\end{lemma}

The following Lemma, is a higher order manifestation of the redshift effect.

\begin{lemma}\label{lem: redshift, lem 2}
	For any $k\geq 0$ there exist positive constants 
	\begin{equation}
	\kappa_k >0,\qquad \bar{\kappa}_k>0
	\end{equation}
	such that for the equation
	\begin{equation}
	\Box_{\mathring{g}}\psi=F,
	\end{equation}
	we obtain that the following holds on $\mathcal{H}^+$
	\begin{equation}\label{eq: redshift, lem 2, eq 1}
	\Box_{\mathring{g}}N^k\psi=\kappa_k N^{k+1}\psi+\sum_{1\leq i\leq 5}\quad \sum_{0\leq \sum_i m_i\leq |m|\leq k+1,\: m_5\leq k} c_m \Omega_1^{m_1}\Omega_2^{m_2}\Omega_3^{m_3} \partial_{\bar{t}}^{m_4}N^{m_5}\psi+N^k F,
	\end{equation}
	where $c_m$ are smooth functions on $\mathcal{H}^+$, and the following holds on $\bar{\mathcal{H}}^+$
	\begin{equation}\label{eq: redshift, lem 2, eq 2}
	\Box_{\mathring{g}}N^k\psi=\bar{\kappa}_k N^{k+1}\psi+\sum_{1\leq i\leq 5}\quad \sum_{0\leq \sum_i m_i\leq |m|\leq k+1,\: m_5\leq k} \bar{c}_m \Omega_1^{m_1}\Omega_2^{m_2}\Omega_3^{m_3} \partial_{\bar{t}}^{m_4}N^{m_5}\psi+N^k F,
	\end{equation}
	where $\bar{c}_m$ are smooth functions on $\bar{\mathcal{H}}^+$.
\end{lemma}

\subsection{An estimate for metrics close to \texorpdfstring{$\mathring{g}$}{g}}

Note the following lemma
\begin{lemma}\label{lem: metric close to SdS} 
	Let $g(v)=\mathring{g}+h(v)$ be a Lorentzian metric that belongs in the class of Section~\ref{subsec: assumption}, where $v\in \Gamma\left(T\mathcal{M}_\delta\right)$ is considered fixed, and let $\delta(M,\Lambda)>0$ be sufficiently small.

	Then there exists an $\epsilon(\delta)>0$ such that for
	\begin{equation}
		|h_{ij}(v)|,\qquad |h^{ij}(v)|\leq \sqrt{\epsilon}
	\end{equation}
	in Cartesian coordinates,  the hypersurface $\{\bar{t}=\tau\}$, is a Cauchy hypersurface, with respect to $g(v)$, in $D_\delta(0,\infty)$. Moreover, the unit normal $n$ of $\{\bar{t}=c\}$ is timelike on $D_\delta(0,\infty)$ with respect to $g(v)$ and the redshift vector field, see Lemma~\ref{lem: redshift, lem 1}, is timelike with respect to $g(v)$, as well. Furthermore, the hypersurfaces 
	\begin{equation}
	\mathcal{H}^+_{\delta},\qquad \bar{\mathcal{H}}^+_\delta
	\end{equation}
	are spacelike hypersurfaces with respect to the metric $g(v)$ in $D_\delta(0,\infty)$.

	Finally, we obtain
	\begin{equation}\label{eq: lem: metric close to SdS, eq 1} 
	\begin{aligned}
		\int_{\{\bar{t}=\tau\}}\mathbb{T}(g(v))(n,n)[\psi] dg(v)_{\{\bar{t}=\tau\}}&	\sim \int_{\{\bar{t}=\tau\}}\mathbb{T}(\mathring{g})(\mathring{n},\mathring{n})[\psi] d\mathring{g}_{\{\bar{t}=\tau\}},\\
		\int_{\mathcal{H}^+_\delta}\mathbb{T}(g(v))(N,n)[\psi]dg(v)_{\mathcal{H}^+_\delta}	&	\sim \int_{\mathcal{H}^+_\delta}\frac{d\mathring{g}_{\mathcal{H}^+_\delta}}{\sqrt{|1-\mu|}}\left(\delta(\partial_r\psi)^2+(\partial_{\bar{t}}\psi)^2+|\slashed{\nabla}\psi|^2\right),\\
		\int_{\bar{\mathcal{H}}^+_\delta}\mathbb{T}(g(v))(N,n)[\psi]dg(v)_{\bar{\mathcal{H}}^+_\delta}	&	\sim \int_{\bar{\mathcal{H}}^+_\delta}\frac{d\mathring{g}_{\bar{\mathcal{H}}^+_\delta}}{\sqrt{|1-\mu|}}\left(\delta(\partial_r\psi)^2+(\partial_{\bar{t}}\psi)^2+|\slashed{\nabla}\psi|^2\right),
	\end{aligned}
	\end{equation}
	where the constants only depend on the black hole parameters, and $dg(v)_{\{\bar{t}=\tau\}},dg(v)_{\mathcal{H}^+_\delta},dg(v)_{\bar{\mathcal{H}}^+_\delta}$ are the induced volume form of the spacetime volume form of the metric $g(v)$ on the hypersurfaces $\{\bar{t}=\tau\},\mathcal{H}^+_\delta,\bar{\mathcal{H}}^+_\delta$, respectively. For the energy momentum tensor $\mathbb{T}$ see Section \ref{subsec: energy momentum tensor}, for the redshift vector field $N$ see Lemma~\ref{lem: redshift, lem 1}. 
\end{lemma}
\begin{proof}
	The proof of this Lemma is direct by noting that the energy momentum tensor of the metric $g(v)$ is
	\begin{equation}
		\mathbb{T}(g(v))(X,Y)[\psi]=X\psi Y\psi-\frac{1}{2}g(v)(X,Y)|\nabla \psi|_{g(v)}^2,
	\end{equation}
for any two smooth vector fields $X,Y$. 

Specifically, by taking
\begin{equation}
	|h_{ij}(v)|,\quad |h^{ij}(v)|\leq \sqrt{\epsilon}
\end{equation}
sufficiently small, we conclude the causal behaviour of the hypersurfaces mentioned, and by pointwise estimates on the integrands we conclude 
\begin{equation}
	\begin{aligned}
		\int_{\{\bar{t}=\tau\}}\mathbb{T}(g(v))(n,n)[\psi] dg(v)_{\{\bar{t}=\tau\}}&	\sim \int_{\{\bar{t}=\tau\}}\mathbb{T}(\mathring{g})(\mathring{n},\mathring{n})[\psi] d\mathring{g}_{\{\bar{t}=\tau\}},\\
		\int_{\mathcal{H}^+_\delta}\mathbb{T}(g(v))(N,n)[\psi]dg(v)_{\mathcal{H}^+_\delta}	&	\sim \int_{\mathcal{H}^+_\delta}\frac{d\mathring{g}_{\mathcal{H}^+_\delta}}{\sqrt{|1-\mu|}}\left(\delta(\partial_r\psi)^2+(\partial_{\bar{t}}\psi)^2+|\slashed{\nabla}\psi|^2+\sqrt{\epsilon(\delta)}\left((\partial_r\psi)^2+(\partial_{\bar{t}}\psi)^2+|\slashed{\nabla}\psi|^2\right)\right),\\
		\int_{\bar{\mathcal{H}}^+_\delta}\mathbb{T}(g(v))(N,n)[\psi]dg(v)_{\bar{\mathcal{H}}^+_\delta}	&	\sim \int_{\bar{\mathcal{H}}^+_\delta}\frac{d\mathring{g}_{\bar{\mathcal{H}}^+_\delta}}{\sqrt{|1-\mu|}}\left(\delta(\partial_r\psi)^2+(\partial_{\bar{t}}\psi)^2+|\slashed{\nabla}\psi|^2+\sqrt{\epsilon(\delta)}\left((\partial_r\psi)^2+(\partial_{\bar{t}}\psi)^2+|\slashed{\nabla}\psi|^2\right)\right).
\end{aligned}
\end{equation}
Now, for $\epsilon(\delta)\ll \delta$ we conclude~\eqref{eq: lem: metric close to SdS, eq 1} .
\end{proof}

\subsection{Elliptic estimates and estimates on \texorpdfstring{$\mathcal{H}^+_\delta,\bar{\mathcal{H}}^+_\delta$}{g} for the quasilinear wave equation}\label{subsec: elliptic estimates}

Note the following elliptic estimate 
\begin{lemma}\label{lem: subsec: elliptic estimates, lem 1}
	Let $k\geq 4$ and let $\psi$ satisfy the quasilinear wave equation~\eqref{eq: quasilinear} in $D_\delta(\tau_1,\tau_2)$ for $\tau_1\leq \tau_2$ and for a sufficiently small $\delta(M,\Lambda)>0$, where for the tensors $a,h$ see the Sections \ref{subsec: semilinear derivatives}, \ref{subsec: assumption} respectively. Then, for any $r_+<r_0<r_1<\bar{r}_+$, there exist constants
	\begin{equation}
	C(r_0,r_1,k,M,\Lambda,A_{[k+1]},B_{[k+1]}),\qquad C(k,M,\Lambda,A_{[k+1]},B_{[k+1]})>0
	\end{equation}
	independent of $\delta$, where for $A_{[k+1]},B_{[k+1]}$ see Sections \ref{subsec: semilinear derivatives}, \ref{subsec: assumption} respectively, and there exists an $\epsilon=\epsilon(\delta,k,M,\Lambda,A_{[k+1]},B_{[k+1]})$ sufficiently small such that the following holds. 
	
	If  
	\begin{equation}\label{eq: lem: subsec: elliptic estimates, lem 1, eq 1}
		\sup_{D_\delta(\tau_1,\tau_2)}\sum_{1\leq j\leq k-1}\sum_{\partial\in\{\partial_{\bar{t}},\partial_r,\Omega_1,\Omega_2,\Omega_3\}}|\partial^j\psi|\leq\sqrt{\epsilon}
	\end{equation}
	holds, we obtain that
\begin{equation}\label{eq: lem: subsec: elliptic estimates, lem 1, eq 2}
	\begin{aligned}
		&	\int_{\{\bar{t}=\tau^\prime\}}\sum_{1\leq i_1+i_2+i_3\leq j}\sum_\alpha\left(\partial_{\bar{t}}^{i_1}\partial_r^{i_2}\Omega_\alpha^{i_3}\psi\right)^2 d\mathring{g}_{\{\bar{t}=\tau^\prime\}}\leq C \int_{\{\bar{t}=\tau\}} \sum_{1\leq i\leq j-1} \mathbb{T}(\mathring{g})(N,n)[N^i\psi] d\mathring{g}_{\{\bar{t}=\tau\}},\\
		&	\int_{\{\bar{t}=\tau^\prime\}\cap [r_0,r_1]} \sum_{1\leq i_1+i_2+i_3\leq j}\sum_\alpha\left(\partial_{\bar{t}}^{i_1}\partial_r^{i_2}\Omega_\alpha^{i_3}\psi\right)^2 d\mathring{g}_{\{\bar{t}=\tau\}}\leq C(r_0,r_1) \int_{\{\bar{t}=\tau^\prime\}\cap[r_0,r_1]}\sum_{1\leq i_1+i_2\leq j}\sum_{\alpha} \left(\partial_{\bar{t}}^{i_1}\Omega_\alpha^{i_2}\psi\right)^2 d\mathring{g}_{\{\bar{t}=\tau\}},\\
		&  \int_{\mathcal{H}^+_\delta}\frac{d\mathring{g}_{\mathcal{H}_\delta}}{\sqrt{|1-\mu|}}\left( \delta \sum_{1\leq i_1+i_2+i_3\leq j,i_2\geq 1}\sum_\alpha\left(\partial_{\bar{t}}^{i_1}\partial_r^{i_2}\Omega_\alpha^{i_3}\psi\right)^2 +\sum_{1\leq i_1+i_2\leq j}\sum_\alpha\left(\partial_{\bar{t}}^{i_1}\Omega_\alpha^{i_2}\psi\right)^2\right)	\leq C \int_{\mathcal{H}^+_\delta}\sum_{1\leq i\leq j-1}\mathbb{T}(\mathring{g})(N,n)[N^i\psi] d\mathring{g}_{\mathcal{H}_\delta},\\
		&  \int_{\bar{\mathcal{H}}^+_\delta}\frac{d\mathring{g}_{\mathcal{H}_\delta}}{\sqrt{|1-\mu|}}\left( \delta \sum_{1\leq i_1+i_2+i_3\leq j,i_2\geq 1}\sum_\alpha\left(\partial_{\bar{t}}^{i_1}\partial_r^{i_2}\Omega_\alpha^{i_3}\psi\right)^2 +\sum_{1\leq i_1+i_2\leq j}\sum_\alpha\left(\partial_{\bar{t}}^{i_1}\Omega_\alpha^{i_2}\psi\right)^2\right)	\leq C \int_{\bar{\mathcal{H}}^+_\delta}\sum_{1\leq i\leq j-1}\mathbb{T}(\mathring{g})(N,n)[N^i\psi] d\mathring{g}_{\mathcal{H}_\delta},
	\end{aligned}
\end{equation}
for $1\leq j\leq k+2$.
\end{lemma}
\begin{proof}
The proof of this Lemma comes from the estimates of Lemma~\ref{lem: metric close to SdS}.

\textit{For the first inequality of~\eqref{eq: lem: subsec: elliptic estimates, lem 1, eq 2}}, for $j=1$, we use the smallness~\eqref{eq: lem: subsec: elliptic estimates, lem 1, eq 1}, and Lemma~\ref{lem: metric close to SdS}. We prove the estimate for all orders $j$, by first noting that the energy momentum tensor of $g(\nabla\psi)$ is 
\begin{equation}\label{eq: proof: lem: subsec: elliptic estimates, lem 1, eq 1}
	\mathbb{T}\left(g(\nabla\psi)\right)(X,n)[\tilde{X}^i\psi]=\mathbb{T}\left(\mathring{g}\right)(X,n)[\tilde{X}^i\psi]-\frac{1}{2}h(X,n,\nabla\psi)\left|\nabla\tilde{X}^i\psi\right|^2_{g(\nabla\psi)},
\end{equation}
for any two smooth vector field $X,\tilde{X}$. Therefore, for $\epsilon(\delta)>0$ sufficiently small, we use elliptic estimates on the difference $\Box_{\mathring{g}}\psi=\left(\Box_{\mathring{g}}-\Box_{g(\nabla\psi)}\right)\psi +\partial\psi\cdot\partial\psi$ and conclude. We also used that $k\geq 4$.

\textit{For the second inequality of~\eqref{eq: lem: subsec: elliptic estimates, lem 1, eq 2}}, we obtain the estimate by arguing as above and by taking $\epsilon(\delta)>0$ small. We also used that $k\geq 4$.

\textit{For the the last two inequalities of~\eqref{eq: lem: subsec: elliptic estimates, lem 1, eq 2}}, on the hypersurfaces $\mathcal{H}^+_\delta,\bar{\mathcal{H}}^+_\delta$, to get the estimate for $j=1$, we use the smallness~\eqref{eq: lem: subsec: elliptic estimates, lem 1, eq 1}, and Lemma~\ref{lem: metric close to SdS}. We prove the analogous estimate for all orders $j$, for a sufficiently small $\epsilon(\delta)>0$ by noting the form of the energy momentum tensor~\eqref{eq: proof: lem: subsec: elliptic estimates, lem 1, eq 1}, and by using elliptic estimates on the difference $\Box_{\mathring{g}}\psi=\left(\Box_{\mathring{g}}-\Box_{g(\nabla\psi)}\right)\psi +\partial\psi\cdot\partial\psi$. We also used that $k\geq 4$.
\end{proof}

Finally, note the following estimate
\begin{lemma}\label{lem: subsec: elliptic estimates, lem 2}
Let $k\geq 4$. Let $\psi$ satisfy the quasilinear wave equation~\eqref{eq: quasilinear} in $D_\delta(\tau_1,\tau_2)$ for a sufficiently small $\delta(k,M,\Lambda,A_{[k+1]},B_{[k+1]})>0$, where for the tensors $a,h$ see the Sections \ref{subsec: semilinear derivatives}, \ref{subsec: assumption} respectively

Then, there exists a constant
\begin{equation}
C(k,M,\Lambda,A_{[k+1]},B_{[k+1]})>0
\end{equation}
independent of $\delta$, and an $\epsilon=\epsilon(\delta,k,M,\Lambda,A_{[k+1]},B_{[k+1]})>0$ sufficiently small such that the following holds. If  
\begin{equation}\label{eq: lem: subsec: elliptic estimates, lem 1, eq 1.1}
	\sup_{D_\delta(\tau_1,\tau_2)}\sum_{1\leq j\leq k-1}\sum_{\partial\in\{\partial_{\bar{t}},\partial_r,\Omega_1,\Omega_2,\Omega_3\}}|\partial^j\psi|\leq\sqrt{\epsilon},
\end{equation}
then we obtain that
\begin{equation}
\begin{aligned}
&	\Big|\int_{\mathcal{H}^+_\delta}\sum_{1\leq i_1+i_2\leq j-1}\sum_{\alpha} \mathbb{T}(\mathring{g})(\partial_{\bar{t}},n)[\partial_{\bar{t}}^{i_1}\Omega_\alpha^{i_2}\psi]d\mathring{g}_{\mathcal{H}_\delta}-\int_{\mathcal{H}^+_\delta}\sum_{0\leq i_1+i_2\leq j-1}\sum_\alpha\left(\partial_{\bar{t}}\partial_{\bar{t}}^{i_1}\Omega_\alpha^{i_2}\psi\right)^2\frac{1}{\sqrt{|1-\mu|}}d\mathring{g}_{\mathcal{H}_\delta}\Big|	\\ 
&	\quad\quad\quad\leq C \delta\int_{\mathcal{H}^+_\delta} \sum_{1\leq i\leq j-1} \mathbb{T}(\mathring{g})(N,n)[N^i\psi] d\mathring{g}_{\mathcal{H}_\delta},\\
&\Big|	\int_{\bar{\mathcal{H}}^+_\delta}\sum_{1\leq i_1+i_2\leq j-1}\sum_{\alpha} \mathbb{T}(\mathring{g})(\partial_{\bar{t}},n)[\partial_{\bar{t}}^{i_1}\Omega_\alpha^{i_2}\psi]d\mathring{g}_{\mathcal{H}_\delta}-\int_{\bar{\mathcal{H}}^+_\delta}\sum_{0\leq i_1+i_2\leq j-1}\sum_\alpha\left(\partial_{\bar{t}}\partial_{\bar{t}}^{i_1}\Omega_\alpha^{i_2}\psi\right)^2\frac{1}{\sqrt{|1-\mu|}}d\mathring{g}_{\mathcal{H}_\delta}\Big|	\\
&	\quad\quad\quad\leq C \delta\int_{\bar{\mathcal{H}}^+_\delta} \sum_{1\leq i\leq j-1} \mathbb{T}(\mathring{g})(N,n)[N^i\psi] d\mathring{g}_{\mathcal{H}_\delta},\\
\end{aligned}
\end{equation}
for $1\leq j\leq k+2$. 
\end{lemma}
\begin{proof}
We note that for for the metric $\mathring{g}$, we obtain 
\begin{equation}\label{eq: proof prop: well posedness quasilinear, eq 7.1}
\begin{aligned}
\mathbb{T}(\mathring{g})(\partial_{\bar{t}},\mathring{n}_{\mathcal{H}_\delta})[\partial_{\bar{t}}^j\psi]	= \Bigg(&  (\partial_{\bar{t}}^{j+1}\psi)^2\left(c_1(r)-\frac{1}{2}\mathring{g}_{\bar{t}\bar{t}}\mathring{g}^{\bar{t}\bar{t}}-\frac{1}{2}c_2(r)\mathring{g}_{\bar{t}r}\mathring{g}^{\bar{t}\bar{t}}\right)  \\
&	+\partial_{\bar{t}}\partial_{\bar{t}}^j\psi\partial_r\partial_{\bar{t}}^j\psi \left(c_2(r)-c_1(r)\mathring{g}_{\bar{t}\bar{t}}\mathring{g}^{\bar{t}r}-c_2(r)\mathring{g}_{\bar{t}r}\mathring{g}^{\bar{t}r}\right)\\
&	+\left|\slashed{\nabla}\partial_{\bar{t}}^j\psi\right|^2\left(-\frac{1}{2}c_1(r)g_{\bar{t}\bar{t}}-\frac{1}{2}c_2(r)g_{\bar{t}r}\right)\\
&	+(\partial_r\partial_{\bar{t}}^j\psi)^2\left(-\frac{1}{2}c_1(r)\mathring{g}_{\bar{t}\bar{t}}\mathring{g}^{rr}-\frac{1}{2}c_2(r)\mathring{g}_{\bar{t}r}\mathring{g}^{rr}\right)\Bigg),
\end{aligned}
\end{equation}
where for $\mathring{n}_{\mathcal{H}_\delta},c_1(r),c_2(r)$ see Section \ref{subsec: volume forms}, and specifically $c_1(r)=\frac{1}{\sqrt{|1-\mu|}}+\mathcal{O}(\sqrt{|1-\mu|})$ as $r\rightarrow r_+$ and $c_2(r)=-\sqrt{|1-\mu|}$. At $r=r_+-\delta$ the following hold
\begin{equation}\label{eq: proof prop: well posedness quasilinear, eq 7.2}
\begin{aligned}
\left|-\frac{1}{2}\mathring{g}_{\bar{t}\bar{t}}\mathring{g}^{\bar{t}\bar{t}}-\frac{1}{2}c_2(r)\mathring{g}_{\bar{t}r}\mathring{g}^{\bar{t}\bar{t}}\right|	&	\sim \sqrt{|1-\mu|}\sim\delta^{1/2},\\
\left|c_2(r)-c_1(r)\mathring{g}_{\bar{t}\bar{t}}\mathring{g}^{\bar{t}r}-c_2(r)\mathring{g}_{\bar{t}r}\mathring{g}^{\bar{t}r}\right|&\sim \sqrt{|1-\mu|}\sim\delta^{1/2}, \\ 
\left|-\frac{1}{2}c_1(r)g_{\bar{t}\bar{t}}-\frac{1}{2}c_2(r)g_{\bar{t}r}\right|	&	\sim \sqrt{|1-\mu|}\sim\delta^{1/2},\\
\left|-\frac{1}{2}c_1(r)\mathring{g}_{\bar{t}\bar{t}}\mathring{g}^{rr}-\frac{1}{2}c_2(r)\mathring{g}_{\bar{t}r}\mathring{g}^{rr}\right|	&	\sim |1-\mu|^{3/2}\sim\delta^{3/2}, 
\end{aligned}
\end{equation}
see Remark \ref{rem: sec: preliminairies, rem 0} for the inverse metric components. A similar estimate holds for the hypersurface $\bar{\mathcal{H}}^+_\delta$, since $\mathring{n}_{\bar{\mathcal{H}}^+_\delta}$ has a favorable form, see Section \ref{subsec: volume forms}. We obtain the linear version of the Lemma. 

For the quasilinear version we use the smallness~\eqref{eq: lem: subsec: elliptic estimates, lem 1, eq 1.1} to obtain 
\begin{equation}
\begin{aligned}
&	\Bigg|\int_{\mathcal{H}^+_\delta}\sum_{1\leq i_1+i_2\leq j-1}\sum_{\alpha} \mathbb{T}(\partial_{\bar{t}},\mathring{n}_{\mathcal{H}_\delta})[\partial_{\bar{t}}^{i_1}\Omega_\alpha^{i_2}\psi]d\mathring{g}_{\mathcal{H}_\delta}-\int_{\mathcal{H}^+_\delta}\sum_{0\leq i_1+i_2\leq j-1}\sum_\alpha\left(\partial_{\bar{t}}\partial_{\bar{t}}^{i_1}\Omega_\alpha^{i_2}\psi\right)^2\frac{1}{\sqrt{|1-\mu|}}d\mathring{g}_{\mathcal{H}_\delta}\Bigg|	\\ 
&	\quad\quad\quad\leq C \sqrt{\delta}\int_{\mathcal{H}^+_\delta} \sum_{1\leq i\leq j-1} \mathbb{T}(N,\mathring{n}_{\mathcal{H}_\delta})[N^i\psi] d\mathring{g}_{\mathcal{H}_\delta}+\sqrt{\epsilon(\delta)}\int_{\mathcal{H}^+_\delta}\sum_{1\leq i_1+i_2+i_3\leq j}\sum_\alpha\left(\partial_{\bar{t}}^{i_1}\partial_r^{i_2}\Omega_\alpha^{i_3}\psi\right)^2\frac{d\mathring{g}_{\mathcal{H}_\delta}}{\sqrt{|1-\mu|}}		,\\
&	\Bigg| \int_{\bar{\mathcal{H}}^+_\delta}\sum_{1\leq i_1+i_2\leq j-1}\sum_{\alpha} \mathbb{T}(\partial_{\bar{t}},\mathring{n}_{\mathcal{H}_\delta})[\partial_{\bar{t}}^{i_1}\Omega_\alpha^{i_2}\psi]d\mathring{g}_{\mathcal{H}_\delta}-\int_{\bar{\mathcal{H}}^+_\delta}\sum_{0\leq i_1+i_2\leq j-1}\sum_\alpha\left(\partial_{\bar{t}}\partial_{\bar{t}}^{i_1}\Omega_\alpha^{i_2}\psi\right)^2\frac{1}{\sqrt{|1-\mu|}}d\mathring{g}_{\mathcal{H}_\delta}\Bigg|	\\
&	\quad\quad\quad\leq C \sqrt{\delta}\int_{\bar{\mathcal{H}}^+_\delta} \sum_{1\leq i\leq j-1} \mathbb{T}(N,\mathring{n}_{\mathcal{H}_\delta})[N^i\psi] d\mathring{g}_{\mathcal{H}_\delta}+\sqrt{\epsilon(\delta)}\int_{\bar{\mathcal{H}}^+_\delta}\sum_{1\leq i_1+i_2+i_3\leq j}\sum_\alpha\left(\partial_{\bar{t}}^{i_1}\partial_r^{i_2}\Omega_\alpha^{i_3}\psi\right)^2 \frac{d\mathring{g}_{\mathcal{H}_\delta}}{\sqrt{|1-\mu|}},\\
\end{aligned}
\end{equation}
for all $1\leq j\leq k+2$, where $C=C(k,M,\Lambda,A_{[k+1]},B_{[k+1]})$. Therefore, for a sufficiently small $\delta$ and $\epsilon(\delta)\ll \delta$, we use Lemma~\ref{lem: subsec: elliptic estimates, lem 1} and conclude.
\end{proof}

\subsection{Proof of Proposition \ref{prop: well posedness quasilinear}}

Now we proceed to the proof of the Proposition.

\begin{proof}
	 
	Let $\tau_1\geq 0$. 
	 
	We will use the derivatives notation of Section \ref{subsec: derivatives notation}. Moreover, for the semilinear term $\partial\psi\cdot\partial\psi$ and for the definition of smooth tensors $a,h$, see Sections \ref{subsec: semilinear derivatives}, \ref{subsec: assumption} respectively. In this proof, when we write 
	\begin{equation}
	h_{ab}
	\end{equation}
	it is to be understood $h_{ab}(\nabla\psi)$ in Cartesian coordinates, see~\eqref{eq: cartesian manifold}. When we write 
	\begin{equation}
	h^{ab}
	\end{equation}
	it is to be understood as $h^{ab}(\nabla\psi)=g^{ab}(\nabla\psi)-\mathring{g}^{ab}$, in Cartesian coordinates. Finally, in this proof when we write $\lesssim$ it is to be understood that we omit a constant $C(k,M,\Lambda,A_{[m]},B_{[m]})$, where for $A_{[m]},B_{[m]}$ see Sections \ref{subsec: semilinear derivatives}, \ref{subsec: assumption} respectively, where $m\leq k+1$.

	\textbf{Existence and uniqueness.}

	The existence and uniqueness part of the proof follow readily, by well known arguments, from the energy estimates~\eqref{eq: prop: well posedness quasilinear, eq 1},~\eqref{eq: prop: well posedness quasilinear, eq 2}. Therefore, we only prove~\eqref{eq: prop: well posedness quasilinear, eq 1},~\eqref{eq: prop: well posedness quasilinear, eq 2}, assuming existence.
	
	\textbf{An auxiliary energy.}
		
	We here introduce an auxiliary energy with which we will estimate the $\partial_{\bar{t}}$ flux on $\{\bar{t}=\tau\}$, see already~\eqref{eq: proof prop: well posedness quasilinear, eq 7.01}. We define 
	\begin{equation}\label{eq: proof prop: well posedness quasilinear, eq 2.3}
	E_{j,q}[\psi](\tau)=\int_{\{\bar{t}=\tau\}\cap\{r_++q\leq r\leq\bar{r}_+-q\}}\sum_{1\leq i_1+i_2\leq j}\sum_\alpha\left(\partial_{\bar{t}}^{i_1}\Omega_\alpha^{i_2}\psi\right)^2+(1-\mu)\sum_{1\leq i_1+i_2+i_3\leq j} \sum_\alpha \left(\partial_{\bar{t}}^{i_1}\partial_r^{i_2}\Omega_\alpha^{i_3}\psi\right)^2 d\mathring{g}_{\{\bar{t}=\tau\}},
	\end{equation}
	for $q>0$ as in the redshift Lemma~\ref{lem: redshift, lem 1}.

	\textbf{The linear wave equation analogue}
	
	Before proving the Cauchy stability estimate~\eqref{eq: prop: well posedness quasilinear, eq 1},~\eqref{eq: prop: well posedness quasilinear, eq 2} for the solutions of the quasilinear wave equation, we sketch the proof of~\eqref{eq: prop: well posedness quasilinear, eq 1},~\eqref{eq: prop: well posedness quasilinear, eq 2} for the solutions of the linear wave equation
	\begin{equation}\label{eq: proof prop: well posedness quasilinear, eq 2.07}
		\Box_{\mathring{g}}\psi=0,
	\end{equation}
	so that the reader can see the similarity between the two proofs. Also see the Lecture notes~\cite{DR5}. The steps of the proof of the relevant quasilinear estimate are similar, \textit{but by also treating the non-linearities appropriately}. 
	
	First, we apply Proposition \ref{prop: energy identity, prop 1, SdS} with multiplier $\partial_{\bar{t}}$ and obtain that there exists a constant $C(M,\Lambda)>1$ such that 
	\begin{equation}\label{eq: proof prop: well posedness quasilinear, eq 2.08}
		E_{1,q}[\psi](\tau_2)\leq C(M,\Lambda) E_1[\psi](\tau_1)
	\end{equation}
	for all $\tau_1\leq \tau_2$. By commuting the wave equation appropriately many times with $\partial_{\bar{t}}\Omega_\alpha$ and by using elliptic estimates, then~\eqref{eq: proof prop: well posedness quasilinear, eq 2.08} implies 
	\begin{equation}\label{eq: proof prop: well posedness quasilinear, eq 2.081}
		E_{j,q}[\psi](\tau_2)\leq C(j,M,\Lambda) E_j[\psi](\tau_1)
	\end{equation}
	for all $j\geq 1$, for some constant $C(j,M,\Lambda)>1$.

	Second, we apply the divergence Theorem with multiplier the red-shift vector field $N$, see Lemma~\ref{lem: redshift, lem 1}, and obtain
	\begin{equation}\label{eq: proof prop: well posedness quasilinear, eq 2.09}
		E_1[\psi](\tau_2)+\int_{\tau_1}^{\tau_2} d\tau E_1[\psi](\tau)\leq C(j,M,\Lambda)\int_{\tau_1}^{\tau_2} E_{1,q}[\psi](\tau) +C(j,M,\Lambda) E_1[\psi](\tau_1),
	\end{equation}
	for all $\tau_1\leq \tau_2$. By commuting the wave equation $j-1$ times with $N$, also see the redshift Lemma~\ref{lem: redshift, lem 2} and by using elliptic estimates, then~\eqref{eq: proof prop: well posedness quasilinear, eq 2.09} implies 
	\begin{equation}\label{eq: proof prop: well posedness quasilinear, eq 2.10}
		E_j[\psi](\tau_2)+\int_{\tau_1}^{\tau_2} d\tau E_j[\psi](\tau)\leq C(j,M,\Lambda)\int_{\tau_1}^{\tau_2} E_{j,q}[\psi](\tau) +C(j,M,\Lambda)E_j[\psi](\tau_1),
	\end{equation}
	for all $j\geq 1$.

	Now, by using~\eqref{eq: proof prop: well posedness quasilinear, eq 2.081} and~\eqref{eq: proof prop: well posedness quasilinear, eq 2.10} we obtain that there exists a constant $C(j,M,\Lambda)$ such that 	
	\begin{equation}\label{eq: proof prop: well posedness quasilinear, eq 2.1}
	E_j[\psi](\tau_2)+\int_{\tau_1}^{\tau_2}E_j[\psi](\tau)d\tau\leq C(j,M,\Lambda) E_j[\psi](\tau_1) (\tau_2-\tau_1)+C(j,M,\Lambda) E_j[\psi](\tau_1),
	\end{equation}
	for all $\tau_1\leq \tau_2$. Note that~\eqref{eq: proof prop: well posedness quasilinear, eq 2.09} holds for all $\tau_1^\prime\geq \tau_1$ in the place of $\tau_1$. We readily obtain that 
	\begin{equation}
	E_j[\psi](\tau)\leq C(j,M,\Lambda)E_j[\psi](\tau_1)
	\end{equation}
	for all $\tau\geq \tau_1$, which concludes the proof of the energy estimates~\eqref{eq: prop: well posedness quasilinear, eq 1},~\eqref{eq: prop: well posedness quasilinear, eq 2} for~\eqref{eq: proof prop: well posedness quasilinear, eq 2.07}.

	\textbf{The bootstrap.}
	
	We now return to the proof of the Proposition \ref{prop: well posedness quasilinear}.

	We introduce the bootstrap assumption  
	\begin{equation}\label{eq: proof prop: well posedness quasilinear, eq 3}
	\sup_{D_\delta(\tau_1,\tau_1+\tau_{\text{max}})}\sum_{1\leq j\leq k-1}\sum_{\partial\in\{\partial_{\bar{t}},\partial_r,\Omega_1,\Omega_2,\Omega_3\}}|\partial^j\psi|\leq C_{\text{b}}\sqrt{\epsilon}
	\end{equation}
	for a large constant $C_{\text{b}}(M,\Lambda)>0$ to be determined later.

	\textbf{Volume form notation.}

	In this proof in any integrals over spacetime domains or hupersurfaces without explicit volume forms it is to be understood that the volume forms of the integrals over spacetime domains are with respect to the metric $g(\nabla\psi)$. The volume forms of the integrals over hypersurfaces are with respect to the induced volume form of the metric $g(\nabla\psi)$ on those hypersurfaces.

	\textbf{The multiplier estimates.}

	First, we write the general energy identities (multiplier estimates) we need. We will later use them for the vector fields $\partial_{\bar{t}},N$.

	We apply the divergence Theorem, see Proposition \ref{prop: energy identity, prop 1, SdS}, for the quasilinear wave equation~\eqref{eq: quasilinear}, for the metric $g(\nabla\psi)$, with multiplier $X$ and appropriately many commutations with a vector field $\tilde{X}$, on the extended region
	\begin{equation}
		D_{\delta}(\tau_1,\tau_1+\tau_{\textit{max}}).
	\end{equation}
	We consider any 
	\begin{equation}
		\tau_2-\tau_1\leq \tau_{\textit{max}}.
	\end{equation}
	We obtain 
	\begin{equation}\label{eq: proof prop: well posedness quasilinear, eq 4}
	\begin{aligned}
	&		\left(\int_{\{\bar{t}=\tau_2\}\cap D_\delta(\tau_1,\tau_2)}+\int_{\mathcal{H}^+_{\delta}\cap D_\delta(\tau_1,\tau_2)}+\int_{\bar{\mathcal{H}}^+_\delta\cap D_\delta(\tau_1,\tau_2)}\right)\mathbb{T}(g(\nabla\psi))(X,n)[\tilde{X}^{i}\psi]\\
	&    +\int\int_{D_\delta(\tau_1,\tau_2)} 	 \frac{1}{2}\:^{(X)}\pi_{\mu\nu}(g(\nabla\psi))\mathbb{T}^{\mu\nu}(g(\nabla\psi))[\tilde{X}^i\psi]\\
	&   = \quad \int_{\{\bar{t}=\tau_1\}\cap D_\delta(\tau_1,\tau_2)} \mathbb{T}(g(\nabla\psi))(X,n)[\tilde{X}^{i}\psi]-\int\int_{D_\delta(\tau_1,\tau_2)}X\tilde{X}^{i}\psi\left(\sum_{l=0}^{j-1}\tilde{X}^l\left[\Box_{g(\nabla\psi)},\tilde{X}\right]\tilde{X}^{i-1-l}\psi+\tilde{X}^i\left(a^{\alpha\beta}\partial_\alpha\psi\partial_\beta\psi\right)\right) 
	\end{aligned}
	\end{equation}
	for all $1\leq i\leq k+1$. Note the deformation tensor of $g(\nabla\psi)$ is
	\begin{equation}\label{eq: proof prop: well posedness quasilinear, eq 5}
	\begin{aligned}
	&  \frac{1}{2}\:^{(X)}\pi^{\mu\nu}(g(\nabla\psi))\mathbb{T}_{\mu\nu}(g(\nabla\psi))[\tilde{X}^{i}\psi]    \\
	&	= \frac{1}{2} \left(\mathcal{L}_X\left(\mathring{g}+h(\nabla\psi)\right)\right)_{\mu\nu}g^{\mu\alpha}(\nabla\psi)g^{\nu \beta}(\nabla\psi)\mathbb{T}_{\alpha\beta}(g(\nabla\psi))[\tilde{X}^{i}\psi]\\
	&  =  \frac{1}{2}\left(\mathcal{L}_X\mathring{g}+\mathcal{L}_X h(\nabla\psi)\right)_{\mu\nu}\left(\mathring{g}^{\mu\alpha}+h^{\mu\alpha}(\nabla\psi)\right)\left(\mathring{g}^{\nu\beta}+h^{\nu\beta}(\nabla\psi)\right)\left(\mathbb{T}(\mathring{g})(\partial_\alpha,\partial_\beta)[\tilde{X}^i\psi] -\frac{1}{2}h(\partial_\alpha,\partial_\beta,\nabla\psi)|\nabla \tilde{X}^i\psi|_{g(\nabla\psi)}^2\right)\\
	&	\:\dot{=}	\frac{1}{2}\left(^{(X)}\pi_{\mu\nu}(\mathring{g})\mathbb{T}^{\mu\nu}(\mathring{g})[\tilde{X}^{i}\psi] \right) +^{(X)}\pi_{\textit{non-lin}}(\tilde{X}^i\psi),
	\end{aligned}
	\end{equation}
	with 
	\begin{equation}
		\begin{aligned}
		&	^{(X)}\pi_{\textit{non-lin}}(\tilde{X}^i\psi)
		\\
		&	\quad= \left(\mathcal{L}_X\mathring{g}\right)_{\mu\nu}\mathring{g}^{\mu\alpha}\left(\mathring{g}^{\nu\beta}\frac{-1}{2}h_{\alpha\beta}(\nabla\psi)|\nabla\tilde{X}^i\psi|^2_{g(\nabla\psi)}+ h^{\nu\beta}\mathbb{T}(\mathring{g})(\partial_{\alpha},\partial_\beta)[\tilde{X}^i\psi]+h^{\nu\beta}(\nabla\psi)\frac{-1}{2}h_{\alpha\beta}(\nabla\psi)\right)\\
		&	\quad\quad + (\mathcal{L}_X\mathring{g})_{\mu\nu}h^{\mu\alpha}(\nabla\psi)\Bigg(\mathring{g}^{\nu\beta}\mathbb{T}(\mathring{g})(\partial_\alpha,\partial_\beta)[\tilde{X}^i\psi]\\
		&	\quad\quad\quad\quad\quad\quad\quad\quad\quad\quad\quad\quad\quad+\mathring{g}^{\nu\beta}\frac{-1}{2}h_{\alpha\beta}(\nabla\psi)|\nabla\tilde{X}^i\psi|^2_{g(\nabla\psi)}+ h^{\nu\beta}\mathbb{T}(\mathring{g})(\partial_{\alpha},\partial_\beta)[\tilde{X}^i\psi]+h^{\nu\beta}(\nabla\psi)\frac{-1}{2}h_{\alpha\beta}(\nabla\psi)\Bigg)\\
		&	\quad\quad +\left(\mathcal{L}_Xh(\nabla\psi)\right)_{\mu\nu}\mathring{g}^{\mu\alpha}\Bigg(\mathring{g}^{\nu\beta}\mathbb{T}(\mathring{g})(\partial_\alpha,\partial_\beta)[\tilde{X}^i\psi]\\
		&	\quad\quad\quad\quad\quad\quad\quad\quad\quad\quad\quad\quad\quad+\mathring{g}^{\nu\beta}\frac{-1}{2}h_{\alpha\beta}(\nabla\psi)|\nabla\tilde{X}^i\psi|^2_{g(\nabla\psi)}+ h^{\nu\beta}\mathbb{T}(\mathring{g})(\partial_{\alpha},\partial_\beta)[\tilde{X}^i\psi]+h^{\nu\beta}(\nabla\psi)\frac{-1}{2}h_{\alpha\beta}(\nabla\psi)\Bigg)\\
		&	\quad\quad +\left(\mathcal{L}_X h(\nabla\psi) h^{\mu\alpha}(\nabla\psi)\right)\Bigg(\mathring{g}^{\nu\beta}\mathbb{T}(\mathring{g})(\partial_\alpha,\partial_\beta)[\tilde{X}^i\psi]\\
		&	\quad\quad\quad\quad\quad\quad\quad\quad\quad\quad\quad\quad\quad+\mathring{g}^{\nu\beta}\frac{-1}{2}h_{\alpha\beta}(\nabla\psi)|\nabla\tilde{X}^i\psi|^2_{g(\nabla\psi)}+ h^{\nu\beta}\mathbb{T}(\mathring{g})(\partial_{\alpha},\partial_\beta)[\tilde{X}^i\psi]+h^{\nu\beta}(\nabla\psi)\frac{-1}{2}h_{\alpha\beta}(\nabla\psi)\Bigg)
		\end{aligned}
	\end{equation}
	where we used that the energy momentum tensor of  $g(\nabla\psi)$ is 
	\begin{equation}\label{eq: proof prop: well posedness quasilinear, eq 6}
		\mathbb{T}(g(\nabla\psi))(X,Y)[\psi]= \mathbb{T}(\mathring{g})(X,Y)[\psi] -\frac{1}{2}h(X,Y,\nabla\psi)|\nabla \psi|_{g(\nabla\psi)}^2
	\end{equation}
	for any two smooth vector field $X,Y$, with
	\begin{equation}
		|\nabla\psi|^2_{g(\nabla\psi)}=g^{ab}(\nabla\psi)\partial_a\psi\partial_b\psi.
	\end{equation}

	\textbf{The $\partial_{\bar{t}}$ multiplier estimate.}
	
	We apply the energy identity~\eqref{eq: proof prop: well posedness quasilinear, eq 4} with multiplier $\partial_{\bar{t}}$ and commutators  $\partial_{\bar{t}}^{i_1}\Omega_\alpha^{i_2}$, in view of~\eqref{eq: proof prop: well posedness quasilinear, eq 5} and of  
	\begin{equation}
		^{(\partial_{\bar{t}})}\pi_{\mu\nu}(\mathring{g})\mathbb{T}^{\mu\nu}(\mathring{g})=0.
	\end{equation}
	We obtain 
	\begin{equation}\label{eq: proof prop: well posedness quasilinear, eq 7.-1}
	\begin{aligned}
	&		\int_{\{\bar{t}=\tau_2\}}\mathbb{T}(g(\nabla\psi))(\partial_{\bar{t}},n)[\partial_{\bar{t}}^{i_1}\Omega_\alpha^{i_2}\psi]\\
	&   = \int_{\{\bar{t}=\tau_1\}\cap D_\delta(\tau_1,\tau_2)} \mathbb{T}(g(\nabla\psi))(\partial_{\bar{t}},n)[\partial_{\bar{t}}^{i_1}\Omega_\alpha^{i_2}\psi]\\
	&	\quad-\left(\int_{\mathcal{H}^+_{\delta}\cap D_\delta(\tau_1,\tau_2)}+\int_{\bar{\mathcal{H}}^+_\delta\cap D_\delta(\tau_1,\tau_2)}\right)\left(\mathbb{T}(g(\nabla\psi))(\partial_{\bar{t}},n)[\partial_{\bar{t}}^{i_1}\Omega_\alpha^{i_2}\psi]\right) \\
	&\quad -\int\int_{D_\delta(\tau_1,\tau_2)}\:^{(\partial_{\bar{t}})}\pi_{\textit{non-lin}}(\partial_{\bar{t}}^{i_1}\Omega_\alpha^{i_2}\psi)	\\
	&	\quad-\int\int_{D_\delta(\tau_1,\tau_2)}\partial_{\bar{t}}\partial_{\bar{t}}^{i_1}\Omega_\alpha^{i_2}\psi\left(\sum_{j_1+j_2+j_3+j_4= i_1+i_2-1}\sum_{j_5+j_6=1}\partial_{\bar{t}}^{j_1}\Omega_\alpha^{j_2}\left[\Box_{g(\nabla\psi)},\partial_{\bar{t}}^{j_5}\Omega_\alpha^{j_6}\right]\partial_{\bar{t}}^{j_3}\Omega_\alpha^{j_4}\psi+\partial_{\bar{t}}^{i_1}\Omega_\alpha^{i_2}\left(a^{\alpha\beta}\partial_\alpha\psi\partial_\beta\psi\right)\right). \\
	\end{aligned}
	\end{equation}
	Therefore, by summing over $i_1,i_2,\alpha$ and by adding in the hypersurface terms 
	\begin{equation}
		\sum_{1\leq i_1+i_2\leq i}\sum_\alpha\left(\int_{\mathcal{H}^+_\delta\cap D_\delta(\tau_1,\tau_2)}+\int_{\bar{\mathcal{H}}^+_\delta\cap D_\delta(\tau_1,\tau_2)}\right)\left(\partial_{\bar{t}}^{i_1+1}\Omega_\alpha^{i_2}\psi\right)^2 \frac{1}{\sqrt{|1-\mu|}} d\mathring{g}_{\mathcal{H}_\delta}
	\end{equation}
	on both sides of~\eqref{eq: proof prop: well posedness quasilinear, eq 7.-1} and recall the  the energy momentum tensor of $g(\nabla\psi)$, see~\eqref{eq: proof prop: well posedness quasilinear, eq 6}, we obtain 
	\begin{equation}\label{eq: proof prop: well posedness quasilinear, eq 7}
		\begin{aligned}		
			&	\sum_{1\leq i_1+i_2\leq i}\sum_\alpha	\int_{\{\bar{t}=\tau_2\}}\mathbb{T}(\mathring{g})(\partial_{\bar{t}},n)[\partial_{\bar{t}}^{i_1}\Omega_\alpha^{i_2}\psi]\\
			&	+\sum_{1\leq i_1+i_2\leq i}\sum_\alpha\left(\int_{\mathcal{H}^+_\delta\cap D_\delta(\tau_1,\tau_2)}+\int_{\bar{\mathcal{H}}^+_\delta\cap D_\delta(\tau_1,\tau_2)}\right)\left(\partial_{\bar{t}}^{i_1+1}\Omega_\alpha^{i_2}\psi\right)^2 \frac{1}{\sqrt{|1-\mu|}} d\mathring{g}_{\mathcal{H}_\delta}\\
			&   =\sum_{1\leq i_1+i_2\leq i}\sum_\alpha \int_{\{\bar{t}=\tau_1\}} \mathbb{T}(g(\nabla\psi))(\partial_{\bar{t}},n)[\partial_{\bar{t}}^{i_1}\Omega_\alpha^{i_2}\psi]\\
			&	\quad +	\sum_{1\leq i_1+i_2\leq i}\sum_{\alpha}\int_{\{\bar{t}=\tau_2\}}\frac{1}{2}h(\partial_{\bar{t}},n,\nabla\psi)|\nabla \partial_{\bar{t}}^{i_1}\Omega_\alpha^{i_2} \psi|_{g(\nabla\psi)}^2\\
			&	\quad-\sum_{1\leq i_1+i_2\leq i}\sum_\alpha\left(\int_{\mathcal{H}^+_{\delta}\cap D_\delta(\tau_1,\tau_2)}+\int_{\bar{\mathcal{H}}^+_\delta\cap D_\delta(\tau_1,\tau_2)}\right)\left(\mathbb{T}(g(\nabla\psi))(\partial_{\bar{t}},n)[\partial_{\bar{t}}^{i_1}\Omega_\alpha^{i_2}\psi]\right) \\
			&	\quad+\sum_{1\leq i_1+i_2\leq i}\sum_\alpha\left(\int_{\mathcal{H}^+_\delta\cap D_\delta(\tau_1,\tau_2)}+\int_{\bar{\mathcal{H}}^+_\delta\cap D_\delta(\tau_1,\tau_2)}\right)\left(\partial_{\bar{t}}^{i_1+1}\Omega_\alpha^{i_2}\psi\right)^2 \frac{1}{\sqrt{|1-\mu|}} d\mathring{g}_{\mathcal{H}_\delta}\\
			&\quad -\sum_{1\leq i_1+i_2\leq i}\sum_\alpha\int\int_{D_\delta(\tau_1,\tau_2)}\:^{(\partial_{\bar{t}})}\pi_{\textit{non-lin}}(\partial_{\bar{t}}^{i_1}\Omega_\alpha^{i_2}\psi)	\\
			&	\quad-\sum_{1\leq i_1+i_2\leq i}\sum_\alpha\int\int_{D_\delta(\tau_1,\tau_2)}\partial_{\bar{t}}\partial_{\bar{t}}^{i_1}\Omega_\alpha^{i_2}\psi\Bigg(\sum_{j_1+j_2+j_3+j_4= i_1+i_2-1}\sum_{j_5+j_6=1}\partial_{\bar{t}}^{j_1}\Omega_\alpha^{j_2}\left[\Box_{g(\nabla\psi)},\partial_{\bar{t}}^{j_5}\Omega_\alpha^{j_6}\right]\partial_{\bar{t}}^{j_3}\Omega_\alpha^{j_4}\psi\\
			&\quad\quad\quad\quad\quad\quad\quad\quad\quad\quad\quad\quad\quad\quad\quad\quad\quad\quad\quad+\partial_{\bar{t}}^{i_1}\Omega_\alpha^{i_2}\left(a^{\alpha\beta}\partial_\alpha\psi\partial_\beta\psi\right)\Bigg), \\		
		\end{aligned}
	\end{equation}
	Now, note that there exist constants $c(M,\Lambda),C(M,\Lambda)>0$, such that 
	\begin{equation}
	c(M,\Lambda)\left((\partial_{\bar{t}}\psi)^2+|\slashed{\nabla}\psi|^2\right)\leq\mathbb{T}(\mathring{g})(\partial_{\bar{t}},\mathring{n}_{\{\bar{t}=\tau\}})[\psi]+C(M,\Lambda)|1-\mu|(\partial_r\psi)^2,
	\end{equation}
	in sufficiently small neighborhoods of the horizons and 
	\begin{equation}
		 c(M,\Lambda,q)\left((\partial_{\bar{t}}\psi)^2+(\partial_r\psi)^2+|\slashed{\nabla}\psi|^2\right)\leq\mathbb{T}(\mathring{g})(\partial_{\bar{t}},\mathring{n}_{\{\bar{t}=\tau\}})[\psi]
	\end{equation}
	in $(r_++q,\bar{r}_+-q)$. Therefore, in view of Lemma~\ref{lem: metric close to SdS} and the elliptic estimate of Lemma~\ref{lem: subsec: elliptic estimates, lem 1}, we note that there exists a $\delta>0$ sufficiently small, and an $\epsilon(\delta,C_{\text{b}})>0$ sufficiently small, such that
	\begin{equation}\label{eq: proof prop: well posedness quasilinear, eq 7.01}
	\begin{aligned}
	&		c E_{j,q}(\tau_2)\\
	&   \leq C E_j(\tau_1)\\
	&	\quad+C\delta\int_{\{\bar{t}=\tau_2\}}\sum_{1\leq i_1+i_2+i_3\leq j}\sum_\alpha\left(\partial_{\bar{t}}^{i_1}\partial_r^{i_2}\Omega_\alpha^{i_3}\psi\right)^2 d\mathring{g}_{\{\bar{t}=\tau\}}\\
	&	\quad-\sum_{1\leq i_1+i_2\leq j-1}\sum_\alpha\left(\int_{\mathcal{H}^+_{\delta}\cap D_\delta(\tau_1,\tau_2)}+\int_{\bar{\mathcal{H}}^+_\delta\cap D_\delta(\tau_1,\tau_2)}\right)\left(\mathbb{T}(g(\nabla\psi))(\partial_{\bar{t}},n)[\partial_{\bar{t}}^{i_1}\Omega_\alpha^{i_2}\psi]\right) \\
	&	\quad+\sum_{1\leq i_1+i_2\leq j-1}\sum_\alpha\left(\int_{\mathcal{H}^+_\delta\cap D_\delta(\tau_1,\tau_2)}+\int_{\bar{\mathcal{H}}^+_\delta\cap D_\delta(\tau_1,\tau_2)}\right)\left(\partial_{\bar{t}}^{i_1+1}\Omega_\alpha^{i_2}\psi\right)^2 \frac{1}{\sqrt{|1-\mu|}} d\mathring{g}_{\mathcal{H}_\delta}\\
	&	\quad +	\sum_{1\leq i_1+i_2\leq i}\sum_{\alpha}\int_{\{\bar{t}=\tau_2\}}\frac{1}{2}h(\partial_{\bar{t}},n,\nabla\psi)|\nabla \partial_{\bar{t}}^{i_1}\Omega_\alpha^{i_2} \psi|_{g(\nabla\psi)}^2\\
	&\quad -\sum_{1\leq i_1+i_2\leq j-1}\sum_\alpha\int\int_{D_\delta(\tau_1,\tau_2)}\:^{(\partial_{\bar{t}})}\pi_{\textit{non-lin}}(\partial_{\bar{t}}^{i_1}\Omega_\alpha^{i_2}\psi)	\\
&	\quad-\sum_{1\leq i_1+i_2\leq j-1}\sum_\alpha\int\int_{D_\delta(\tau_1,\tau_2)}\partial_{\bar{t}}\partial_{\bar{t}}^{i_1}\Omega_\alpha^{i_2}\psi\Bigg(\sum_{j_1+j_2+j_3+j_4= i_1+i_2-1}\sum_{j_5+j_6=1}\partial_{\bar{t}}^{j_1}\Omega_\alpha^{j_2}\left[\Box_{g(\nabla\psi)},\partial_{\bar{t}}^{j_5}\Omega_\alpha^{j_6}\right]\partial_{\bar{t}}^{j_3}\Omega_\alpha^{j_4}\psi\\
&	\quad\quad\quad\quad\quad\quad\quad\quad\quad\quad\quad\quad\quad\quad\quad\quad\quad\quad\quad\quad+\partial_{\bar{t}}^{i_1}\Omega_\alpha^{i_2}\left(a^{\alpha\beta}\partial_\alpha\psi\partial_\beta\psi\right)\Bigg) \\
&   = C E_j(\tau_1)\\
&	\quad+C\delta\int_{\{\bar{t}=\tau_2\}}\sum_{1\leq i_1+i_2+i_3\leq j}\sum_\alpha\left(\partial_{\bar{t}}^{i_1}\partial_r^{i_2}\Omega_\alpha^{i_3}\psi\right)^2 d\mathring{g}_{\{\bar{t}=\tau\}}\\
&	\quad-\sum_{1\leq i_1+i_2\leq j-1}\sum_\alpha\left(\int_{\mathcal{H}^+_{\delta}\cap D_\delta(\tau_1,\tau_2)}+\int_{\bar{\mathcal{H}}^+_\delta\cap D_\delta(\tau_1,\tau_2)}\right)\left(\mathbb{T}(g(\nabla\psi))(\partial_{\bar{t}},n)[\partial_{\bar{t}}^{i_1}\Omega_\alpha^{i_2}\psi]\right) \\
&	\quad+\sum_{1\leq i_1+i_2\leq j-1}\sum_\alpha\left(\int_{\mathcal{H}^+_\delta\cap D_\delta(\tau_1,\tau_2)}+\int_{\bar{\mathcal{H}}^+_\delta\cap D_\delta(\tau_1,\tau_2)}\right)\left(\partial_{\bar{t}}^{i_1+1}\Omega_\alpha^{i_2}\psi\right)^2 \frac{1}{\sqrt{|1-\mu|}} d\mathring{g}_{\mathcal{H}_\delta}\\
&	\quad +	\sum_{1\leq i_1+i_2\leq i}\sum_{\alpha}\int_{\{\bar{t}=\tau_2\}}\frac{1}{2}h(\partial_{\bar{t}},n,\nabla\psi)|\nabla \partial_{\bar{t}}^{i_1}\Omega_\alpha^{i_2} \psi|_{g(\nabla\psi)}^2\\
&\quad -\sum_{1\leq i_1+i_2\leq j-1}\sum_\alpha\int\int_{D_\delta(\tau_1,\tau_2)}\:^{(\partial_{\bar{t}})}\pi_{\textit{non-lin}}(\partial_{\bar{t}}^{i_1}\Omega_\alpha^{i_2}\psi)	\\
&	\quad-\sum_{1\leq i_1+i_2\leq j-1}\sum_\alpha\int\int_{D_\delta(\tau_1,\tau_2)}\partial_{\bar{t}}\partial_{\bar{t}}^{i_1}\Omega_\alpha^{i_2}\psi\Bigg(\sum_{j_1+j_2+j_3+j_4= i_1+i_2-1}\sum_{j_5+j_6=1}\partial_{\bar{t}}^{j_1}\Omega_\alpha^{j_2}\left[h^{ab}(\nabla\psi)\partial_a\partial_b+S^c\:_{ab}(h)\partial_c,\partial_{\bar{t}}^{j_5}\Omega_\alpha^{j_6}\right]\partial_{\bar{t}}^{j_3}\Omega_\alpha^{j_4}\psi\\
&	\quad\quad\quad\quad\quad\quad\quad\quad\quad\quad\quad\quad\quad\quad\quad\quad\quad\quad\quad\quad+\partial_{\bar{t}}^{i_1}\Omega_\alpha^{i_2}\left(a^{\alpha\beta}\partial_\alpha\psi\partial_\beta\psi\right)\Bigg) \\	
\end{aligned}
\end{equation}
for all $1\leq j\leq k+1$, for some constants $c(M,\Lambda),C(k,M,\Lambda,A_{[k+1]},B_{[k+1]})>0$, where for $S^c\:_{ab}$ see~\eqref{eq: energy estimate for quasilinear, 2, 2, the nonlinear part of the christophel symbols}. Note that to get the energies $E_{j,q},E_j$ on the left hand side and right hand side respectively, we used the bootstrap assumption~\eqref{eq: proof prop: well posedness quasilinear, eq 3} and Lemma~\ref{lem: metric close to SdS}, for $\epsilon(\delta,C_{\text{b}})>0$ sufficiently small.

We now estimate the hypersurface terms on $\mathcal{H}^+_\delta,\bar{\mathcal{H}}^+_\delta$ on the right hand side of~\eqref{eq: proof prop: well posedness quasilinear, eq 7.01} (in view of Lemma~\ref{lem: subsec: elliptic estimates, lem 2}) and the two last (non-linear) terms of~\eqref{eq: proof prop: well posedness quasilinear, eq 7.01}. Specifically, in view of the bootstrap~\eqref{eq: proof prop: well posedness quasilinear, eq 3}, and the definitions of the smooth tensors $a,\: h$, see Sections \ref{subsec: semilinear derivatives}, \ref{subsec: assumption} respectively, there exists a $\delta>0$ sufficiently small, and an $\epsilon(\delta,C_{\text{b}})>0$ sufficiently small, such that for the two non-linear terms we use the Lemma~\ref{lem: metric close to SdS} and the coarea formula~\eqref{subsec: coarea}, to obtain
	\begin{equation}\label{eq: proof prop: well posedness quasilinear, eq 7.3}
	\begin{aligned}
	&		 E_{k+1,q}(\tau_2)\\
	&   \leq C E_{k+1}(\tau_1)+C\delta E_{k+1}(\tau_2)\\
	&\quad +	C\sqrt{\delta}\left(\int_{\mathcal{H}^+_{\delta}\cap D_\delta(\tau_1,\tau_2)}+\int_{\bar{\mathcal{H}}_\delta^+\cap D_\delta(\tau_1,\tau_2)}\right)\sum_{1\leq i\leq k}\mathbb{T}(N,n)[N^i\psi]  d\mathring{g}_{\mathcal{H}_\delta}\\
	&	\quad +\sqrt{\epsilon(\delta)} C \int_{\tau_1}^{\tau_2}d\tau E_j(\tau),
	\end{aligned}
	\end{equation}
	for $k\geq 7$, where $C=C(k,M,\Lambda,A_{[k+1]},B_{[k+1]})>0$. We appropriately distributed derivatives and used a Sobolev estimate, see Lemma~\ref{lem: sobolev estimate}, in view of the fact that $k=7$ is the smallest integer such that 
	\begin{equation}
		\floor{\frac{k+1}{2}}+2\leq k-1.
	\end{equation}

	\textbf{The $N$ multiplier estimate.}
	
	We apply the energy identity~\eqref{eq: proof prop: well posedness quasilinear, eq 4} with multiplier $N$ and commutators  $N^i$, in view of~\eqref{eq: proof prop: well posedness quasilinear, eq 5}. 
	
	We obtain
	\begin{equation}\label{eq: proof prop: well posedness quasilinear, eq 8}
	\begin{aligned}
	&		\left(\int_{\{\bar{t}=\tau_2\}\cap D_\delta(\tau_1,\tau_2)}+\int_{\mathcal{H}^+_{\delta}\cap D_\delta(\tau_1,\tau_2)}+\int_{\bar{\mathcal{H}}^+_\delta\cap D_\delta(\tau_1,\tau_2)}\right)\mathbb{T}(g(\nabla\psi))(N,n)[N^i\psi]\\
	&    +\int\int_{D_\delta(\tau_1,\tau_2)} 	 \frac{1}{2}\:^{(N)}\pi_{\mu\nu}(\mathring{g})\mathbb{T}^{\mu\nu}(\mathring{g})[N^i\psi]\\
	&   = \quad \int_{\{\bar{t}=\tau_1\}\cap D_\delta(\tau_1,\tau_2)} \mathbb{T}(g(\nabla\psi))(N,n)[N^i\psi]\\
	&	\quad-\int\int_{D_\delta(\tau_1,\tau_2)}\:^{(N)}\pi_{\textit{non-lin}}(N^i\psi)-\int\int_{D_\delta(\tau_1,\tau_2)}N^{i+1}\psi\left(\sum_{l=0}^{i-1}N^l\left[\Box_{g(\nabla\psi)},N\right]N^{i-1-l}\psi+N^i\left(a^{\alpha\beta}\partial_\alpha\psi\partial_\beta\psi\right)\right) 
	\end{aligned}
	\end{equation}
	for all $1\leq i\leq k$. 
	
	Therefore, for a sufficiently small $\delta$ we obtain  
	\begin{equation}\label{eq: proof prop: well posedness quasilinear, eq 9}
	\begin{aligned}
	&		\left(\int_{\{\bar{t}=\tau_2\}\cap D_\delta(\tau_1,\tau_2)}+\int_{\mathcal{H}^+_{\delta}\cap D_\delta(\tau_1,\tau_2)}+\int_{\bar{\mathcal{H}}^+_\delta\cap D_\delta(\tau_1,\tau_2)}\right)\mathbb{T}(g(\nabla\psi))(N,n)[N^{i}\psi]\\
	&    +\int\int_{D_\delta(\tau_1,\tau_2)\cap \left(\{r_+-q\leq r\leq r_++q\}\cup\{\bar{r}_+-q\leq r\leq\bar{r}_++q\}\right)} 	 \frac{1}{2}\:^{(N)}\pi_{\mu\nu}(\mathring{g})\mathbb{T}^{\mu\nu}(\mathring{g})[N^{i}\psi]\\
	&   = \quad \int_{\{\bar{t}=\tau_1\}\cap D_\delta(\tau_1,\tau_2)} \mathbb{T}(g(\nabla\psi))(N,n)[N^{i}\psi]\\
	&	\quad -\int\int_{D_\delta(\tau_1,\tau_2)\setminus \left(\{r_+-q\leq r\leq r_++q\}\cup\{\bar{r}_+-q\leq r\leq\bar{r}_++q\}\right)} 	 \frac{1}{2}\:^{(N)}\pi_{\mu\nu}(\mathring{g})\mathbb{T}^{\mu\nu}(g(\nabla\psi))[N^{i}\psi]\\
	&	\quad-\int\int_{D_\delta(\tau_1,\tau_2)}\:^{(N)}\pi_{\textit{non-lin}}(N^i\psi)-\int\int_{D_\delta(\tau_1,\tau_2)}N^{i+1}\psi\left(\sum_{l=0}^{i-1}N^l\left[\Box_{g(\nabla\psi)},N\right]N^{i-1-l}\psi+N^i\left(a^{\alpha\beta}\partial_\alpha\psi\partial_\beta\psi\right)\right), 
	\end{aligned}
	\end{equation}
	for all $1\leq i\leq k$, where for $q$ see Lemma~\ref{lem: redshift, lem 1}. Now, we estimate the second term on the right hand side of~\eqref{eq: proof prop: well posedness quasilinear, eq 9}, and obtain 
	\begin{equation}\label{eq: proof prop: well posedness quasilinear, eq 10}
	\begin{aligned}
	&		\left(\int_{\{\bar{t}=\tau_2\}\cap D_\delta(\tau_1,\tau_2)}+\int_{\mathcal{H}^+_{\delta}\cap D_\delta(\tau_1,\tau_2)}+\int_{\bar{\mathcal{H}}^+_\delta\cap D_\delta(\tau_1,\tau_2)}\right)\mathbb{T}(g(\nabla\psi))(N,n)[N^{i}\psi]\\
	&    +\int\int_{D_\delta(\tau_1,\tau_2)\cap \left(\{r_+-q\leq r\leq r_++q\}\cup\{\bar{r}_+-q\leq r\leq\bar{r}_++q\}\right)} 	 \frac{1}{2}\:^{(N)}\pi_{\mu\nu}(\mathring{g})\mathbb{T}^{\mu\nu}(\mathring{g})[N^{i}\psi]\\
	&   \leq \quad \int_{\{\bar{t}=\tau_1\}\cap D_\delta(\tau_1,\tau_2)} \mathbb{T}(g(\nabla\psi))(N,n)[N^{i}\psi]\\
	&	\quad +C\tau_{\textit{max}}\sup_{\tau\in[\tau_1,\tau_2]}\int_{\{\bar{t}=\tau\}\cap\{r_++q\leq r\leq\bar{r}_+-q\}}\sum_{1\leq i_1+i_2+i_3= i+1}\sum_\alpha\left(\partial_{\bar{t}}^{i_1}\partial_r^{i_2}\Omega_\alpha^{i_3}\psi\right)^2\\
	&	\quad-\int\int_{D_\delta(\tau_1,\tau_2)}\:^{(N)}\pi_{\textit{non-lin}}(N^i\psi)-\int\int_{D_\delta(\tau_1,\tau_2)}N^{i+1}\psi\left(\sum_{l=0}^{i-1}N^l\left[\Box_{g(\nabla\psi)},N\right]N^{i-1-l}\psi+N^i\left(a^{\alpha\beta}\partial_\alpha\psi\partial_\beta\psi\right)\right),
	\end{aligned}
	\end{equation}
	for all $1\leq i\leq k$, where the constant $C>0$ depends only on the black hole parameters. 
	
	Now, in view of Lemma~\ref{lem: metric close to SdS} and the elliptic estimates of Lemma~\ref{lem: subsec: elliptic estimates, lem 1}, there exist $\delta,\epsilon(\delta,C_{\text{b}})$, sufficiently small, such that the estimate~\eqref{eq: proof prop: well posedness quasilinear, eq 10} implies 
	\begin{equation}\label{eq: proof prop: well posedness quasilinear, eq 11.1}
	\begin{aligned}
	&		cE_j(\tau_2)+c\int_{\tau_1}^{\tau_2}d\tau E_j(\tau)\\
	&	 +c\left(\int_{\mathcal{H}^+_{\delta}\cap D_\delta(\tau_1,\tau_2)}+\int_{\bar{\mathcal{H}}_\delta^+\cap D_\delta(\tau_1,\tau_2)}\right)\sum_{1\leq i\leq j-1}\mathbb{T}(N,n)[N^i\psi]d\mathring{g}_{\mathcal{H}_\delta}\\
	&   \leq C E_j(\tau_1) +C\tau_{\textit{max}}\sup_{\tau\in[\tau_1,\tau_2]}E_{j,q}(\tau)\\
	&	\quad-\int\int_{D_\delta(\tau_1,\tau_2)}\sum_{1\leq i\leq j-1}\:^{(N)}\pi_{\textit{non-lin}}(N^i\psi)\\
	&	\quad -\int\int_{D_\delta(\tau_1,\tau_2)}\sum_{1\leq i\leq j-1}N^{i+1}\psi\left(\sum_{l=0}^{i-1}N^l\left[\Box_{g(\nabla\psi)},N\right]N^{i-1-l}\psi+N^i\left(a^{\alpha\beta}\partial_\alpha\psi\partial_\beta\psi\right)\right)\\
	&   = C E_j(\tau_1) +C\tau_{\textit{max}}\sup_{\tau\in[\tau_1,\tau_2]}E_{j,q}(\tau)\\
	&	\quad-\int\int_{D_\delta(\tau_1,\tau_2)}\sum_{1\leq i\leq j-1}\:^{(N)}\pi_{\textit{non-lin}}(N^i\psi)\\
	&	\quad -\int\int_{D_\delta(\tau_1,\tau_2)}\sum_{1\leq i\leq j-1}N^{i+1}\psi\left(\sum_{l=0}^{i-1}N^l\left[\Box_{\mathring{g}},N\right]N^{i-1-l}\psi\right)\\
	&	\quad -\int\int_{D_\delta(\tau_1,\tau_2)}\sum_{1\leq i\leq j-1}N^{i+1}\psi\left(\sum_{l=0}^{i-1}N^l\left[h^{ab}\partial_a\partial_b +S^c\:_{ab}\partial_c,N\right]N^{i-1-l}\psi+N^i\left(a^{\alpha\beta}\partial_\alpha\psi\partial_\beta\psi\right)\right),
	\end{aligned}
	\end{equation}
	for all $1\leq j\leq k+1$, for some constants $c(k,M,\Lambda,A_{[k+1]},B_{[k+1]}),C(k,M,\Lambda,A_{[k+1]},B_{[k+1]})$, where for $S^c\:_{ab}$ see~\eqref{eq: energy estimate for quasilinear, 2, 2, the nonlinear part of the christophel symbols}.

	Similarly to~\eqref{eq: proof prop: well posedness quasilinear, eq 7.01}, we estimate the two last (non-linear) terms of~\eqref{eq: proof prop: well posedness quasilinear, eq 11.1}. Note that we estimate the third to last bulk term of~\eqref{eq: proof prop: well posedness quasilinear, eq 11.1}, in view on the linear $N$-commutation Lemma~\ref{lem: redshift, lem 2}. Specifically, in view of the bootstrap~\eqref{eq: proof prop: well posedness quasilinear, eq 3}, and the definitions of the smooth tensors $a,\: h$, see Sections \ref{subsec: semilinear derivatives}, \ref{subsec: assumption} respectively, there exists a $\delta>0$ sufficiently small, and an $\epsilon(\delta,C_{\text{b}})>0$ sufficiently small, such that for the two non-linear terms we use the Lemma~\ref{lem: metric close to SdS} and the coarea formula~\eqref{subsec: coarea}, to obtain
	\begin{equation}\label{eq: proof prop: well posedness quasilinear, eq 11.2}
	\begin{aligned}
	&		E_{k+1}(\tau_2)+\int_{\tau_1}^{\tau_2}d\tau E_{k+1}(\tau)+\left(\int_{\mathcal{H}^+_{\delta}\cap D_\delta(\tau_1,\tau_2)}+\int_{\bar{\mathcal{H}}_\delta^+\cap D_\delta(\tau_1,\tau_2)}\right)\sum_{1\leq i\leq k}\mathbb{T}(N,n)[N^i\psi]d\mathring{g}_{\mathcal{H}_\delta}\\
	&   \leq \quad C E_{k+1}(\tau_1) +C\tau_{\textit{max}}\sup_{\tau\in[\tau_1,\tau_2]}E_{j,q}(\tau)\\
	\end{aligned}
	\end{equation}
	for $k\geq 7$, where $ 0<C=C(k,M,\Lambda,A_{[k+1]},B_{[k+1]})$. We appropriately distributed derivatives and used a Sobolev injection, see Lemma~\ref{lem: sobolev estimate}, in view of that $k= 7$ is the smallest integer such that 
	\begin{equation}
	\floor{\frac{k+1}{2}}+2\leq k-1.
	\end{equation}	
	We obtain

	\textbf{The combination of the $\partial_{\bar{t}}$ and $N$ multiplier estimates.} 
	
	Now, we combine the $\partial_{\bar{t}}$-estimate, namely~\eqref{eq: proof prop: well posedness quasilinear, eq 7.3}, and the $N$-estimate, namely~\eqref{eq: proof prop: well posedness quasilinear, eq 11.2}. 
	
	We obtain 
	\begin{equation}\label{eq: proof prop: well posedness quasilinear, eq 11.3}
	\begin{aligned}
	&		E_{k+1}(\tau_2)+\int_{\tau_1}^{\tau_2}d\tau E_{k+1}(\tau)+\left(\int_{\mathcal{H}^+_{\delta}\cap D_\delta(\tau_1,\tau_2)}+\int_{\bar{\mathcal{H}}_\delta^+\cap D_\delta(\tau_1,\tau_2)}\right)\sum_{1\leq i\leq k}\mathbb{T}(N,n)[N^i\psi]d\mathring{g}_{\mathcal{H}_\delta}\\
	&   \leq \quad C E_{k+1}(\tau_1) +C\tau_{\textit{max}}\Bigg(C E_{k+1}(\tau_1)+C\delta E_{k+1}(\tau_2)\\
	&\quad\quad\quad\quad\quad\quad\quad\quad\quad\quad\quad +	C\sqrt{\delta}\left(\int_{\mathcal{H}^+_{\delta}\cap D_\delta(\tau_1,\tau_2)}+\int_{\bar{\mathcal{H}}_\delta^+\cap D_\delta(\tau_1,\tau_2)}\right)\sum_{1\leq i\leq k}\mathbb{T}(N,n)[N^i\psi]  d\mathring{g}_{\mathcal{H}_\delta}\\
	&\quad\quad\quad\quad\quad\quad\quad\quad\quad\quad\quad +\sqrt{\epsilon(\delta)} C \int_{\tau_1}^{\tau_2}d\tau E_{k+1}(\tau)\Bigg),
	\end{aligned}
	\end{equation}
	for $k\geq 7$.
	
	\textbf{The integral inequality.}
	Therefore, for $\delta(\tau_{\textit{max}}),\epsilon(\tau_{\textit{max}},\delta)>0$ sufficiently small, then from inequality~\eqref{eq: proof prop: well posedness quasilinear, eq 11.3} we obtain 
	\begin{equation}\label{eq: proof prop: well posedness quasilinear, eq 11.4}
	\begin{aligned}
	&		E_{k+1}(\tau_2)+\int_{\tau_1}^{\tau_2}d\tau E_{k+1}(\tau)  \leq C E_{k+1}(\tau_1) +C\left(\tau_2-\tau_1\right) E_{k+1}(\tau_1),\\
	\end{aligned}
	\end{equation}
	for $k\geq 7$, and for a constant $C(k,M,\Lambda,A_{[k+1]},B_{[k+1]})>0$. Note that~\eqref{eq: proof prop: well posedness quasilinear, eq 11.4} is the quasilinear analogue of the integral inequality~\eqref{eq: proof prop: well posedness quasilinear, eq 2.1}. 
	
	\textbf{Finishing the proof, improving the bootstrap~\eqref{eq: proof prop: well posedness quasilinear, eq 3}.}

	Note that the estimate~\eqref{eq: proof prop: well posedness quasilinear, eq 11.4} holds for any $\tau_1^\prime\geq \tau_1$ in the place of $\tau_1$. 
	
	Therefore, arguing as in the linear case, see the integral inequality~\eqref{eq: proof prop: well posedness quasilinear, eq 2.1}, we obtain from~\eqref{eq: proof prop: well posedness quasilinear, eq 11.4} that there exists a constant
	\begin{equation}
	C_{\textit{wp}}(k,M,\Lambda,A_{[k+1]},B_{[k+1]})>1
	\end{equation}
	independent of $\tau_{\textit{max}}$, such that  
	\begin{equation}\label{eq: proof prop: well posedness quasilinear, eq 12}
	\begin{aligned}
	E_{k+1}[\psi](\tau^\prime)	&	\leq C_{\textit{wp}} E_{k+1}[\psi](\tau_1),\\
	\end{aligned}
	\end{equation}
	for all $\tau^\prime\in [\tau_1,\tau_1+\tau_{\textit{max}}]$. Moreover, by using~\eqref{eq: proof prop: well posedness quasilinear, eq 12} and the smallness $E_{k+1}[\psi](\tau_1)\leq \epsilon$, we also prove that 
	\begin{equation}
		E_{k+2}[\psi](\tau^\prime)	\leq C_{\textit{wp}} E_{k+2}[\psi](\tau_1),
	\end{equation}
	by repeating the arguments of the above proof, and after redefining $C_{\textit{wp}}(k,M,\Lambda,A_{[k+1]},B_{[k+1]})>1$ appropriately, provided that $E_{k+2}[\psi](\tau_1)<\infty$.

	Finally, for a sufficiently large $C_{\text{b}}(k,M,\Lambda,A_{[k+1]},B_{[k+1]})\gg C_{\textit{wp}}$ and a sufficiently small $\epsilon>0$, we improve the bootstrap assumption~\eqref{eq: proof prop: well posedness quasilinear, eq 3} by Sobolev inequalities, see Lemma~\ref{lem: sobolev estimate}, on the left hand side of inequality~\eqref{eq: proof prop: well posedness quasilinear, eq 12}.

	We conclude the energy estimates~\eqref{eq: prop: well posedness quasilinear, eq 1},~\eqref{eq: prop: well posedness quasilinear, eq 2} and therefore the Proposition. 
\end{proof}

Note also the following Remark 
\begin{remark}\label{rem: appendix: local well posedness, rem 1}
A posteriori, one can remove the dependence of $\delta$ on the $\tau_{\textit{step}}$ parameter. To do so, note that we dropped the terms on the hypersurfaces $\mathcal{H}^+_\delta,\bar{\mathcal{H}}^+_\delta$ from~\eqref{eq: proof prop: well posedness quasilinear, eq 11.3}. By using the $N$ redshift vector field in the spacetime regions
\begin{equation}
	\{r_+-\tilde{\delta}\leq r\leq r_+-\delta\},\qquad  \{\bar{r}_++\delta\leq r\leq\bar{r}_++\tilde{\delta}\}
\end{equation}
for $\delta>0$ as in Theorem~\ref{thm: quasilinear, 1, local bootstrap}, and for a sufficiently small $\tilde{\delta}$ independent of $\tau_{\textit{step}}$, we can absorb the contributions at $\mathcal{H}^+_\delta,\bar{\mathcal{H}}^+_\delta$ by~\eqref{eq: proof prop: well posedness quasilinear, eq 11.3}, and conclude the Cauchy stability result of Proposition \ref{prop: well posedness quasilinear} for a $\tilde{\delta}>0$ in the place of $\delta$.
\end{remark}

%%%%%%%%%%%%%%%%%%%%%%%%%%%%%%%%%%%%%%%%%%%%%%%%%%%%%%%%%%%%%
\bibliographystyle{plain}
\bibliography{MyBibliography}

\end{document}